\documentclass[conference]{IEEEtran}

\usepackage[pdftex]{graphicx}
\usepackage{amsmath}
\usepackage{url}

\usepackage{amssymb}
\usepackage{amsfonts}
\usepackage{booktabs}
\usepackage{multirow}
\usepackage{amsthm}
\usepackage[table]{xcolor}
\usepackage{subcaption}
\usepackage[hidelinks]{hyperref}
\usepackage{circledsteps}

\usepackage{url}
\makeatletter
\g@addto@macro{\UrlBreaks}{\UrlOrds}
\makeatother
\usepackage{ulem}

\begin{document}
\bstctlcite{BSTcontrol}

\title{FlyTrap: Physical Distance-Pulling Attack\\ Towards Camera-based Autonomous Target Tracking Systems}

\author{\IEEEauthorblockN{Shaoyuan Xie \quad Mohamad Habib Fakih \quad Junchi Lu \quad Fayzah Alshammari \quad Ningfei Wang \\ Takami Sato \quad Halima Bouzidi \quad Mohammad Abdullah Al Faruque \quad Qi Alfred Chen}
	\IEEEauthorblockA{University of California, Irvine}}

\IEEEoverridecommandlockouts
\makeatletter\def\@IEEEpubidpullup{6.5\baselineskip}\makeatother
\IEEEpubid{\parbox{\columnwidth}{
        \vspace{-2.2cm}
        Contact email: \texttt{shaoyux@uci.edu} \\[0.5em]
        $^*$This is an \textbf{extended version} of the paper accepted at:\\
		Network and Distributed System Security (NDSS) Symposium 2026\\
		23-27 February 2026, San Diego, CA, USA\\
		ISBN 979-8-9919276-8-0\\  
		https://dx.doi.org/10.14722/ndss.2026.230904\\
		www.ndss-symposium.org
}
\hspace{\columnsep}\makebox[\columnwidth]{}}

\maketitle

\begin{abstract}
Autonomous Target Tracking (ATT) systems, especially ATT drones, are widely used in applications such as surveillance, border control, and law enforcement, while also being misused in stalking and destructive actions. Thus, the security of ATT is highly critical for real-world applications. Under the scope, we present a new type of attack: \textit{distance-pulling attacks} (DPA) and a systematic study of it, which exploits vulnerabilities in ATT systems to dangerously reduce tracking distances, leading to drone capturing, increased susceptibility to sensor attacks, or even physical collisions. To achieve these goals, we present \textit{FlyTrap}, a novel physical-world attack framework that employs an adversarial umbrella as a deployable and domain-specific attack vector. FlyTrap is specifically designed to meet key desired objectives in attacking ATT drones: physical deployability, closed-loop effectiveness, and spatial-temporal consistency. Through novel progressive distance-pulling strategy and controllable spatial-temporal consistency designs, FlyTrap manipulates ATT drones in real-world setups to achieve significant system-level impacts. Our evaluations include new datasets, metrics, and closed-loop experiments on real-world white-box and even commercial ATT drones, including DJI and HoverAir. Results demonstrate FlyTrap's ability to reduce tracking distances within the range to be captured, sensor attacked, or even directly crashed, highlighting urgent security risks and practical implications for the safe deployment of ATT systems. Video demonstrations and code can be found at \url{https://sites.google.com/view/av-ioat-sec/flytrap}.
\end{abstract}

\IEEEpeerreviewmaketitle

\section{Introduction}
\label{sec:intro}

Autonomous Target Tracking (ATT), also known as \textit{Active Track}~\cite{activetrack} or \textit{Dynamic Track}~\cite{dynamictrack}, allows autonomous systems to follow a designated target (e.g., a person) while maintaining a consistent distance~\cite{cheng2017autonomous, chuang2019autonomous, pestana2013vision, pestana2014computer}. Drones~\cite{dji2024, djimavic2023, skydio2023, autel2020} have become the leading platform for ATT due to their versatility, supporting applications such as security surveillance~\cite{forbes_skydio, illinois_drone}, border control~\cite{us-cbp_drone}, and law enforcement~\cite{new_auto} beyond entertainment. Some of these applications are already deployed in practice, e.g., U.S. police departments are using ATT drones to track individuals for law enforcement~\cite{new_auto}. Conversely, this technology poses significant security, privacy, and safety threats if misused in criminal scenarios, e.g., to facilitate stalking~\cite{canonsburg_drone_stalking}, or lethal/destructive actions by carrying explosives or weapons~\cite{forbes_ukraine, reuters_ukraine, reuters_drones_2024}.

All these real-world applications, whether benign or criminal, make the security of ATT systems critically important. In this work, we exploit new vulnerabilities in ATT systems by causing drones to dangerously decrease their tracking distance from targets, which we define as \textit{distance-pulling attack} (DPA). The DPA can lead to various severe physical consequences by causing the drone to be: (1) physically captured after being pulled into a reachable range (e.g., by a net gun~\cite{garcia2020autonomous, chen2022anti, net_capture}); (2) made much more attackable by a wider band of sensor attacks (e.g., camera spoofing~\cite{zhou2022doublestar}, acoustic attacks~\cite{son2015rocking, ji2021poltergeist}), which by nature have range limitations~\cite{zhou2022doublestar, son2015rocking}; and (3) physically crashed, after the distance between the drone and the tracking target is shortened close enough to be within physically hitting distance. In contrast to other attacks on object tracking that can be possibly applied to ATT, such as those that can cause the model to lose track of the target~\cite{muller2022physical, wiyatno2019physical}, our proposed DPA can enable a more fundamental elimination of the drone since the attacker can physically capture it, reverse-engineer it~\cite{schiller2023drone}, and/or identify the underlying pilot as law enforcement evidence collection~\cite{new_jersey_drone_incursions}. Thus, understanding the security challenges and practical implications of DPA against ATT systems, especially those that are already commercially available, is imperative.

\begin{figure}[t!]
    \centering
    \includegraphics[width=\linewidth]{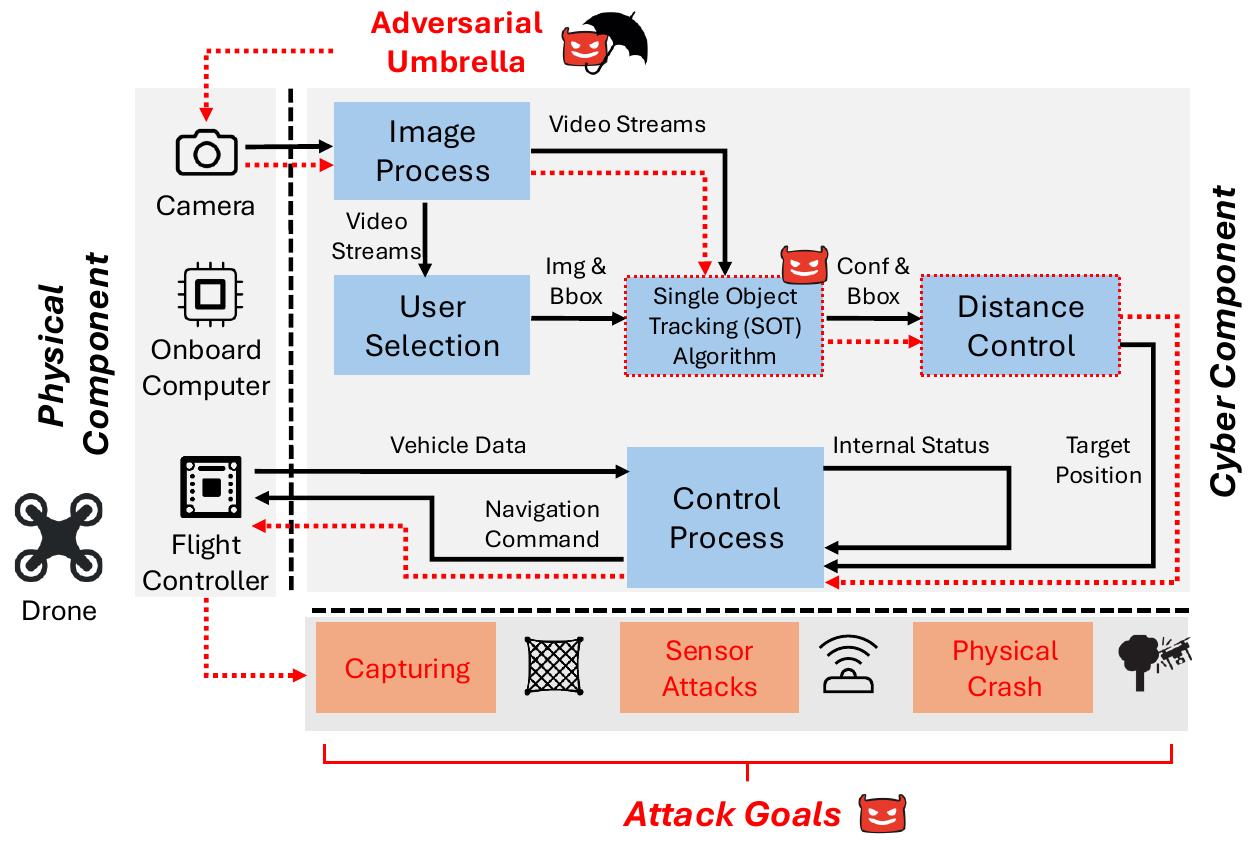}
    \caption{Overview of the Autonomous Target Tracking (ATT) system data flow and our proposed distance-pulling attack (DPA) propagation path. We treat the camera as a physical entry point and use the adversarial umbrella to attack the Single Object Tracking (SOT) model and then the distance control algorithm to cause system-level distance-pulling effects, achieving attack goals including drone capturing, range-limited sensor attacks, or direct crashing.}
    \vspace{-0.3cm}
    \label{fig:overview}
\end{figure}

Most modern ATT systems rely on cameras, given the cost-efficiency and ease of deployment on drone platforms~\cite{dji2024, djimavic2023, skydio2023, autel2020}. Specifically, Deep Neural Network (DNN) based Single Object Tracking (SOT)~\cite{bertinetto2016fully, li2018high, cui2022mixformer, li2019siamrpn++} is a core step in the latest camera-based ATT systems to achieve stable target tracking as shown in Fig.~\ref{fig:overview}. Prior works have demonstrated vulnerabilities in SOT models through pixel-level perturbations~\cite{fu2022ad, yan2020cooling, chen2020one} or physical attacks~\cite{ding2021towards, wiyatno2019physical, muller2022physical}. However, these studies primarily focus on manipulating tracking (e.g., move-in and move-out attacks~\cite{muller2022physical}) while the ATT system is composed of both tracking and distance control. Therefore, these prior works do not address the core challenges we identify for DPA against ATT systems below.

Specifically, first, practical entry points for real-world attacks on ATT systems remain under-explored. Previous tracking attacks that focused on digital perturbations~\cite{fu2022ad, yan2020cooling, chen2020one, yan2020hijacking} often lack physical feasibility. Additionally, physical attack vectors such as TV screens~\cite{wiyatno2019physical}, printed letter-size papers~\cite{ding2021towards}, or projectors~\cite{muller2022physical} face significant challenges due to their limited deployability in uncontrolled outdoor environments. Second, prior works fail to experimentally demonstrate the generalizability of their attacks across scenarios, where a single attack pattern is effective against unseen targets and/or backgrounds. Third, the newly proposed DPA against the ATT systems inherently demands closed-loop effectiveness, where current attack results influence future frames. The systems targeted by prior works~\cite{wiyatno2019physical, yan2020cooling, ding2021towards, muller2022physical} do not involve distance control and thus do not consider addressing this newly raised challenge in ATT systems. Lastly, these works~\cite{yan2020cooling, muller2022physical} ignore the spatial-temporal consistency and can be defended by existing consistency checking-based defense methods~\cite{muller2024vogues}. 

To address these critical research challenges, we present the first systematic study on the security of camera-based ATT under a newly defined physical-world DPA. Our approach centers on three key design objectives to ensure the success and stealthiness of the attack: (i) \textit{physical and real-world deployability}, where the attack vector physically misguides the ATT system’s distance control mechanism while remaining easy to deploy, robust to lighting conditions~\cite{muller2022physical}, and remains inconspicuous; (ii) \textit{closed-loop effectiveness}, where the attacks progressively shorten the tracking distance in closed-loop fashion, achieving system-level physical impacts; and (iii) \textit{spatial-temporal consistency}, which allows the DPA to be consistent spatially and temporally, evading latest anomaly detection-based defenses~\cite{muller2024vogues, man2023person, han2024visionguard, yu2024physense}.

To achieve the above objectives, we introduce \textit{FlyTrap}, a novel physical DPA against ATT systems. FlyTrap is the first to systematically tackle these challenges by utilizing \textit{adversarial umbrella} as a novel domain-specific attack vector, i.e., a physical attack vector that an ATT-tracked target can naturally and dynamically deploy for self-coverage. The umbrella, designed for ease of carriage and inconspicuous deployment, can be naturally oriented upward toward the drone. In the ATT system context, using such an attack vector can simultaneously offer advantages in physical realizability and real-world deployability as desired above. Additionally, we design a novel progressive distance-pulling strategy, enabling continuous distance-pulling under closed-loop control. We further design our attack to maintain spatial-temporal consistency, which can bypass current state-of-the-art consistency cross-checking-based defense mechanisms~\cite{muller2024vogues, man2023person, han2024visionguard, yu2024physense}. Our approach combines novel attack vectors, progressive distance-pulling, and a controllable design for spatial-temporal consistency to achieve physical, real-world deployable, closed-loop, effective, and spatial-temporal consistent attacks.

In evaluation, we construct new datasets and introduce metrics to evaluate the system-level impact of the proposed DPA, which shows high effectiveness, scenario universality, and spatial-temporal consistency. In physical experiments, we craft real-world adversarial umbrella prototypes optimized on different white-box models. Then, we implement a full-stack ATT drone from scratch. The experimental setups provide a closed-loop evaluation to understand the physical impact of our FlyTrap design under the white-box assumption. Our white-box, closed-loop physical experiments show that FlyTrap can achieve 100\% success rate in pulling the drone close enough to induce capturing, sensor attacks, and/or direct crashes. To further assess the real-world impacts, we conduct black-box DPA against three commercial drones: the DJI Mini 4 Pro, the DJI NEO, and HoverAir X1. The results show that our newly proposed DPA can indeed cause system-level DPA attack effects on them. We further show end-to-end FlyTrap-enabled DPA demonstrations against these commercial drones, leading to their capture or crash, demonstrating DPA's strong applicability in real-world attack scenarios. We also investigate the stealthiness of FlyTrap-optimized patterns by conducting a user study with 200 participants, and further discuss potential countermeasures. Video demonstrations and code can be found on our project website at \textbf{\url{https://sites.google.com/view/av-ioat-sec/flytrap}}.  To summarize, our contributions include:
\begin{itemize}
\setlength{\itemsep}{0pt}
\setlength{\parskip}{0pt}
    \item \textbf{Problem formulation}: We are the first to define distance-pulling attacks (DPA) of camera-based ATT drones. We formally define the problem with domain-specific objectives and introduce the adversarial umbrella as a novel, physically deployable attack vector.
    
    \item \textbf{Novel design}: We propose FlyTrap, including a progressive distance-pulling strategy and a controllable spatial-temporal consistency design, encompassed by an end-to-end optimization pipeline for attacking ATT drones.
    
    \item \textbf{Evaluation}: We construct a new dataset and define system-level metrics for comprehensive evaluation. The results highlight our design to achieve closed-loop and spatial-temporal consistent attacks.
    
    \item \textbf{Physical-world impact}: We implement full-stack ATT drones, craft physical adversarial umbrellas, and conduct end-to-end evaluations in real-world setups, showing direct system-level impact. We further performed extensive black-box testing on three commercial drones, showing high real-world applicability of the proposed attacks.
\end{itemize}
  
\section{Background and Problem Formulation}

\subsection{Camera-based Autonomous Target Tracking (ATT) Drone}
\label{sec:atts}

Fig.~\ref{fig:overview} illustrates a typical camera-based ATT system~\cite{zhang2019eye, 7502612, kendall2014on, fradi2018autonomous}, which follows a hierarchical control architecture consisting of an inner and an outer loop~\cite{7502612}. The inner loop, integrated into the flight controller, handles low-level flight stability and receives navigation commands. The outer loop manages high-level perception and decision-making tasks, including image processing, object tracking, distance estimation, and flight path planning.

\noindent{\textbf{Single Object Tracking (SOT) algorithm.}}
\label{sec:track-alg}
The SOT algorithm is a crucial component in the ATT pipeline, primarily responsible for generating navigation commands. Contemporary SOT algorithms are predominantly based on Deep Neural Networks (DNN)~\cite{bertinetto2016fully, li2018high, cui2022mixformer, li2019siamrpn++}. The SOT model uses a \textit{template frame} as a reference and predicts the target’s location in \textit{search frames}, as illustrated in Fig.~\ref{fig:sot}. The target tracking task can be formulated as a conditional prediction, as shown in the equation below:
\begin{equation}
\label{eq:track}
    \left\{ (cx_j, cy_j, w_j, h_j, score_j) \right\}_{j=1}^{M} = F(\mathbf{I}_\text{search} | \mathbf{I}_\text{tplt}),
\end{equation}
where $F$ denotes the SOT model, $\mathbf{I}_\text{search}$ and $\mathbf{I}_\text{tplt}$ denote the search frame and template frame, respectively. $\mathcal{P}_j=(cx_j, cy_j, w_j, h_j)$ denotes the localization results, including the x- and y-axis center coordinates, width, and height. $score_j$ denotes the prediction confidence. $M$ represents the number of prediction proposals. The proposal with the highest confidence is regarded as the final tracking output.

\noindent{\textbf{Distance Control.}}
\label{sec:pos-es}
This component estimates the relative distance to the tracked object using the SOT output and translates it into flight control actions (e.g., next waypoint). The most widely adopted strategy in real-world systems today is 2D-based distance control~\cite{7502612, kendall2014on, fradi2018autonomous}, which infers navigation commands directly from the 2D object bounding box. More specifically, in Fig.~\ref{fig:sot}, the drone adjusts its yaw, roll, and/or altitude to center the bounding box within the current frame and moves forward or backward to maintain the bounding box size, thereby preserving a stable tracking distance. This system design motivates our formulation of DPA as a fundamental system-level attack objective for ATT. Specifically, we strategically shrink the tracking bounding box to deceive the drone into perceiving that the object is moving away, thus moving closer for compensation, leading to reduced tracking distance.

\begin{figure}[t!]
    \centering
    \includegraphics[width=\linewidth]{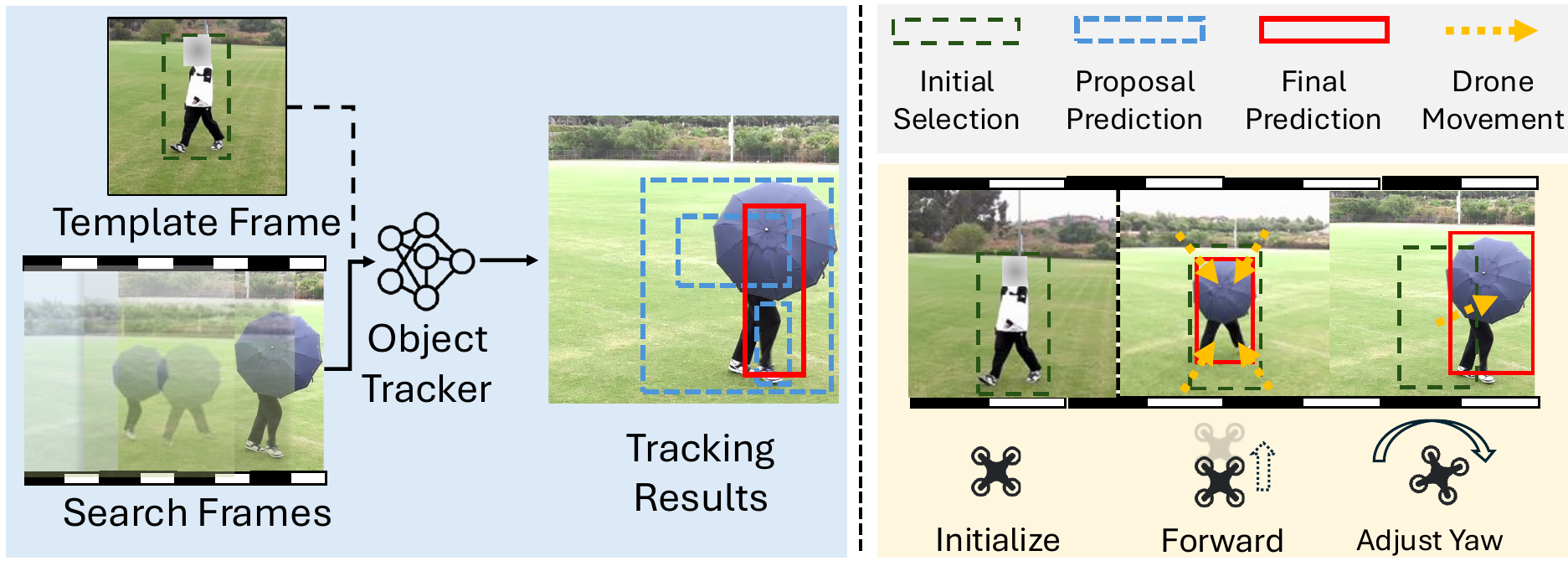}
    \caption{\textit{Left}: Single Object Tracking (SOT) depends on the initialization as the template frame and tracks the target in search frames. \textit{Right}: The drone adjusts its position to keep the box at the center and the same size as the template frame.}
    \label{fig:sot}
\end{figure}

\subsection{Problem Formulation}
\label{sec:problem}

While disrupting the SOT component to lose track of the target can temporarily disable the ATT functionality~\cite{yan2020cooling, muller2022physical, wiyatno2019physical}, it does not fundamentally prevent the system from resuming tracking, either through manual re-selection or operator intervention. As shown in the ATT system pipeline (Section~\ref{sec:atts}), the ATT system operates based on both SOT and distance control for maintaining a stable tracking distance. From this system perspective, we focus on exploiting vulnerabilities in the position control mechanism. A particularly compelling attack objective—and the focus of this work—is to intentionally reduce the tracking distance, pulling the ATT drone dangerously close to the tracked target, which we define as the distance-pulling attack (DPA). As shown in Fig.~\ref{fig:attak-goal}, DPA can be exploited to achieve \textit{A1}: drone capturing, e.g., by using a net gun~\cite{net_capture} (shown in Section~\ref{sec:commercial}); \textit{A2}: range-limited sensor attacks~\cite{son2015rocking, zhou2022doublestar}; or \textit{A3}: causing the drone to crash into the target (also shown in Section~\ref{sec:commercial}). In either case, this can result in a more permanent elimination of tracking capabilities, compared to losing tracking~\cite{yan2020cooling, muller2022physical}.

Considering both the benign and criminally motivated applications of ATT (Section~\ref{sec:intro}), the incentives for this overall attack goal can also be either benign or malicious. For example, when used against benign applications (e.g., security surveillance~\cite{forbes_skydio}, border control~\cite{us-cbp_drone}, and law enforcement~\cite{new_auto}), the attack incentives are malicious and can directly threaten public security by capturing the drone and exploiting vulnerabilities for future counter-measures. However, when used for criminally-motivated scenarios (e.g., stalking~\cite{canonsburg_drone_stalking} or lethal actions~\cite{forbes_ukraine, reuters_ukraine, reuters_drones_2024}), the attack incentives may be benign, empowering individuals to defend themselves, e.g., by capturing unauthorized drones, identifying the pilot, and extracting flight logs to uncover malicious intent~\cite{logan_airport_drone_incident, drone_forensics_paraben}. Thus, although we generally call it an ``attack'' in this paper, the security problem studied can be exploited for social good, and the ``attacker'' may be non-malicious individuals who just want to protect their privacy and safety.

\begin{figure}
    \centering
    \includegraphics[width=0.8\linewidth]{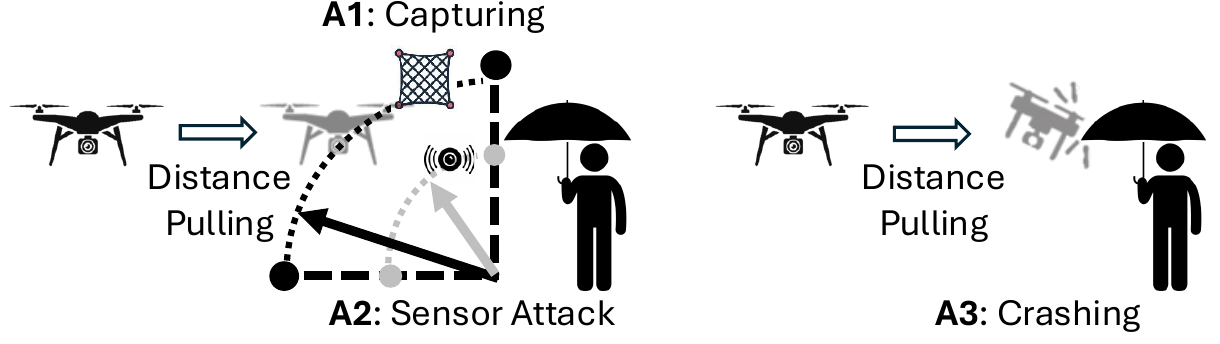}
    \caption{Illustration of distance-pulling attack (DPA) and attack goals targeted in this work. We target to dangerously shorten the tracking distance of ATT drones, which can be exploited to cause the drone to be \textit{A1}: captured; \textit{A2}: under range-limited sensor attacks; or even \textit{A3}: crashed into the attack umbrella.}
    \vspace{-0.3cm}
    \label{fig:attak-goal}
\end{figure}

\subsection{Threat Model}
\label{sec:threat-model}

In this paper, we mainly target ATT drone setups that perform tracking within 20 meters, which is the most typical ATT operation range for person tracking for consumer drones today (e.g., DJI Mini~\cite{dji_active_tracking_2023}, Potensic~\cite{potensic_atom_tracking_2023}, Autel~\cite{autel_dynamic_track_2024}, Skydio 2~\cite{skydio_update_2020}) and is also a range that can more easily allow the attacker to notice the tracking and thus launch the attack. Note that our attack is not limited to this range by design; the attack distance and angle practicality are further discussed in Appendix~\ref{app:distance-shrink-discussion}. We start with a white-box attack design setup, i.e., we assume that the attacker has full knowledge of the SOT model used in the victim ATT system. This can be accomplished by first collecting information about the targeted drones with the ATT feature~\cite{dji2024, djimavic2023, skydio2023, autel2020} and then purchasing the same model and reverse engineering it, which is feasible given recent advances in reverse-engineering such systems~\cite{schiller2023drone} and machine learning models~\cite{wu2022dnd, liu2023decompiling}. We also assume that the attacker can collect videos of different tracking scenarios, but these videos are not necessarily for the same tracking scenario during the attack (i.e., for the same tracking target instance and/or background location when the attack is launched), as we show in the scenario universality evaluation (Section~\ref{sec:universality}).

Although our method is developed under a white-box assumption, it can potentially be extended to black-box settings by leveraging the transferability of adversarial patterns~\cite{liu2016delving}. As black-box settings are more practical, we also evaluate them by performing (1) attack transferability evaluation across open-source models (Section~\ref{sec:transfer}), and (2) direct black-box testing on commercial ATT drones (Section~\ref{sec:commercial}).
\section{Related Works and Design Challenges}
\label{sec:compare}

\subsection{Related Works and Comparisons}
\label{sec:related-work}
\noindent
\textbf{Autonomous systems security.} 
Security research on autonomous systems primarily falls into two categories: sensor security and AI security. For sensor security, prior work has examined threats to commonly used sensors in autonomous systems, including cameras~\cite{ji2021poltergeist, cao2021invisible, yan2016can}, LiDAR~\cite{sato2023lidar, petit2015remote, cao2019adversarial, cao2021invisible}, gyroscopic~\cite{son2015rocking}, IMU~\cite{tu2018injected}, etc. In contrast, autonomous AI security research has primarily targeted self-driving vehicles, such as traffic sign recognition systems~\cite{zhao2019seeing, jia2022fooling, wang2023does, sato2023intriguing}, automatic lane centering~\cite{sato2021dirty, jiao2021end}, high-autonomy autonomous driving systems~\cite{cao2021invisible, sato2023lidar, wan2022too, zhang2022adversarial}, etc. Recently, Zhou et al. were among the first to investigate autonomous AI security in drone contexts, with a focus on stereo camera-based collision avoidance~\cite{zhou2022doublestar}. To the best of our knowledge, we are the first to propose and conduct a system-level security analysis of DPA in camera-based ATT.

\noindent{\textbf{Adversarial attacks on SOT.}} 
While we are the first to propose DPA in ATT systems, prior work has examined vulnerabilities in SOT models individually~\cite{fu2022ad, yan2020cooling, chen2020one, yan2020hijacking, liang2020efficient, jia2021iou, nakka2022universal, wiyatno2019physical, ding2021towards, chen2021unified, muller2022physical}. Specifically, various prior works explored using pixel perturbation to attack SOT models~\cite{fu2022ad, yan2020cooling, chen2020one, yan2020hijacking, liang2020efficient, jia2021iou, nakka2022universal}. However, these studies focus on offline video processing rather than real-time ATT systems, and therefore do not address: (1) physical-world deployment, (2) effective attack across closed-loop ATT control, and (3) spatial-temporal consistency.

Some more recent prior works have started to consider physical-world attack vectors~\cite{wiyatno2019physical, ding2021towards, chen2021unified, muller2022physical}. However, their attack vectors: TV screen~\cite{wiyatno2019physical}, printed paper~\cite{ding2021towards, chen2021unified}, and projectors~\cite{muller2022physical} face serious challenges for practical deployment for the following reasons: printed paper is barely visible from an aerial perspective; TV screen~\cite{wiyatno2019physical} has limitations during the carrying phase; and projectors used in AttrackZone are subject to lighting conditions and require a close enough flat surface for projection, as acknowledged by the authors~\cite{muller2022physical}. Moreover, ATT systems generally operate in well-lit, outdoor environments where attackers have limited control, making projection-based attacks difficult to execute reliably. Moreover, these works were not designed with DPA and closed-loop ATT systems in mind. As a result, these approaches overlook the closed-loop dynamics critical to achieving better distance-pulling effects. Last but not least, the advanced consistency-checking defense can already detect these attacks~\cite{muller2024vogues}, given their insufficient spatial-temporal consistency considerations.

\subsection{Design Challenges}
\label{sec:design-challenge}
Based on the above analysis, we identify key challenges in designing DPA against ATT systems.

\noindent{\textbf{$C_1$: Physical and real-world deployable attack vectors for ATT systems.}} 
Designing effective attacks against ATT systems requires physical, highly deployable vectors. Prior attack vectors, while effective in controlled environments, face significant limitations when deployed in real-world ATT settings~\cite{wiyatno2019physical, muller2022physical, ding2021towards} as discussed in Section~\ref{sec:related-work}. The attackers often have minimal control over environmental factors, especially when they are unwillingly tracked outdoors. This highlights the need for a more versatile, inconspicuous, and deployable physical attack vector.

\noindent{\textbf{$C_2$: Closed-loop distance-pulling effects.}} 
A successful DPA must sustain closed-loop distance-pulling effects. In the context of ATT systems, this means the attack at the current frame will influence the drone's behavior in subsequent frames, reducing the tracking distance in a feedback loop. However, existing works adopt an open-loop\footnote{The open-loop concept in this paper is a similar concept from control theory. By open-loop, we mean the attacker conducts attacks without considering the control feedback loop from the victim systems.} approach, where attacks are optimized independently from the ATT system’s response~\cite{wiyatno2019physical, muller2022physical, ding2021towards}. Such methods fail to address the dynamic and autonomous nature of drone tracking, where attacks must continuously pull the drone closer in response to reduced tracking distances. The motion model in DRP~\cite{sato2021dirty}, designed for lane-centering in ground vehicles, does not generalize to the aerial dynamics of drone tracking.

\noindent{\textbf{$C_3$: Spatial-temporal consistency.}}  
Attacking object tracking introduces additional challenges compared to object detection due to the inherent spatial-temporal consistency of tracking algorithms~\cite{muller2022physical}. Additionally, recent state-of-the-art consistency-based defenses have demonstrated promising performance in securing vision-based autonomous systems~\cite{man2023person, muller2024vogues, han2024visionguard, xu2024physcout, yu2024physense}, further increasing the difficulty of maintaining consistency in attacks. Although prior work has considered the spatial-temporal-based defense~\cite{muller2022physical}, their methods primarily target simplistic approaches, such as Kalman filters, leaving them detectable to more advanced anomaly detection mechanisms~\cite{muller2024vogues}. Addressing this challenge requires developing adversarial attacks under the constraint of maintaining spatial-temporal consistency in both the tracking model and auxiliary consistency-checking mechanisms.
\section{FlyTrap}
\label{sec:methodology}

\begin{figure}[t]
    \centering
    \includegraphics[width=\linewidth]{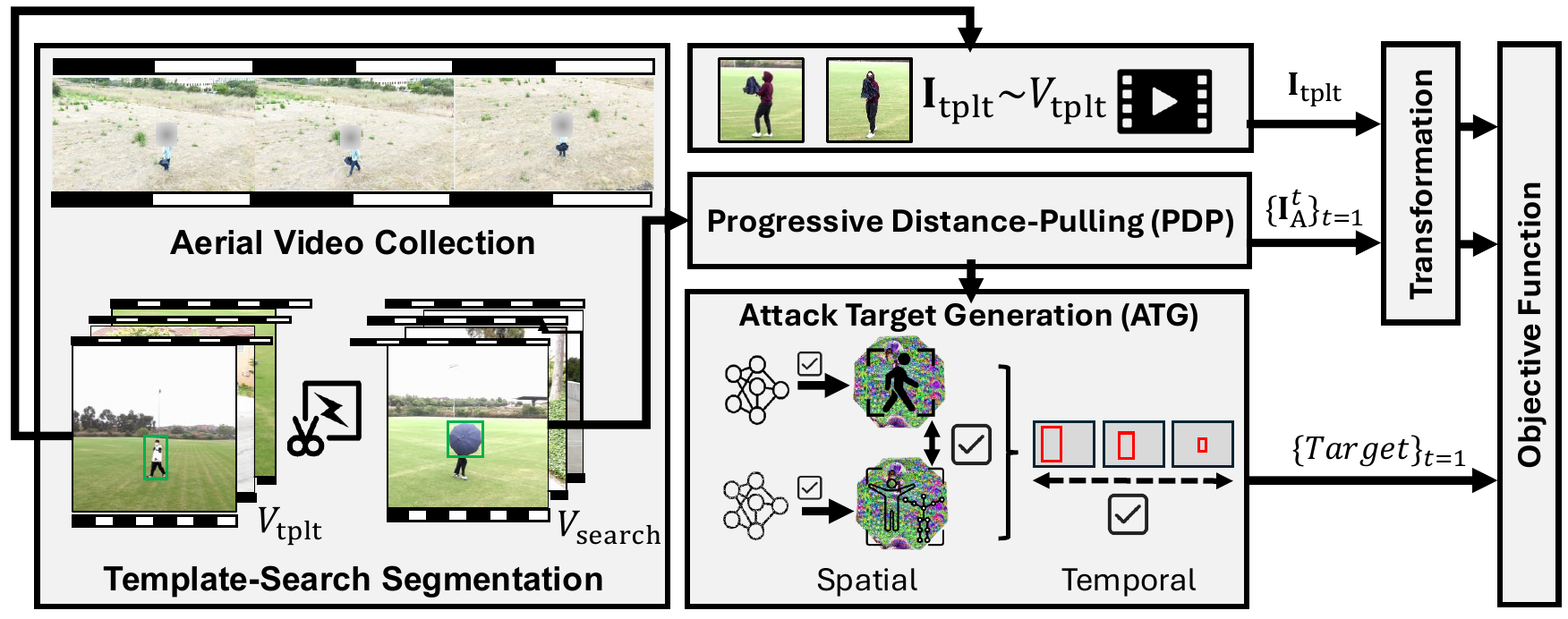}
    \caption{FlyTrap overall pipeline. We design an adversarial umbrella as a domain-specific and deployable attack vector. The progressive distance-pulling (PDP) achieves the closed-loop distance-pulling effects while the attack target generation (ATG) constrains the spatial-temporal consistency.}
    \label{fig:opt-pipeline}
    \vspace{-0.3cm}
\end{figure}

This section introduces FlyTrap, the first physical and system-level DPA targeting camera-based ATT drones. As shown in Fig.~\ref{fig:opt-pipeline}, to achieve the attack goal and address design challenges (Section~\ref{sec:design-challenge}), our FlyTrap attack introduces novel designs: attack vectors, a progressive distance-pulling strategy, and controllable spatial-temporal consistency.

\subsection{Design Overview}
\label{sec:overview}

\noindent\textbf{Adversarial umbrella: A domain-specific, physically deployable attack vector.}
We propose \textit{adversarial umbrellas} as a novel class of physical attack vectors tailored for camera-based ATT drones. An umbrella is an ideal medium for adversarial patterns because: (1) it offers a large, nearly flat, rigid surface for pattern printing; (2) it naturally fits outdoor scenarios, requiring minimal setup and offering ease of transport and deployment; and (3) it offers fine control, allowing the attacker to maximize exposure and obscure themselves. Additionally, umbrellas do not require elaborate directional alignment or power sources, directly addressing challenge $C_1$ in the ATT drone context. In deployment, the attackers only need to cover their upper bodies and point the umbrella at the drone. While standing still is sufficient, crouching and hiding the entire body can further increase success by occluding any visible parts. Note that we don't mean to claim the physical adversarial patch as the major scientific contribution, but rather a practical delivery mechanism to support our design below.

\noindent
\textbf{Progressive distance-pulling via physical modeling.}
To address the challenge of closed-loop effectiveness ($C_2$), we proposed modeling the appearance of adversarial patterns as the drone gradually approaches, simulating the effects of reducing distance under DPA. By incorporating camera geometry, physical rendering, and our proven upper-bound shrink rate setups, our design ensures that the attack remains effective as the drone approaches, ensuring consistent distance-pulling effects. 

\noindent
\textbf{Controllable spatial-temporal consistency.}
To address the spatial-temporal consistency challenge ($C_3$), we introduce an attack target generator for adaptive attacks that jointly constrain spatial and temporal features across models and frames. The attack target generator explicitly encodes the spatial-temporal constraint by manipulating features like box shape, key points, or pose estimation within the adversarial region, simulating human-like motion and appearance. This enables us to bypass consistency-based defense systems, which are receiving growing attention in securing autonomous vehicles.

\begin{figure}[t!]
    \centering
    \includegraphics[width=\linewidth]{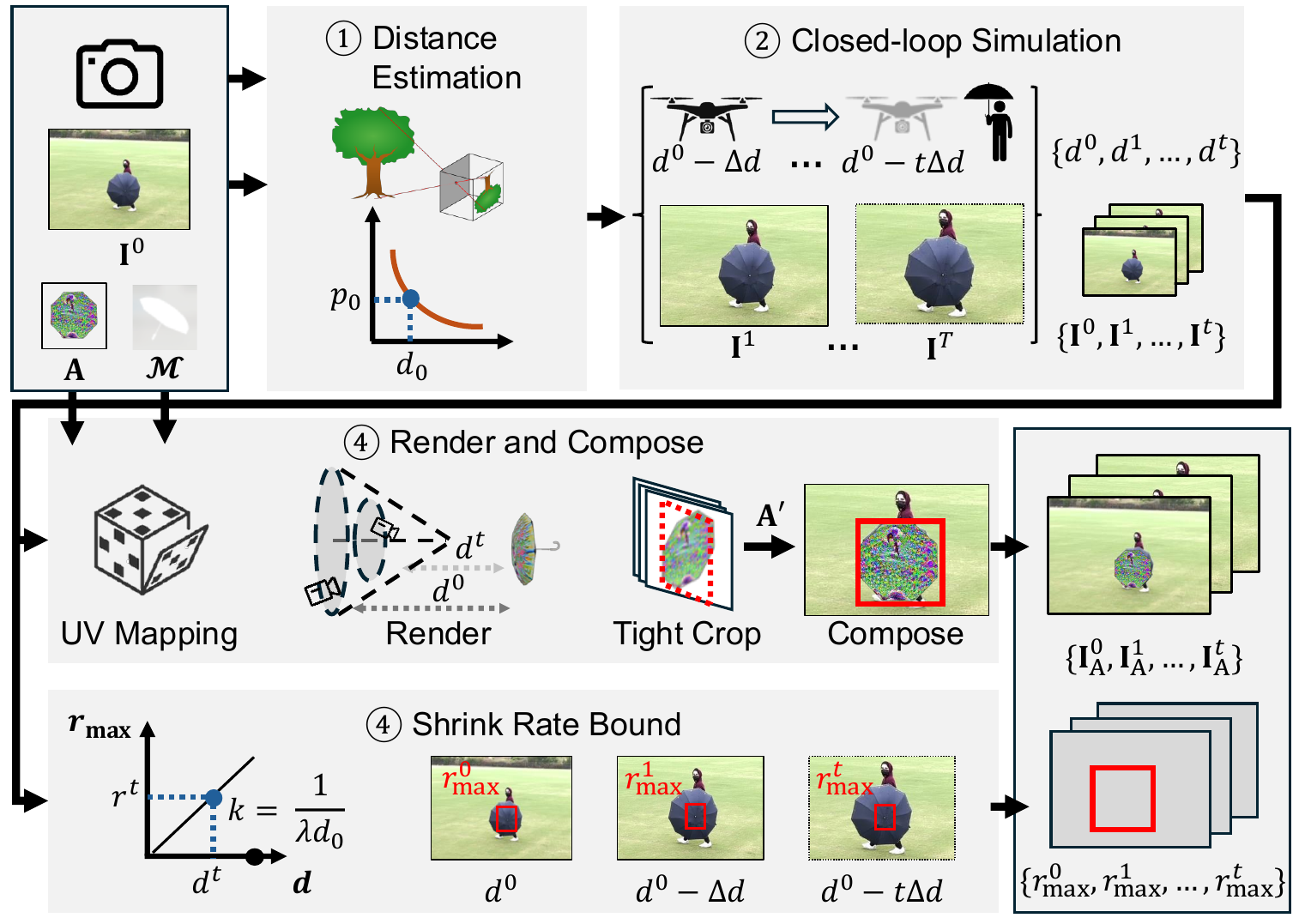}
    \caption{Design for progressive distance-pulling via physical modeling. We propose to leverage computer graphics to simulate the closed-loop dynamics of the DPA process. We further derive the upper bound to set the shrink rate for each stage.}
    \vspace{-0.3cm}
    \label{fig:physical-engine}
\end{figure}

\subsection{Progressive Distance-Pulling via Physical Modeling}
\label{sec:progressive-distance-pulling}

In this section, we introduce our solution to challenge~$C_2$. As shown in Fig.~\ref{fig:physical-engine}, via progressive distance-pulling (PDP), we simulate the effect through physical modeling in computer graphics, including \Circled{1}~distance estimation, \Circled{2}~closed-loop simulation, \Circled{3}~rendering and composition, and our proven \Circled{4} shrink rate bound. The initial inputs are the camera model, the search frame $\mathbf{I}^0\in \mathbb{R}^{H\times W\times 3}$, where $H$ and $W$ represent the height and width, the initial adversarial pattern $\mathbf{A}\in \mathbb{R}^{H_a\times W_a\times 3}$, where $H_a$ and $W_a$ represent the height and width of the adversarial pattern, and the umbrella 3D mesh $\mathcal{M}=(\mathcal{V}, \mathcal{E}, \mathcal{F})$, defining vertices, edges, and faces. The output is a set of simulated images $\{\mathbf{I}_{\mathbf{A}}^0, \mathbf{I}_{\mathbf{A}}^1, ..., \mathbf{I}_{\mathbf{A}}^t\}$ and maximum shrink rates $\{r^0_{\max}, r^1_{\max}, ..., r^t_{\max}\}$. The shrink rate is defined as the ratio between the attacked bounding box area and the umbrella’s visible area.

In step~\Circled{1}, we adopt a pinhole camera model with focal length $f$. The relationship between the pixel length $p$ and the actual length $s$ is defined by: $d = \frac{f\cdot s}{p}$. The focal length can be retrieved from the camera’s specifications, and with the pixel length in each image $\mathbf{I}^0$, we can estimate the distance between the drone and the object. This estimate serves as the initial distance $d^0$ for the subsequent closed-loop simulation.

In step~\Circled{2}, we simulate distance-pulling behavior as the drone incrementally approaches the target. Starting from the initial distance $d^0$, we iteratively reduce it using a user-defined interval: $d^t = d^0 - t\Delta d$, where $t$ denotes the time step and $\Delta d$ the distance decrement per step. Given each distance, we estimate the corresponding pixel length and synthesize a sequence of progressively zoomed-in images from $\mathbf{I}^0$, denoted as $\{\mathbf{I}^1, \ldots, \mathbf{I}^t\}$. We assume the camera is oriented directly toward the tracked object, an assumption justified by the attacker’s ability to aim the umbrella at the drone.

\begin{table*}[t!]
\centering
\caption{Overview of existing representative defense methods leveraging spatial and temporal features to secure perception models in autonomous vehicles. This table excludes sensor-level attacks, as it focuses solely on adversarial attacks targeting machine learning-based perception models. Therefore, sensor attacks (e.g., GPS Spoofing in PhyScout~\cite{xu2024physcout}) are not summarized here. OD and MOT refer to Object Detection and Multi-Object Tracking, respectively. ATG represents the attack target generator.}
\resizebox{\linewidth}{!}{%
    \begin{tabular}{c|c|c|c|c|c}
    \toprule
    Defense & Victim Model & Attack Goal & Spatial Feature & Temporal Feature & ATG \\
    \midrule
    \rowcolor{gray!20} PercepGuard~\cite{man2023person} & OD & Misclassification & Box Shape & Box Behavior, Ego Vehicle States & Inject Box Aspect Ratio \\
    PhyScout~\cite{xu2024physcout} & OD & Hiding, Appearing, Misclassification & Box Feature Point & Box Behavior, Ego Vehicle States & Inject Feature Point \\
    \rowcolor{gray!20} VOGUES~\cite{muller2024vogues} & MOT & Move-In, Move-Out, Hiding & Component Position & Component Behavior & Inject Human Pose\\
    PhySense~\cite{yu2024physense} & OD, MOT & Misclassification & 3D Shape, Texture & Object Behavior, Object Interaction & Inject Human Behavior\\
    \rowcolor{gray!20} VisionGuard~\cite{han2024visionguard} & OD & Hiding, Appearing, Misclassification & N/A &Ego Vehicle States & Multi-Stage Shrink Rate \\
    \bottomrule
    \end{tabular}%
}
\vspace{-0.3cm}
\label{tab:existing-defense}
\end{table*}

Then, we simulate the umbrella geometry with high physical fidelity. In the rendering step, we first construct a UV mapping, which projects a 2D adversarial pattern $\mathbf{A}$ onto the 3D model. The UV mapping function $\Phi: \mathbb{R}^2 \rightarrow \mathbb{R}^3$ maps 2D coordinates $\mathbf{u}_i \in \mathbb{R}^2$ of the adversarial pattern to the corresponding 3D positions $\mathbf{v}_i \in \mathbb{R}^3$ on the mesh vertices $\mathcal{V}$. This allows the seamless attachment of the adversarial pattern onto the umbrella's surface, accounting for its curvature and topology, which is essential for maintaining real-world fidelity during optimization. We render the umbrella by placing a virtual camera at the estimated positions $\{d^0, d^1, \ldots, d^t\}$ from the mesh. The camera is oriented toward the umbrella’s center, with elevation angle $\theta = 0$ and azimuth angle $\phi = 0$. We detail how camera angle randomization improves real-world robustness in Section~\ref{sec:phy-robust}. After rendering, we apply image processing steps, including grayscale conversion, binarization, and morphological dilation, to perform edge segmentation and remove background margins, which produces a tightly cropped rendered image $\mathbf{A}' \in \mathbb{R}^{H_a' \times W_a' \times 3}$, facilitating seamless composition in the next stage. $H_a'$ and $W_a'$ represent the height and width of the rendered and cropped adversarial patterns. The overall process can be expressed as:
\begin{equation}
    \label{eq:render}
    \mathbf{A}' = \text{TightCrop}(\text{Render}(\Phi(\mathbf{A}), d, \theta, \phi)).
\end{equation}
\indent
In the composing step, we compose the rendered adversarial pattern $\mathbf{A'}$ to the target location in the simulated image $\{\mathbf{I}^0, \mathbf{I}^1, ..., \mathbf{I}^t\}$ and get the final adversarial images $\{\mathbf{I}_{\mathbf{A}}^0, \mathbf{I}_{\mathbf{A}}^1, ..., \mathbf{I}_{\mathbf{A}}^t\}$. This is achieved by computing a projection matrix followed by an affine transformation.

Finally, in step~\Circled{4}, we formally derive Theorem~\ref{thm:metric_consistency}, which establishes the relationship between the shrink rate and the resulting pulling distance. Thus, given a distance $d^t$ in the closed-loop simulation, to ensure the attack can pull the drone into the next simulated distance $d^{t+1}$, the maximum shrink rate at step $t$ should be $r^t_{\max} = \frac{d^{t+1}}{d^0}$ divided by a constant $\lambda$, serving as the upper bound when setting the shrink rate for each distance. While one could trivially set all shrink rates to zero, doing so fails to control the spatial-temporal consistency, which is detailed in the next section.

To formally justify this relationship, we provide the following theorem based on the pinhole camera model, which establishes a mathematical link between the shrink rate and the resulting change in physical distance.

\newtheorem{theorem}{Theorem}
\label{thm:metric_consistency}
\begin{theorem}
    Let $d_0$ be the initial distance between the drone and the target, and let $d_a$ be the final distance. Let $r_a$ be the target shrink rate under a pinhole camera model with focal length $f$, and assume the area ratio between the umbrella and the human is a constant $\lambda = \tfrac{s_u}{s_h}$. If $r_a = \tfrac{d_a}{\lambda d_0}$, then the drone can be pulled to a distance of $d_a$, which is shown in Fig.~\ref{fig:physical-engine}.
\end{theorem}
\begin{proof}
Under a pinhole camera model, the pixel length $p$ of an object of length $s$ at distance $d$ is $p = \tfrac{f s}{d}$. Initially, the umbrella's pixel length is $p_{u0} = \tfrac{f\,s_u}{d_0}$. During the attack, the bounding box size is shrunk by a factor $r_a$, making its pixel length $\tfrac{f\,s_h}{d_0} r_a$. The drone compensates by advancing until the bounding box size equals the original human pixel length:
\[
  \tfrac{f\,s_u}{d}\,r_a \;=\; \tfrac{f\,s_h}{d_0}.
\]
\indent
Solving for $d$ yields $d = \lambda r_a\,d_0$. Consequently, if $r_a = \tfrac{d_a}{\lambda d_0}$, we have $d = d_a$.
\end{proof}

\subsection{Controllable Spatial-Temporal Consistency}
\label{sec:design-consistency}

Table~\ref{tab:existing-defense} summarizes key spatial-temporal consistency defenses proposed to detect adversarial perception attacks. These defenses share a common principle: cross-validating the victim models' predictions against auxiliary estimations derived from independent features or models. Therefore, we design the FlyTrap to be a highly spatial-temporal controllable DPA against ATT drones by introducing the attack target generator (ATG), shown in Fig.~\ref{fig:opt-pipeline}. ATG enables the attack to explicitly encode both spatial and temporal consistency during optimization, thereby allowing FlyTrap to bypass diverse spatial-temporal defense mechanisms. ATG formulates model-specific constraints to preserve spatial consistency as part of the optimization objective, ensuring intra-model consistency. For instance, we constrain the predicted box to maintain a human-like shape and appear within semantically plausible locations~\cite{man2023person, yu2024physense}. Additionally, ATG can embed adversarial feature points to mislead feature extractors~\cite{xu2024physcout} or craft deceptive human poses to fool pose estimators~\cite{muller2024vogues}. These manipulations are feasible due to the attacker’s full control over the umbrella pattern in FlyTrap. To ensure inter-model spatial consistency, ATG jointly optimizes multiple perception models, such as SOT, object detector, and pose estimator, such that their outputs align coherently. This coordination ensures spatial consistency within and across perception models.

In addition to spatial alignment, ATG enforces temporal consistency by aligning features across simulated frames. In the simulated images in PDP: $\{\mathbf{I}_{\mathbf{A}}^0, \mathbf{I}_{\mathbf{A}}^1, \ldots, \mathbf{I}_{\mathbf{A}}^t\}$, ATG determines the shrink rates $\{r^0, r^1,\ldots, r^t\}$ for each frame, constrained by the upper bounds specified in Theorem~\ref{thm:metric_consistency}. By selecting conservative shrink rates, ATG ensures a stable drone trajectory, minimizing abrupt changes that might trigger anomaly detectors. The gradual shrink rate design allows DPA to mimic benign scenarios in which the tracked object moves away at a plausible speed. This can prevent sudden box movement~\cite{man2023person} or sudden drone movement afterwards~\cite{han2024visionguard}. Finally, ATG enforces temporal feature alignment across frames. For instance, it can inject a consistent human pose throughout the PDP simulated attack sequence $\{\mathbf{I}_{\mathbf{A}}^0, \mathbf{I}_{\mathbf{A}}^1, ..., \mathbf{I}_{\mathbf{A}}^t\}$, thus evading defenses that monitor temporal behavior~\cite{muller2024vogues, man2023person, yu2024physense}.

To demonstrate ATG's generalizability, we categorize existing spatial-temporal consistency defenses into three classes and show how ATG can bypass each class, shown in Table~\ref{tab:existing-defense}. For future defense, ATG can also potentially bypass them if they fall within the categorized classes below. (1) \textit{Box feature-based defenses} inspect properties of the bounding box predicted by the victim model (e.g., SOT in our case), such as shape, location, and feature points. Representative examples include PercepGuard~\cite{man2023person} and PhyScout~\cite{xu2024physcout}. Both approaches examine box-level features. The ATG can set the attack target by manipulating the aspect ratio and feature point within the prediction box accordingly. (2) \textit{Extra visual feature-based defenses} analyze visual features beyond the primary victim model, such as those extracted from additional detectors or pose estimators. Examples include VOGUES~\cite{muller2024vogues} and PhySense~\cite{yu2024physense}. These defenses validate the spatial-temporal consistency of auxiliary visual cues, such as human pose and motion behavior, using modules like auxiliary detectors and temporal models. ATG can constrain the inter-model consistency to attack extra models simultaneously to bypass them. (3) \textit{Ego vehicle state-based defenses} detect anomalies by monitoring the smoothness of ego vehicle dynamics, such as velocity and acceleration. VisionGuard~\cite{han2024visionguard} is a representative example. However, unlike object detectors studied in VisionGuard, SOT models inherently exhibit spatial-temporal consistency, reducing abrupt vehicle movement changes afterwards~\cite{li2018high, cui2022mixformer}. Second, ATG can generate a multi-stage shrink rate to make the drone movement even smoother, thus bypassing defenses relying on vehicle state checking. It should be noted that although PercepGuard also uses vehicle states, it is mainly for assisting the box behavior prediction. PhyScout uses vehicle states for the reconstruction of 3D feature points. Neither of them uses ego states as major or direct evidence for detecting the underlying perception attacks. Therefore, we don't include them in this category.

\subsection{Overall Optimization Pipeline}
After introducing FlyTrap's key novel designs, we present the comprehensive optimization process from data collection to real-world robustness in this section, as shown in Fig.~\ref{fig:opt-pipeline}.

\subsubsection{Dataset Construction}
We first introduce the dataset construction process to train the adversarial pattern. Specifically, we collect aerial videos that depict common deployment areas for ATT drones. Each video is split into two segments, the template video and the search video: $V=\{V_\text{tplt}, V_\text{search}\}$. The first segment tracks a person, corresponding template videos $V_\text{tplt}=\{\mathbf{I}_{\text{tplt}_1}, \mathbf{I}_{\text{tplt}_2}, ...\}$, where $\mathbf{I}_{\text{tplt}_n}\in \mathbb{R}^{H\times W\times 3}$ represent $n$-th template frame in the video. The template frames are used to initialize the tracker. The second segment records the same subject deliberately opening an umbrella and pointing it at the drone to simulate adversarial behavior, resulting in search frames $V_\text{search}=\{\mathbf{I}_{\text{search}_1}, \mathbf{I}_{\text{search}_2}, ...\}$. $\mathbf{I}_{\text{search}_n}\in \mathbb{R}^{H\times W\times 3}$ represents the $n$-th search frame, where the tracker makes predictions. These search frames serve as inputs to the PDP (Section~\ref{sec:progressive-distance-pulling}), which simulates adversarial umbrellas and distance-pulling dynamics. We further perform automated labeling and down-sampling, as detailed in Appendix~\ref{app:data-process}.

\subsubsection{Adversarial Objective Function}
\label{sec:adv-obj-func}

We leverage the attack target from the ATG as our optimization goal. Specifically, for each PDP time step $t$, we guide the SOT model to predict a bounding box of size $w_{a}^t$ and $h_a^t$ and centered at $cx_a^t$ and $cy_a^t$, represented as a tuple $\mathcal{P}_a^t = (cx_a^t, cy_a^t, w_a^t, h_a^t)$. We set all $cx_a^t$ and $cy_a^t$ to the center of the umbrella by default:
\begin{equation}
\textstyle
\label{eq:loc} 
    \mathcal{L}_{\text{loc}} = \frac{1}{NMT} \sum_{i=1}^{N} \sum_{j=1}^{M} \sum_{t=1}^{T}\left\lVert \mathcal{P}_{i, j}^t \ominus \mathcal{P}_{a}^t \right\rVert,
\end{equation}
where $N$, $M$, and $T$ are the overall number of search video frames, tracking candidate proposals, and PDP time steps. Additionally, the control algorithms are designed to react to tracking results only if their confidence scores are sufficiently high to ensure safe autonomous flight~\cite{dji_activetrack_tracking_state}. Thus, we maximize the predicted confidence $score_i$ to ensure that our injected tracking results can propagate throughout the ATT drones:
\begin{equation}
\label{eq:cls}
    \mathcal{L}_{\text{cls}} = \frac{1}{NMT} \sum_{i=1}^{N} \sum_{j=1}^{M} \sum_{t=1}^{T}\left[-\log(score_{i,j}^t)\right].
\end{equation}
\indent
Beyond SOT, we also co-optimize across multiple models to satisfy spatial-temporal consistency constraints for adaptive attacks. For example, to achieve spatial consistency, we assign the same location and confidence objectives (Eq.~\ref{eq:loc} and \ref{eq:cls}) to an object detector. For those defenses that use auxiliary pose estimation model~\cite{muller2024vogues}, we assume the attackers can control their pose right before launching the attack, and preserve the temporally consistent pose by injecting it in consecutive attack frames. Specifically, we optimize the pose estimation heat map $\mathbf{H}^t$ at each time step to remain close enough to a benign reference map $\mathbf{H}_{\text{benign}}$, which is averaged across the last few frames in the template video before umbrella deployment:
\begin{equation}
\label{eq:pose}
    \mathcal{L}_{pose} = \frac{1}{NT}\sum_{i=1}^N \sum_{t=1}^T \left\lVert \mathbf{H}_i^t - \mathbf{H}_{\text{benign}} \right\rVert.
\end{equation}
\indent
Lastly, to maintain physical-world realizability, we regularize the adversarial patterns with a total variation (TV) loss:
\begin{equation}
\label{eq:tv-loss}
\mathcal{L}_{\mathrm{TV}}(\mathbf{A}) = 
\sum_{i=1}^{H_a - 1} \sum_{j=1}^{W_a - 1}
\left\lVert \mathbf{A}_{i+1,j} - \mathbf{A}_{i,j} \right\rVert +
\left\lVert \mathbf{A}_{i,j+1} - \mathbf{A}_{i,j} \right\rVert,
\end{equation}
where $\mathbf{A}_{i,j}$ represent pixel values at location $(i, j)$ on the adversarial pattern $\mathbf{A}$ before rendering. Finally, we optimize the adversarial pattern as a weighted sum of all the objectives:
\begin{equation}
    \min_{\mathbf{A}}\mathbb{E}_{\mathcal{T}\sim\mathcal{T}_\text{compose}}\big[\sum_k w_k\mathcal{L}_k\big],
\end{equation}
where $w_k$ is the tuned weight to balance the $k$-th objective function and $\mathcal{T}$ denotes the transformation detailed below.

\subsubsection{Physical-world robustness}
\label{sec:phy-robust}
To overcome the influence of innumerable physical factors, we stack a set of expectations over transformation (EoT) within the optimization process~\cite{athalye2018synthesizing}. In the rendering operation, we randomly select the camera
elevation angle $\mathcal{T}_1(\cdot): \theta\sim[-5^{\circ}, 5^{\circ}]$ and azimuth angle $\mathcal{T}_2(\cdot): \phi\sim[-5^{\circ}, 5^{\circ}]$ to simulate the attacker pointing the umbrella slightly off the camera center. We randomly sample the angle of rotation $\mathcal{T}_3(\cdot): \psi\sim[-20^{\circ}, 20^{\circ}]$ of the umbrella to simulate the imperfect vertical direction of the physical adversarial pattern. Additionally, we add image transformations to the adversarial patterns, including Gaussian noise ($\mathcal{T}_4$), brightness ($\mathcal{T}_5$), contrast ($\mathcal{T}_6$), saturation ($\mathcal{T}_7$), and hue transformation ($\mathcal{T}_8$) to simulate complex physical world environments. The final transformation is composed of all the transformations $\mathcal{T}_\text{compose} = \left\{ \mathcal{T}_1 \circ \mathcal{T}_2 \circ \cdots \circ \mathcal{T}_8 \right\}$. Additionally, the PDP naturally incorporates estimation error due to imperfect physical modeling, accounting for the real-world imperfection control assumed in the closed-loop simulation.
\section{Attack Evaluation}

\subsection{General Experimental Setups}
\label{sec:eval-setup}

\subsubsection{Dataset collection}
\label{sec:exp-data}
We collected an aerial-view dataset for training and evaluation. The dataset includes video recordings featuring four individuals with diverse appearances and covers four typical drone deployment environment types, including two grass fields, two parking lots, one bare ground area, and one drivable road. For each combination, we recorded two videos: one for training and one for testing. In total, the dataset includes 23 training videos comprising 11,898 frames and 25 evaluation videos comprising 13,594 frames. Ethics considerations can be found in Section~\ref{sec:ethics}.

\subsubsection{Models}
\label{sec:models}
In our experiments, we choose SiamRPN-based~\cite{li2018high} SOT models as victims, following prior work~\cite{muller2022physical}, due to their strong trade-off between tracking accuracy and computational efficiency. To further broaden our evaluation, we also include MixFormer, a state-of-the-art Transformer-based SOT model~\cite{vaswani2017attention}, which represents the recent trend toward more expressive yet computationally intensive tracking architectures. By incorporating both CNN-based and Transformer-based models, this combination provides comprehensive coverage of the current SOT model landscape.

\subsubsection{Metrics}
\label{sec:metrics}
Under the DPA setting, we define evaluation metrics to capture the system-level impact of the attacks. Specifically, we define two metrics: (1) open-loop attack success rate (ASR$_\text{open}$) and (2) closed-loop attack success rate (ASR$_\text{closed}$). ASR$_\text{open}$ is defined as successful if all of the following conditions are satisfied: (1) to ensure the drone is expected to be pulled to within the attacker-desired distance: the bounding box area must be smaller than a shrinkage threshold $r_a$ of the umbrella area; (2) to ensure the ATT system doesn’t lose track and thus fail in distance-pulling: the prediction confidence must exceed a predefined threshold $score_a$ and (3) to ensure the umbrella is the trigger: the attacked prediction bounding box must lie entirely within the umbrella boundary:
\begin{align*}
\mathcal{C}_\text{open} :
\begin{cases}
    a \le r_a \cdot a_u, \\
    score \geq score_a, \\
    (c_x, c_y, w, h) \subseteq (c_{u_x}, c_{u_y}, w_u, h_u),
    \end{cases}
\end{align*}
where $a$ denotes the bounding box area and the $u$ subscript denotes umbrella. We evaluate frames from the testing dataset and compute the ASR$_\text{open}$ as:
\[
    \text{ASR$_\text{open}$} = \frac{\sum_{i=1}^{N} \mathbb{I}(\mathcal{C}_i)}{N},
\]
where $\mathbb{I}$ is the indicator function, $\mathcal{C}_i$ represents the condition for the $i$-th sample, and $N$ is the total number of frames. To capture ASR comprehensively across varying threshold settings, we introduce a metric similar to mean Average Precision (mAP) used in object detection~\cite{everingham2010pascal, lin2014microsoft}. We define mean ASR$_\text{open}$ (mASR$_\text{open}$) as the average ASR$_\text{open}$ over a set of thresholds $r_{a}$ and $score_a$ ranging from 0.1 to 0.9 in increments of 0.1 to provide broad coverage. For the ASR$_\text{closed}$, we define success if the final distance $d$ between the drone and the attacker is below a distance threshold $d_a$:
\begin{align*}
    \mathcal{C}_{\text{closed}}: d \leq d_a.
\end{align*}
~~~~The ASR$_\text{closed}$ is computed as the average success rate across multiple real-world drone flights. The success criterion for ASR$_\text{closed}$ is straightforward: the distance threshold can be the maximum range to capture the drone (e.g., 9 meters~\cite{net_capture}) or it can be the working distance for sensor attacks (e.g., 6 meters for projector~\cite{zhou2022doublestar}) or hitting distance (e.g., 0.5 meters).
\subsection{Attack Effectiveness}
\label{sec:digital-eval}

\subsubsection{Evaluation Methodology}
\label{sec:exp-eff-methodology}
As a baseline, we use target photos (TGT) cropped from the first frame of each training video, corresponding to the same person and location, as they naturally resemble the genuine target being tracked. The TGT can be regarded as a simple human figure printing baseline attack. We use grid search to find the ratio of the printed human figure to the umbrella that can maximize mASR$_\text{open}$ and use that for fair baseline comparisons. More TGT generation details are included in the Appendix~\ref{app:tgt}. Since TGT also applies an image on the umbrella surface, the mASR$_\text{open}$ can be naturally applied to it. This evaluation involved 4 people $\times$ 4 locations = 16 TGT combinations in total. The mASR$_{\text{open}}$ was averaged over 16 TGTs $\times$ 6 testing videos = 96 combinations of experiments. Regarding FlyTrap, we select two people as the tracked target and two locations as the background for training. Then, we evaluate FlyTrap on the 6 testing videos of the same target person and background. Vanilla FlyTrap can also be considered as a baseline from the previous SOT shrinking attack~\cite{ding2021towards, yan2020cooling}, but with our new contributions
of the umbrella modeling, DPA-specific objective function design, and attack vector-specific robustness design.

\begin{table}[t!]
\centering
\caption{Evaluation of attack effectiveness (mASR$_\text{open}$). SiamAlex, SiamRes, and SiamMob refer to SiamRPN~\cite{li2018high} combined with AlexNet~\cite{krizhevsky2012imagenet}, ResNet~\cite{he2016deep}, and MobileNet~\cite{howard2017mobilenets}, respectively. The FlyTrap attack consistently outperforms the TGT baseline,  despite its visual similarity to the tracked object. FlyTrap$_\text{PDP}$ consistently outperforms the vanilla version.}
\resizebox{\linewidth}{!}{%
    \begin{tabular}{l|cccc|c}
    \toprule
    \textbf{Attack} & \textbf{MixFormer} & \textbf{Siam-Alex.} & \textbf{Siam-Res.} & \textbf{Siam-Mob.} & \textbf{Avg.} \\
    \midrule
    TGT & $46.3\%$ & $37.2\%$ & $24.9\%$ & $35.5\%$ & $36.0\%$ \\
    \midrule
    FlyTrap & $42.0\%$ & $17.0\%$ & $44.3\%$ & $32.1\%$ & $33.9\%$ \\
    \rowcolor{gray!20} FlyTrap$_\text{PDP}$ & $78.7\%$ & $35.6\%$ & $50.8\%$ & $49.1\%$ & $53.6\%$ \\
    \bottomrule
    \end{tabular}%
}
\label{tab:main-res}
\end{table}

\begin{table}[t!]
\centering
\caption{Evaluation of scenario universality for attacks across unseen target-location combinations (mASR$_\text{open}$). We evaluate the attack on two unseen people and two unseen locations with different target-location combinations. The FlyTrap in this table is the version with the PDP design.}
\resizebox{\linewidth}{!}{%
    \begin{tabular}{l|cc|cc|cc}
    \toprule
    & \multicolumn{6}{c}{\textbf{Scenario Universality}} \\
    \cline{2-7}
    \multirow{3}{*}{\centering \textbf{Model}} & \multicolumn{2}{c|}{Location (6 Videos)} & \multicolumn{2}{c|}{Person (7 Videos)} & \multicolumn{2}{c}{Both (6 Videos)} \\
    & TGT & \cellcolor{gray!20}FlyTrap & TGT & \cellcolor{gray!20}FlyTrap & TGT & \cellcolor{gray!20}FlyTrap \\
    \midrule
    MixFormer & $25.5\%$ & \cellcolor{gray!20}$85.9\%$ & $11.6\%$ & \cellcolor{gray!20}$40.4\%$ & $6.8\%$ & \cellcolor{gray!20}$34.1\%$ \\
    SiamRPN-Alex & $34.3\%$ & \cellcolor{gray!20}$50.2\%$ & $24.2\%$ & \cellcolor{gray!20}$67.9\%$ & $21.7\%$ & \cellcolor{gray!20}$33.0\%$ \\
    SiamRPN-Res & $20.7\%$ & \cellcolor{gray!20}$55.2\%$ & $10.4\%$ & \cellcolor{gray!20}$63.5\%$ & $9.9\%$ & \cellcolor{gray!20}$42.8\%$ \\
    SiamRPN-Mob & $28.8\%$ & \cellcolor{gray!20}$55.9\%$ & $13.8\%$ & \cellcolor{gray!20}$54.5\%$ & $12.2\%$ & \cellcolor{gray!20}$26.0\%$ \\
    \midrule
    Average & $27.3\%$ & \cellcolor{gray!20}$61.8\%$ & $15.0\%$ & \cellcolor{gray!20}$56.6\%$ & $12.6\%$ & \cellcolor{gray!20}$34.0\%$ \\
    \bottomrule
    \end{tabular}%
}
\vspace{-0.3cm}
\label{tab:main-uni}
\end{table}

\subsubsection{Experiment Results}
The main results of our evaluation are presented in Table~\ref{tab:main-res}. We find TGT, even though it visually matches the genuine person being tracked, performs considerably limited, with an average mASR$_{\text{open}}$ of 36.0\% across all victim models. In comparison, FlyTrap$_\text{PDP}$ achieves a much higher mASR$_{\text{open}}$ of 53.6\% on average, underscoring its effectiveness. The comparison between FlyTrap with and without PDP design shows its effectiveness in further shrinking the area, which is also observed in the physical experiments. We also study the robustness of FlyTrap to environmental distractions, where multiple similar but unobstructed objects (e.g., other passersby) appear in the same scenario and find that FlyTrap can cause consistent attack effects given the presence of other visual distractions. Please refer to the case study demonstration in Appendix~\ref{sec:case-study}.

It's worth noting that mASR$_{\text{open}}$ is a challenging metric. Specifically, assume the $\mathcal{C}_{\text{open}}$ can always be satisfied when $\forall~r_a \ge 0.5,~\mathbb{I}(\mathcal{C}_{\text{open}})=1$ and $\forall~r_a < 0.5,~ \mathbb{I}(\mathcal{C}_{\text{open}})=0$, the mASR$_{\text{open}}$ will be 50.0\%. However, this can already shorten the tracking distance to half of its original distance as indicated by Theorem~\ref{thm:metric_consistency}. We show in physical experiments (Section~\ref{sec:phy-open-loop}) that the shrink rate can continuously decrease as the distance decreases. Therefore, the mASR$_{\text{open}}$ primarily serves for digital, scalable evaluation before printing adversarial patterns for physical evaluation. Thus, even though the absolute number of mASR$_{\text{open}}$ might not seem as high as expected, we find it already sufficient enough to cause closed-loop impacts as indicated by our physical closed-loop experiments (Section~\ref{sec:phy-close-loop}).

\subsection{Scenario Universality}
\label{sec:universality}

\subsubsection{Evaluation Methodology}
For TGT, we apply the same set of images from Section~\ref{sec:digital-eval} to videos of unseen scenarios, including different target persons and/or different background locations. The results are derived from 16 TGTs $\times$ 19 testing videos =~304 combinations of tests. For FlyTrap, we use the same set of adversarial patterns optimized in Section~\ref{sec:digital-eval} to an unseen person and/or unseen background. We report the mASR$_{\text{open}}$ of universality to location (6 testing videos), to person (7 testing videos), and both (6 testing videos).

\subsubsection{Experiment Results}
In Table~\ref{tab:main-uni}, we observe that the TGT shows limited universality, even for the same tracked person with a different background. Its universality to location is only 27.3\% across all models. The universality to person and to both are even worse. Therefore, TGT might only be useful if the attacker knows the exact scenario, including both the person and location. On the contrary, FlyTrap shows a significantly better universality of 61.8\%. The universality to location and to person can achieve comparable mASR$_{\text{open}}$ as effectiveness shown in Table~\ref{tab:main-res}, suggesting that FlyTrap, when trained on a subset of location or person, can generalize effectively to unseen individuals or environments, satisfying the need to attack ATT drones in unknown deployment places. However, when both the person and location are unseen, the mASR$_{\text{open}}$ are slightly lower, but still 21.4\% higher than TGT.

\begin{figure}[t]
    \centering
    \includegraphics[width=0.6\linewidth]{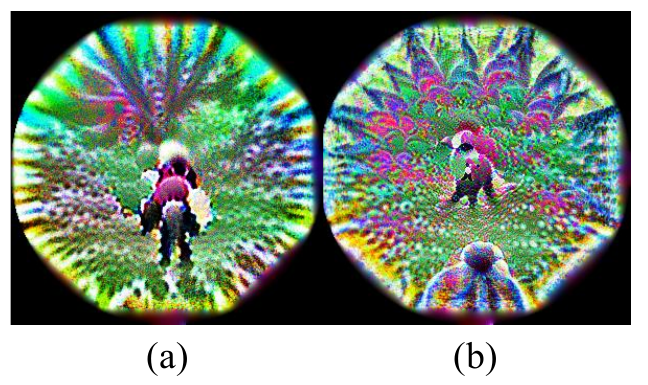}
    \caption{Visualization of adversarial patch designs against MixFormer~\cite{cui2022mixformer}. (a) Umbrella pattern without progressive distance-pulling, which achieves high transferability due to its visual resemblance to a standing human. (b) Umbrella pattern with progressive distance-pulling, exhibiting a structured cascade pattern that enhances continuous distance-pulling effects.}
    \label{fig:adv-pattern}
    \vspace{-0.3cm}
\end{figure}

\subsection{Attack Transferability}
\label{sec:transfer}
\subsubsection{Evaluation Methodology}
We employ the FlyTrap optimized from one victim SOT model for transferring to attack another. We consider FlyTrap with and without PDP designs. We study the transferability of FlyTrap without PDP as we observe an interesting adversarial pattern: the human-shape pattern in Fig.~\ref{fig:adv-pattern} (a), which might be more transferable as it visually resembles a standing human. Then, we report the mASR$_{\text{open}}$ on the same set of testing videos as Section~\ref{sec:exp-eff-methodology}.

\begin{figure}[t]
    \includegraphics[width=\linewidth]{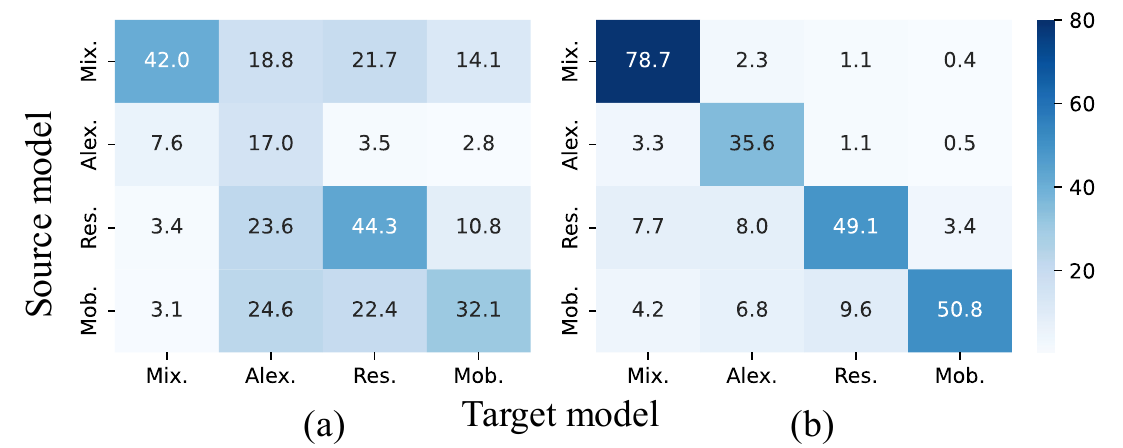}
    \caption{Attack transferability evaluation (mASR$_\text{open}$\%). (a) FlyTrap pattern optimized without progressive distance-pulling (PDP) design; and (b) with PDP design.}
    \label{fig:trans}
\end{figure}

\subsubsection{Experiment Results}

The main results are shown in Fig.~\ref{fig:trans}. Notably, we observe that the human-shape patterns indeed have better transferability, with an average of mASR$_{\text{open}}$ of 13.1\% compared to FlyTrap with PDP design, which is 4.0\%. The phenomenon can be explained by the visual appearance of the optimized pattern shown in Fig.~\ref{fig:adv-pattern}. Specifically, Fig.~\ref{fig:adv-pattern} (a) shows visual resemblance to a standing human, which potentially benefits the transferability, while Fig.~\ref{fig:adv-pattern} (b) exhibits a structured cascade pattern that enhances continuous distance-pulling effects but is more model-specific. Among them, the adversarial pattern against MixFormer~\cite{cui2022mixformer} shows the highest transferability of 18.2\% mASR$_{\text{open}}$ on average. Such a level of transferability can already achieve effective DPA against black-box commercial systems (Section~\ref{sec:commercial}). The results suggest a trade-off between a more continuous distance-pulling attack to a more transferable attack, which we acknowledge as a limitation in Section~\ref{sec:limitation}.

\subsection{Spatial-temporal Consistency}
\label{sec:spatial-temporal}

We select one defense for each of our three categorized classes (Section~\ref{sec:design-consistency}) to evaluate FlyTrap's spatial-temporal consistency for each defense type, using the same set of testing videos in effectiveness evaluation (Section~\ref{sec:digital-eval}).

\begin{table}[t!]
\centering
\caption{Evaluation of the PercepGuard~\cite{man2023person} defense. We report False Alarm Rates (FAR) under benign inputs and True Alarm Rates (TAR) under vanilla and FlyTrap$_\text{ATG}$ attacks. Notably, the FlyTrap$_\text{ATG}$ achieves an average detection rate of only 2.8\%, which is lower than the 5\% benign PercepGuard FAR reported in~\cite{muller2024vogues}.}
\resizebox{0.8\linewidth}{!}{%
    \begin{tabular}{l|c|cc}
    \toprule
    \multirow{2}{*}{\centering \textbf{Model}} & \textbf{False Alarm Rate} & \multicolumn{2}{c}{\textbf{True Alarm Rate}} \\
    & Benign & FlyTrap &  \cellcolor{gray!20} FlyTrap$_\text{ATG}$ \\
    \midrule
    MixFormer & $2.6\%$ & $78.0\%$ & \cellcolor{gray!20}$2.2\%$ \\
    Siam-Alex & $0.0\%$ & $41.2\%$ & \cellcolor{gray!20}$4.9\%$ \\
    Siam-Res & $0.0\%$ & $55.2\%$ & \cellcolor{gray!20}$0.0\%$ \\
    Siam-Mob & $0.0\%$ & $68.2\%$ & \cellcolor{gray!20}$4.2\%$ \\
    \midrule
    Average & $0.7\%$ & $60.7\%$ & \cellcolor{gray!20}$2.8\%$ \\
    \bottomrule
    \end{tabular}%
}
\vspace{-0.3cm}
\label{tab:percepguard}
\end{table}

\subsubsection{PercepGuard Evaluation}
We adopt the released behavior LSTM~\cite{hochreiter1997long} model in the official codebase~\cite{percepguard_code2023}. The input to the LSTM model is the tracked bounding box prediction in the last ten frames. The output of the LSTM model is a probability distribution over several classes. We regard the alarm as raised if the prediction is not ``pedestrians''. The results are shown in Table~\ref{tab:percepguard}. The original PercepGuard true alarm rate (TAR) in benign case and false alarm rate (FAR) in attack case are 99.0\% and 5.0\%, respectively~\cite{man2023person, muller2024vogues}. Our evaluation found a similar trend of low FAR in benign tracking cases and high TAR for vanilla FlyTrap. However, with our ATG design, the FlyTrap$_\text{ATG}$ attack can decrease the TAR to 2.8\%, even lower than the FAR of around 5.0\% reported in the original paper~\cite{man2023person}. The FAR in our evaluation is lower because our collected dataset is single object tracking scenarios, while in the original evaluation, it's tested in driving scenarios with a more complex environment (i.e., the BDD dataset~\cite{yu2020bdd100k}), thus leading to slightly higher FAR. 

\begin{figure}
    \centering
    \includegraphics[width=0.6\linewidth]{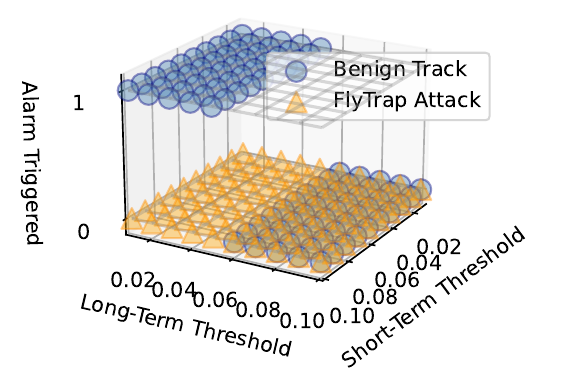}
    \caption{Evaluation results of VisionGuard~\cite{han2024visionguard}. In the z-axis, 0 represents alarm is not activated, and 1 represents that the alarm is triggered. FlyTrap can bypass it across all the tuned thresholds, even when the benign track triggers the alarm.}
    \label{fig:visionguard}
    \vspace{-0.3cm}
\end{figure}

\subsubsection{VOGUES Evaluation} 
VOGUES~\cite{muller2024vogues} was originally proposed for defense Multiple Object Tracking (MOT). Therefore, we follow the original setups while making necessary adaptations for SOT. Following their setups~\cite{vogues_code2023}, we adopt YOLOv3~\cite{redmon2018yolov3} as the object detector and AlphaPose~\cite{fang2022alphapose} as the pose estimator. We compute the spatial consistency by choosing the object detection prediction that has the highest Intersection of Union (IoU) with the SOT prediction to avoid high false positive rates during the SOT adaptation. Since VOGUES doesn't release the LSTM model they use to evaluate the consistency of the human pose, we reproduce it by following their setups: we train the LSTM model using the pose from the UCF101 dataset~\cite{soomro2012ucf101}. Same to their setups, we set the OD and OT prediction IoU threshold as 0.5 and the LSTM threshold as 0.5. If any of the values is below the threshold, an alarm will be raised. We observe that SOT can mostly operate normally even if a benign umbrella is used as camouflage. Therefore, the defense is desired to tolerate the spatial-temporal anomaly induced by a normal umbrella, avoiding an unnecessarily high false alarm rate that hampers ATT drones from operating normally in this case. 

Table~\ref{tab:vogues} shows that with the ATG design, FlyTrap$_\text{ATG}$ can largely decrease the alarm rate across models and almost achieve consistently lower alarm rates than a benign umbrella. We find that the relatively high alarm rate compared with benign tracking is caused by imperfect universal attacks against the object detector across multiple frames and videos, which is an observed limitation in existing attacks against object detectors~\cite{han2024visionguard}. Nonetheless, the ablation study of FlyTrap$_\text{ATG}$ is already sufficient in justifying the effectiveness of our
ATG design and is significantly lower than the effectiveness level that the original VOGUES used to claim success (98.4\%).

\begin{table}[t!]
\centering
\caption{VOGUES~\cite{muller2024vogues} defense evaluation. FlyTrap$_\text{ATG}$ can largely decrease the alarm rate across models and almost achieve consistently lower alarm rates than a benign umbrella.}
\resizebox{0.8\linewidth}{!}{%
    \begin{tabular}{l|cc|cc}
    \toprule
    \multirow{2}{*}{\centering \textbf{Model}} & \multicolumn{2}{c}{\textbf{False Alarm Rate}} & \multicolumn{2}{c}{\textbf{True Alarm Rate}} \\
    & Benign & Umbrella & FlyTrap &  \cellcolor{gray!20} FlyTrap$_\text{ATG}$ \\
    \midrule
    MixFormer & $3.1\%$ & $51.6\%$ & $93.7\%$ & \cellcolor{gray!20} $32.4\%$ \\
    Siam-Alex & $4.2\%$ & $75.6\%$ & $89.3\%$ & \cellcolor{gray!20} $77.9\%$ \\
    Siam-Res & $3.4\%$ & $73.0\%$ & $91.4\%$ & \cellcolor{gray!20} $54.4\%$ \\
    Siam-Mob & $3.9\%$ & $65.1\%$ &$80.8\%$ & \cellcolor{gray!20} $44.8\%$ \\
    \midrule
    Average & $3.7\%$ & $66.3\%$ &$88.8\%$ & \cellcolor{gray!20} $52.4\%$ \\
    \bottomrule
    \end{tabular}%
}
\label{tab:vogues}
\end{table}

\subsubsection{VisionGuard Evaluation} 
We use the official implementation for evaluation~\cite{visionguard_code2024}. To collect drone states, we use the Clover drone simulator~\cite{copterexpress_clover}, which uses Gazebo~\cite{koenig2004design} as the backend and PX4~\cite{px4} as the firmware. We simulate the benign tracking and the FlyTrap attacked tracking in \texttt{OFFBOARD} mode~\cite{px4} by publishing to the ROS topic \texttt{mavros/setpoint\_position/local} to set the target point at a frequency of 20 Hz, simulating the SOT model's inference FPS. Following their setup, we only consider the drone and person moving along the x-axis for simplicity, without losing generality. The drone's states can be retrieved by subscribing to \texttt{mavros/local\_position/velocity\_local}, which is used as the input to the ARIMA~\cite{box2015time} states estimation model. Finally, we use one benign sequence to train the model for predicting the drones' states and test on the FlyTrap attack sequence and another benign track sequence. Each sequence has around 130 frames sampled every 0.1 seconds. We iterate the alarm threshold, including long-term residual and short-term residual, and set the accumulated threshold to 2. In Fig.~\ref{fig:visionguard}, we find that FlyTrap can bypass VisionGuard~\cite{han2024visionguard} across all the tuned parameters, even when the benign track triggers the alarm. The results suggest that the fundamental rationale of VisionGuard to aim for inconsistent attack effects in object detection has limited applicability to the ATT drone context, where the SOT prediction is temporally consistent by nature.

\begin{table}[t!]
\centering
\caption{Attack effectiveness evaluation (mASR$_\text{open}$) with spatial-temporal constraint tailored for different defenses. The constraint has a subtle impact (within 10\%) on the attack performance across the models, which are still significantly higher than TGT.}
\resizebox{\linewidth}{!}{%
    \begin{tabular}{l|cccc}
    \toprule
    \textbf{ATG} & \textbf{MixFormer} & \textbf{Siam-Alex.} & \textbf{Siam-Res.} & \textbf{Siam-Mob.} \\
    \midrule
    - & $78.7\%$ & $35.6\%$ & $50.8\%$ & $49.1\%$ \\
    \midrule
    PercepGuard & $76.8\%$ & $51.1\%$ & $53.5\%$ & $47.6\%$ \\
    VOGUES & $69.4\%$ & $41.4\%$ & $40.6\%$ & $40.5\%$ \\
    \bottomrule
    \end{tabular}%
}
\vspace{-0.3cm}
\label{tab:defense-asr}
\end{table}

\subsubsection{Impact on Attack Effectiveness} In Table~\ref{tab:defense-asr}, we evaluate the impact of ATG design on attack effectiveness. The results show that the ATG design to constrain spatial-temporal consistency has a subtle impact (within 10\%) on the attack performance, which is still significantly higher than TGT (in Table~\ref{tab:main-res}). For example, FlyTrap against MixFormer with spatial-temporal constraints can achieve similar mASR$_{\text{open}}$ compared with FlyTrap without constraints: 76.8\% for PercepGuard~\cite{man2023person} and 69.4\% for VOGUES~\cite{muller2024vogues}. Interestingly, we find the FlyTrap with ATG can even significantly boost the mASR$_{\text{open}}$ for SiamRPN-AlexNet to 51.1\% and 41.4\%.

\subsection{Physical-World Attack Evaluation}
\label{sec:exp-phy}

\begin{figure}
    \centering
    \begin{subfigure}[t]{0.48\linewidth}
        \centering
        \includegraphics[width=\linewidth]{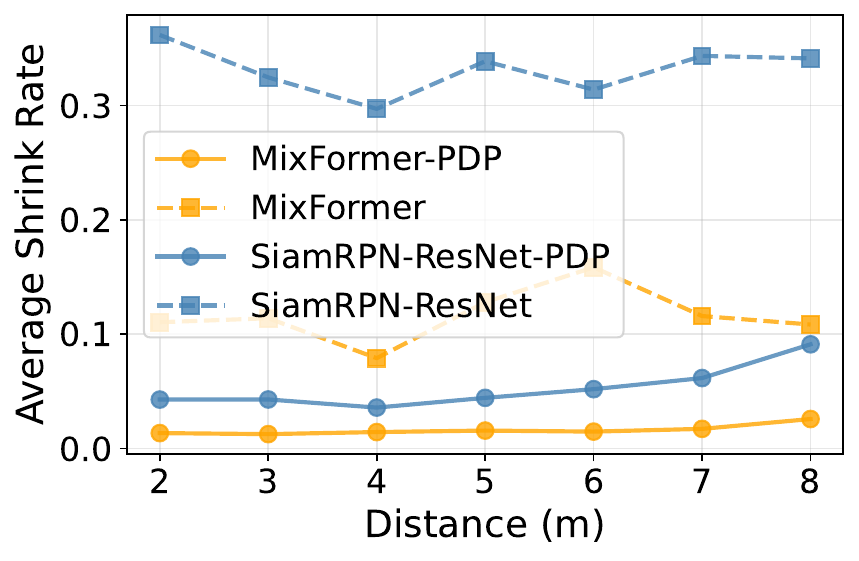}
        \caption{Short-range evaluation}
        \label{fig:shrink_rate_distance}
    \end{subfigure}%
    \hfill
    \begin{subfigure}[t]{0.48\linewidth}
        \centering
        \includegraphics[width=\linewidth]{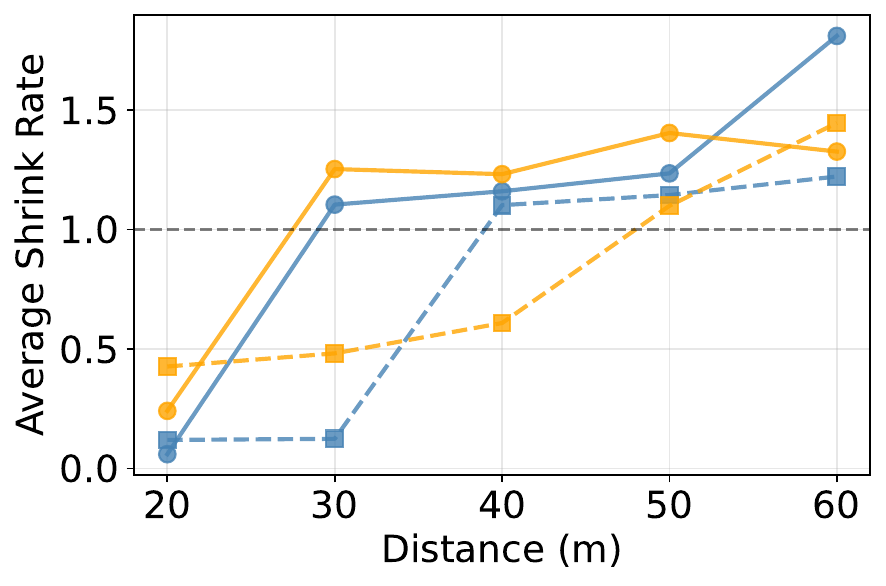}
        \caption{Long-range evaluation}
        \label{fig:max-distance}
    \end{subfigure}
    \caption{Physical evaluation of average shrink rate versus distance. (a) FlyTrap$_\text{PDP}$ achieves a lower shrink rate across all distances in short-range evaluation since the shrink rate continuously decreases as distance decreases. (b) FlyTrap works well below 20m and can
    potentially extend to 30m or 40m. \textit{PDP} means using PDP during attack optimization.}
    \label{fig:distance_evaluation}
    \vspace{-0.3cm}
\end{figure}

\subsubsection{Open-Loop Evaluation Setups}
\label{sec:phy-open-loop-setup}
We create real-world adversarial umbrella prototypes by uploading the optimized adversarial patterns to an online umbrella-printing service. We record 10-second videos at varying distances using a 4K resolution smartphone camera (iPhone 16). This results in around 600 frames with a resolution of 3840 $\times$ 2160 for each video. Then, we run the SOT model offline and evaluate the shrink rate. We report the average shrink rate recorded at different distances. We use umbrellas printed with FlyTrap patterns optimized against MixFormer and SiamRPN-ResNet, both with and without PDP design. For both, the target shrink rate in the objective function (in Section~\ref{sec:adv-obj-func}) is set to 0 to study the effects of closed-loop distance-pulling effects.

\subsubsection{Open-Loop Experiment Results} 
\label{sec:phy-open-loop}

We study the average shrink rates at different distances, shown in Fig.~\ref{fig:distance_evaluation}. In the short-range evaluation, for both SiamRPN and MixFormer, the average shrink rate of FlyTrap$_\text{PDP}$ decreases as the distance decreases since the higher resolution on the adversarial pattern can cause an even lower shrink rate. On the other hand, the shrink rate of the vanilla FlyTrap remains almost the same since it's locked onto one fixed area in the pattern (e.g., the human shape area in Fig.~\ref{fig:adv-pattern} (a)). The results suggest that PDP enables finer optimization at higher resolution, resulting in a significantly smaller tracked area. We further evaluate at long-range distances to study the maximum working distance for the FlyTrap attack. Optimized using collected videos under 20 meters, we find FlyTrap works well at that distance and can potentially extend to 30 or 40 meters. The distance is beyond the maximum functional distance of available ATT drones with 4K videos, which are typically around 20 meters~\cite{dji_active_tracking_2023, skydio_update_2020, potensic_atom_tracking_2023}. It should be noted that the training dataset we collect mainly includes close-range footage. Thus, FlyTrap$_\text{PDP}$ might not be well-optimized for the long-range attack case. We leave it as future work to further study the long-range attack capabilities. More discussions are in the Appendix~\ref{app:distance-shrink-discussion}.

\begin{figure}[t!]
    \centering
    \includegraphics[width=\linewidth]{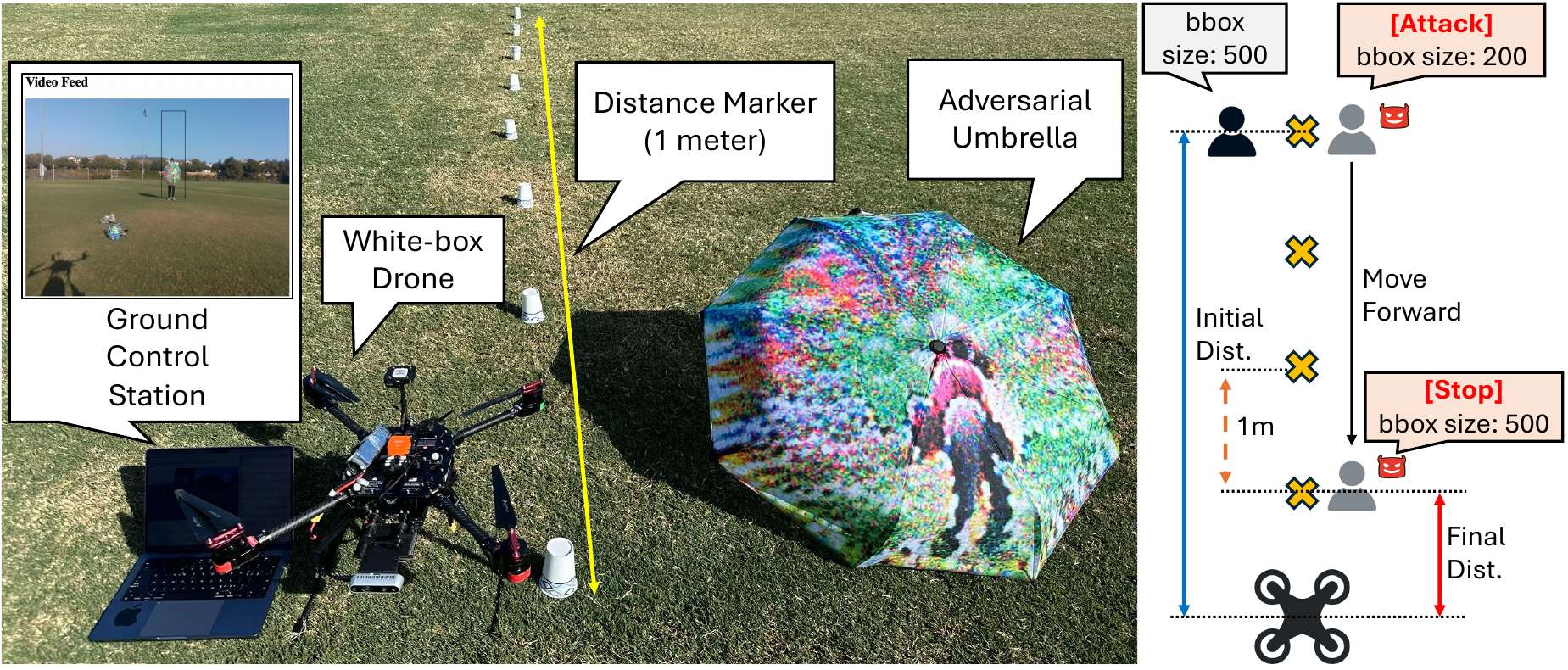}
    \caption{Physical closed-loop evaluation setups. An operator carries the drone to move forward until the shrunk box size matches the original size.}
    \label{fig:physical-setups}
    \vspace{-0.3cm}
\end{figure}

\subsubsection{Closed-Loop Evaluation Setups}
\label{sec:phy-eval-setups} 
To evaluate our attack in physical, closed-loop setups, we build our experimental platform using Holybro X500 v2 drone~\cite{holybro_px4_development_kit}, a medium-lift quadcopter powered by a Pixhawk flight controller. The system uses Robot Operating System (ROS)~\cite{quigley2009ros} as the communication backbone. We detail the implementation in Appendix~\ref{app:att-implement}.

\subsubsection{Closed-loop Evaluation Methodology}
We conduct experiments with the same location and target person as our training videos (Section~\ref{sec:digital-eval}). For safety and distance measurement issues, we simulate ATT behavior without a physical takeoff by manually maneuvering the drone on the ground. The experimenter, who acts as the attacker, then initiates the FlyTrap attack and manually moves
forward
until the tracking box size matches the initialization (shown in Fig.~\ref{fig:physical-setups}), which closely approximates the outcome of an autonomous closed-loop flight. We test all four models with FlyTrap initialized from 7 various starting distances from 6 to 12 meters. For closed-loop attack success criterion $d_a$ (Section~\ref{sec:metrics}), we choose each of the three attack goals in Section~\ref{sec:problem}: \textit{A1}: $d_{a}=9$ for net gun capturing~\cite{netgun2025, ultranetHD2025, netgunstore2025}; \textit{A2}: $d_{a}=6$ for binocular camera projection attack~\cite{zhou2022doublestar}, as it directly related to drone sensor spoofing attacks with range limits; and \textit{A3}: $d_a=0.5$ for crashing, as it within human arm length to push the umbrella. Finally, we report the resulting ASR$_{\text{closed}}$.

\begin{figure}[t!]
    \centering
    \includegraphics[width=\linewidth]{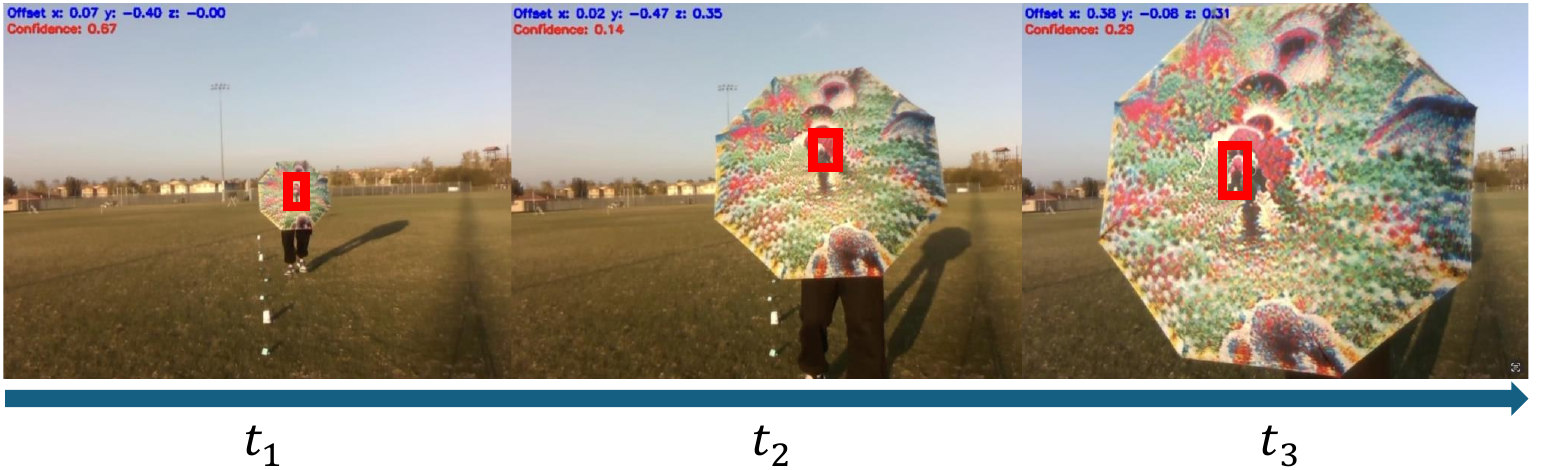}
    \caption{Physical evaluation of FlyTrap$_\text{PDP}$ against MixFormer~\cite{cui2022mixformer}. The video is captured by an on-drone RealSense camera with a resolution of 640 $\times$ 480. The shrink rate decreases largely as the distance decreases, enabling the progressive distance-pulling effects. Bold the box for clarity.}
    \label{fig:continuous_shrink}
\end{figure}

\subsubsection{Closed-loop Experiment Results}
\label{sec:phy-close-loop}
\begin{table}[t!]
\centering
\caption{White-box physical closed-loop evaluation results of ASR$_{\text{closed}}$ against different models under different attack goals. \textit{w/ PDP} means using PDP during attack optimization.}
\resizebox{\linewidth}{!}{%
    \begin{tabular}{l|c|c|c}
    \toprule
    {\textbf{Victim Model}} & \textbf{Capture ($9~m$)} & \textbf{DoubleStar~($6~m$)} & \textbf{Crash~($0.5~m$)} \\
    \midrule
    MixFormer & $100.0\%$ & $100.0\%$ & $0.0\%$ \\
    \cellcolor{gray!20}MixFormer w/ PDP & \cellcolor{gray!20}$100.0\%$ & \cellcolor{gray!20}$100.0\%$ & \cellcolor{gray!20}$\mathbf{100.0\%}$ \\
    \midrule
    Siam-Alex & $100.0\%$ & $100.0\%$ & $0.0\%$ \\
    Siam-Res & $100.0\%$ & $100.0\%$ & $0.0\%$ \\
    \cellcolor{gray!20}Siam-Res w/ PDP & \cellcolor{gray!20}$100.0\%$ & \cellcolor{gray!20}$100.0\%$ & \cellcolor{gray!20}$\mathbf{100.0\%}$ \\
    Siam-Mob & $100.0\%$ & $85.7\%$  & $0.0\%$ \\
    \bottomrule
    \end{tabular}%
}
\label{tab:physical-res}
\end{table}

As shown in Table~\ref{tab:physical-res}, FlyTrap$_\text{PDP}$ substantially reduces the tracking distance of the ATT system to be within the range of capturing, sensor attacks, or direct crashes. These results confirm the physical feasibility of our threat model and demonstrate that the proposed PDP can significantly affect the tracking distance of autonomous tracking drones in closed-loop control. As shown in Fig.~\ref{fig:continuous_shrink}, since we use RealSense camera~\cite{intel_realsense_d435i}, which has lower resolution of 640 $\times$ 480 compared to the videos captured by the smartphone in Section~\ref{sec:phy-open-loop-setup}, the phenomenon that the shrink rate decreases as the distance decreases is even more obvious compared to those in Fig.~\ref{fig:shrink_rate_distance}. The results highlight the physical-world impact and FlyTrap$_\text{PDP}$'s capability to progressively pull the drone as it approaches.

\begin{figure}[t!]
    \centering
    \includegraphics[width=\linewidth]{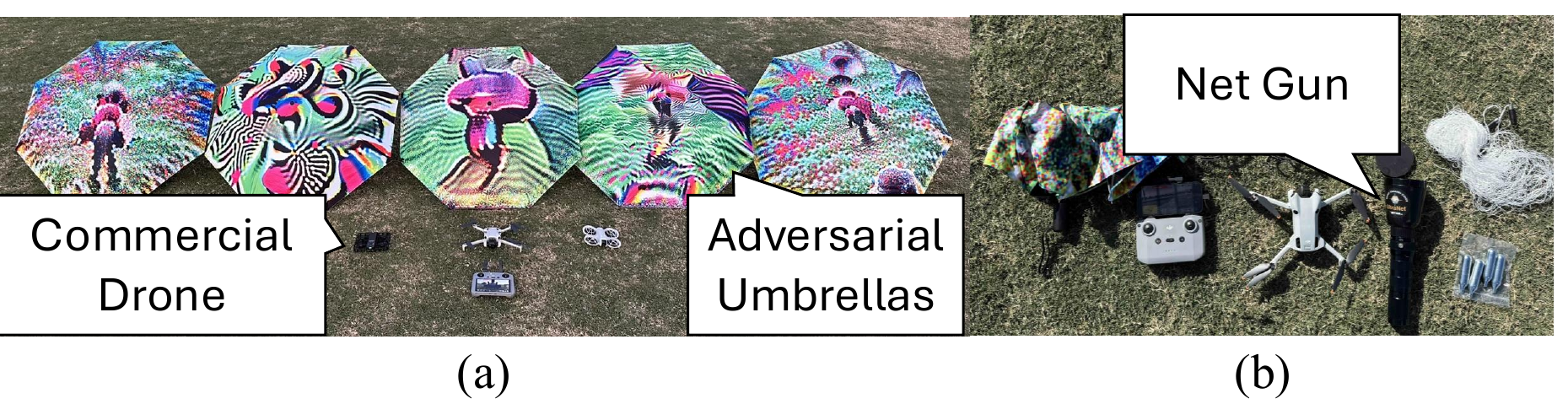}
    \vspace{-0.3cm}
    \caption{Commercial evaluation. (a) Crafted FlyTrap physical umbrellas and three commercial drones. (b) Net gun equipment set used for an end-to-end attack demonstration.}
    \label{fig:commercial-setups}
    \vspace{-0.3cm}
\end{figure}

\begin{figure*}
    \centering
    \includegraphics[width=\linewidth]{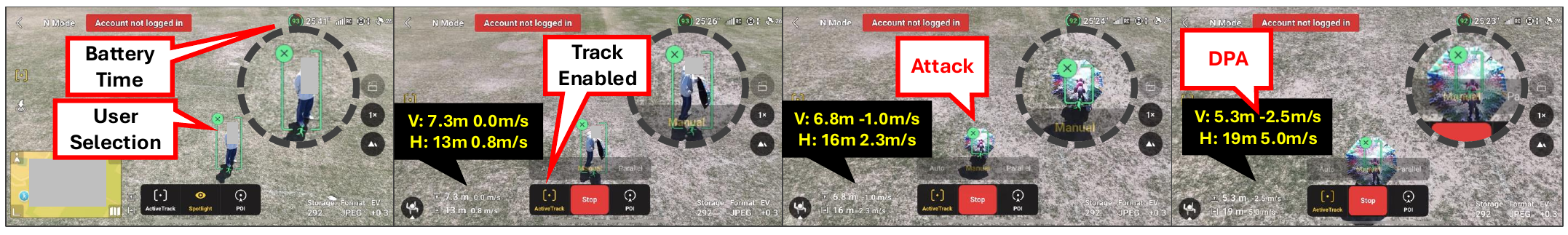}
    \caption{Screenshots from the DJI Mini 4 Pro remote controller interface during a real-world FlyTrap DPA. The green tracking box indicates the output of DJI’s built-in tracking model (zoomed-in view in the figure with a grey dotted border for clarity). In the third frame, the tracker erroneously locks onto a small subregion within the adversarial umbrella pattern. We highlight the vertical and horizontal position and their corresponding velocity of the drone. In the fourth frame, the altitude decreases to 5.3 meters at -2.5 m/s while the drone is moving forward at 5 m/s, suggesting the drone rapidly descends toward the attacker.}
    \label{fig:dji_ui}
    \vspace{-0.3cm}
\end{figure*}

\subsection{Commercial System Attack Evaluation}
\label{sec:commercial}

To assess DPA’s real-world viability and highlight potential vulnerabilities in widely accessible commercial products, we conduct a black-box evaluation on three consumer-grade drones equipped with visual tracking systems. As noticed in~\cite{wang2024revisiting}, the obscure implementations in commercial systems can heavily undermine the phenomenon observed in academia prototypes. Therefore, we evaluate whether our proposed DPA can be conducted in the physical world against real-world products instead of testing our PDP and ATG design, which is specifically for white-box systems (Section~\ref{sec:threat-model}).

\subsubsection{Evaluation Setups}
\label{sec:commercial-setup}
We acquire three commercially available drone products, including DJI Mini 4 Pro~\cite{dji2024}, DJI Neo~\cite{dji2025}, and HoverAir~\cite{hoverair2025}. We employ the same umbrella design described in Section~\ref{sec:phy-eval-setups}, shown in Fig.~\ref{fig:commercial-setups} (a). We purchase a net gun~\cite{ultranetHD2025} to demonstrate the DPA-enabled drone capturing attacks, shown in Fig.~\ref{fig:commercial-setups} (b).

\subsubsection{Evaluation Methodology}
We place ground markers at fixed intervals ($\Delta d$) and have an observer walk parallel to the drone’s flight path to estimate displacement. The flying altitude $h$ is retrieved from the flight log or viewed directly on the controller interface. The distance between the drone and the person is derived by $d=\sqrt{h^2 + (n\Delta d)^2}$. We record the final distance under attacks and report the resulting ASR$_{\text{closed}}$. Each drone is tested in 10 separate flights with fixed initial distances. For the DJI Mini 4 Pro, we set the initial distance as 12 meters, and for the DJI Neo and HoverAir drones, the tracking distance is factory set to around 2 meters. For each flight, we begin by powering up the drone, taking off, and then activating the tracking. We don't report capturing and DoubleStar~\cite{zhou2022doublestar} attack for DJI Neo and HoverAir since their preset initial tracking distance is already within those ranges.

\begin{table}[t!]
\centering
\caption{Commercial evaluation of ASR$_\text{closed}$. Distance limitation for each attack is used as the threshold for ASR$_\text{closed}$. A human-shaped FlyTrap umbrella successfully deceives three consumer drones, with crash outcomes on DJI Neo and HoverAir. N/A means the drones' preset initial tracking distance is already within the distance. These findings demonstrate the real-world existence of the ATT vulnerabilities we identified.}
\resizebox{\linewidth}{!}{%
    \begin{tabular}{l|ccc}
    \toprule
    \textbf{Attacks} & \textbf{DJI Mini 4 Pro} & \textbf{DJI Neo} & \textbf{HoverAir}  \\ 
    \midrule
    Capturing~($9~m$) & \cellcolor{gray!20}$60.0\%$ & N/A & N/A  \\
    DoubleStar~($6~m$) & \cellcolor{gray!20}$30.0\%$ & N/A & N/A \\
    Crash ($0.5~m$) & \cellcolor{gray!20}$0.0\%$ & \cellcolor{gray!20}$60.0\%$ & \cellcolor{gray!20}$80.0\%$ \\
    \bottomrule
    \end{tabular}%
    }
    \label{tab:commercial-masr}
    \vspace{-0.2cm}
\end{table}

\begin{figure}[th]
    \centering
    \includegraphics[width=\linewidth]{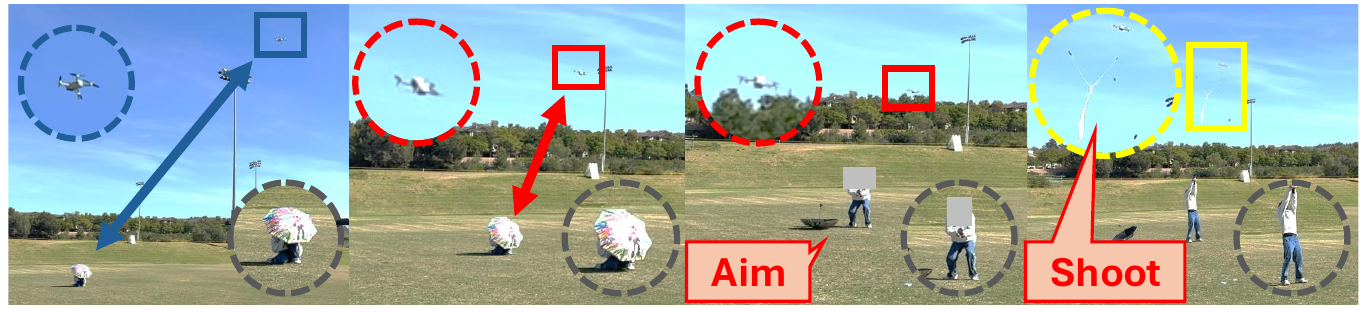}
    \caption{End-to-end FlyTrap-enabled DPA demonstration: \textit{A1: drone capturing}. We use a net gun to shoot the DJI Mini 4 Pro with \textit{Active Tracking} feature when the tracking distance is pulling closer to the attacker. Zoomed-in view in the figure with dotted borders for clarity.}
    \label{fig:commercial-demo2}
    \vspace{-0.3cm}
\end{figure}

\subsubsection{Experiment Results}
We find that one of the umbrellas, i.e., adversarial pattern against MixFormer shown in Fig.~\ref{fig:adv-pattern} (a), can successfully deceive all three tested drones. As shown in Table~\ref{tab:commercial-masr}, we find that with FlyTrap, we can capture or launch sensor attacks on the DJI Mini 4 Pro with a success rate of 60\% and 30\%, respectively. In Fig.~\ref{fig:dji_ui}, we record the screen of the DJI remote controller to verify that the distance-pulling is indeed caused by a shrunk box instead of other factors. In Fig.~\ref{fig:commercial-demo2}, we show an end-to-end DPA-enabled drone-capturing attack. After conducting the DPA, the tracking distance is shortened largely, allowing the attacker to aim and then shoot the capturing net against the drone. For DJI Neo and HoverAir, we observe that these two ultra-light drones are easy to crash, with a success rate of 60\% and 80\%, respectively. We suspect they lack the tracking state-verification mechanisms found in more advanced models like the DJI Mini 4 Pro~\cite{dji_activetrack_tracking_state}. Thus, they try to catch up with the ``shrunk'' object at a relatively high speed, causing the observed collision. Fig.~\ref{fig:commercial-demo} showcases an end-to-end DPA-enabled crashing attack on the HoverAir drone. The attacker can physically hit the drone using the umbrella to cause a collision and crash. We also empirically find that crouching down to cover the whole body can increase the attack success rate, as shown in Fig.~\ref{fig:dji_ui} and \ref{fig:commercial-demo2}.

The failure of other umbrellas might be due to the white-box model, which they are optimized against, being convolutional-based models. Ma et al. find that adversarial examples generated using Transformer-based models tend to be more transferable than CNN-based models~\cite{ma2023transferable}. We also find that adversarial patches against SiamRPN, which is a CNN-based SOT model, show more model-specific patterns (second to fourth umbrellas in Fig.~\ref{fig:commercial-setups} (a)) compared to the human-shape pattern from MixFormer (umbrella in Fig.~\ref{fig:adv-pattern}). Further discussion can be seen in Appendix~\ref{app:commercial-justify}.

Additionally, we find the black-box commercial results consistent with transferability experiments (Section~\ref{sec:transfer}): the pattern with a human-like shape, with the highest transferability among others (18.2\%), can potentially transfer to commercial drones. Despite only one of our crafted umbrellas succeeding, we reveal the first demonstration of the proposed DPA vulnerabilities in deployed commercial drone systems, thus justifying the security problem we identified, and demonstrate two out of three use cases of DPA (\textit{A1} and \textit{A3} in Section~\ref{sec:problem}). We have already performed responsible vulnerability disclosure to the corresponding manufacturers (Section~\ref{sec:ethics}).

\begin{figure}[t]
    \centering
    \includegraphics[width=\linewidth]{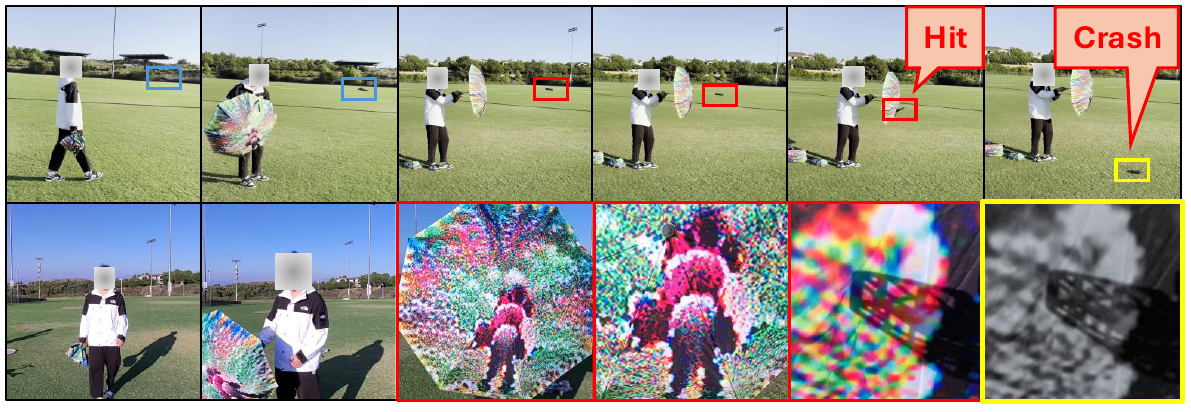}
    \caption{End-to-end FlyTrap-enabled DPA demonstration: \textit{A3: drone crashing}. We use the umbrella to hit the HoverAir drone with \textit{Dolly Track} feature when the tracking distance is pulled within the hitting distance to the attacker.}
    \label{fig:commercial-demo}
    \vspace{-0.5cm}
\end{figure}

\subsection{Attack Stealthiness Evaluation}

The proposed FlyTrap attack is stealthy since the umbrella is folded for most of the carrying time. In this section, we conduct a user study to further justify the stealthiness during the deployment time.

\subsubsection{Evaluation Setups}

The survey consists of two parts. In the first part, we focus on the attack vector, the umbrella: we investigate if it's uncommon to use an umbrella on a non-rainy day. Next, given the adversarial pattern, we investigate if people feel suspicious about it. We released the survey on the Prolific platform~\cite{prolific2025} with 200 participants sampled to reflect a broad demographic distribution in the United States. For the user study, we go through the IRB review process of our institution and receive confirmation of the self-exempt subject search categorization. More details can be found in the ethics discussion in Section~\ref{sec:ethics}.

\subsubsection{Experiment Results}

In terms of the usage of umbrellas, we find 78.6\% of the participants think it's not abnormal to them when seeing someone using an umbrella on a non-rainy day. In terms of the adversarial patterns, some portion of the participants might think our FlyTrap pattern is eye-catching ($\sim$13\%), and rank 3rd among all the candidate umbrellas. However, only 5.9\% of participants found the FlyTrap pattern suspicious, lower than the percentage of 26.7\% who selected ``\textit{none of the umbrellas are suspicious}''. The results suggest the FlyTrap is a deployable attack in the real world without being considered suspicious by the general public, even during the launch time. Please refer to Appendix~\ref{sec:user-study} for full results.
\section{Discussion and Limitations}
\label{sec:defense}

\subsection{Other countermeasures}

In addition to spatial-temporal consistency checking, there are model-level defense strategies such as certified robustness~\cite{levine2020randomized, xiang2021detectorguard, xiang2021patchguard, xiang2021patchguard++}, adversarial training~\cite{shafahi2019adversarial, goodfellow2014explaining, madry2017towards}, and input transformation~\cite{guo2017countering, xie2017mitigating}. Certified defenses provide theoretical guarantees of model robustness against white-box attacks, typically by ensuring bounded prediction error in the presence of adversarial perturbations. However, existing certified robustness techniques and adversarial training methods primarily target misclassification attacks, making their direct application to SOT models non-trivial due to the fundamentally different attack objectives and model behaviors. To the best of our knowledge, there is no specific certified robustness or adversarial training for SOT defense. Input transformation-based defenses, in contrast, are ineffective under strong Expectation over Transformation (EoT) optimization~\cite{athalye2018synthesizing}. Jia et al.~\cite{jia2020robust} proposed a defense specifically tailored for SOT models. However, their method targets pixel-level perturbations and relies on run-time gradient-based iterative denoising. The gradient back-propagation process is well-known to be much slower than the inference, which is computationally intensive and thus impractical for real-time ATT drone applications~\cite{medium2023}. Specifically, ATT drones are more likely to lose the target if the inference speed is limited, given the temporal dependency of their working mechanisms. Therefore, future defense needs to improve the efficiency to support real-time ATT applications.

\subsection{Limitations}
\label{sec:limitation}

We acknowledge the limitation of FlyTrap's black-box transferability. Our approach is mainly designed to achieve better attack effects for white-box ATT systems, and thus might sacrifice the black-box transferability. However, it should not be the major concern given (1) the existing advancement in reverse engineering (in Section~\ref{sec:threat-model}) and (2) the existing approach to boost black-box transferability in attacking SOT~\cite{yin2024dimba, jia2021iou}. They are applicable to FlyTrap by replacing the gradient-based optimization with their black-box searching, which is out of the scope of this work. Furthermore, the black-box commercial system implementation can influence how the commercial systems react towards the attack~\cite{wang2024revisiting}. Nonetheless, we've shown the real-world existence of the vulnerabilities by exploiting the proposed DPA to capture or crash the drone (Section~\ref{sec:commercial}). We leave it as future work to improve the transferability and understand the commercial ATT systems to further investigate the real-world security problems.
\section{Conclusion}
In this paper, we conduct the first systematic study on the security of camera-based Autonomous Target Tracking (ATT) systems, with a focus on the newly proposed physical-world distance-pulling attacks (DPA). We define the problem with domain-specific goals and introduce the adversarial umbrella as a novel, real-world deployable attack vector. By designing a progressive distance-pulling strategy and controllable spatial-temporal consistency, we achieve closed-loop distance-pulling effects and spatial-temporal consistent attacks. Through a new dataset and system-level metrics, we demonstrate the attack’s high effectiveness and generalizability. Physical-world evaluations with adversarial umbrella prototypes and full-stack ATT drones, alongside black-box testing on commercial drones, reveal significant real-world implications. Given the critical security and safety concerns surrounding ATT, we hope our findings will inspire future research and community attention.

\section{Ethics Considerations}
\label{sec:ethics}

\noindent
\textbf{Disclosures.} We evaluated our attack on three commercial drones and confirmed that their ATT features are all vulnerable to the FlyTrap attack, sometimes causing high-speed collisions. To prevent negative impacts, we have performed responsible vulnerability disclosure to both affected manufacturers prior to all public releases of this work, following the ethical standards in the security community.

\noindent
\textbf{Data collection.} We collected videos of researchers (with obscured facial features) to evaluate ML models for tracking and detection. No identifiable private information was involved, and participants could not be identified. As confirmed by our IRB officers, this anonymized setup does not qualify as human subject research under federal policy~\cite{federal_register_2017_p203}.

\noindent
\textbf{User study.} We conducted anonymous visual recognition tasks on Prolific. According to the Federal Policy~\cite{federal_register_2017_p1315}, this qualifies for exempt Category-2~\cite{federal_register_2017_p1375} since no identifiable information was collected and participants faced no physical or psychological risks. Thus, IRB review was not required.

\section*{Acknowledgment}

We would like to thank Chi Zhang, Li-Chen Cheng, and the anonymous reviewers for their valuable and insightful feedback. This work was supported by the National Science Foundation under grant CNS-2145493 and by the NASA University Leadership Initiative under Award 80NSSC24M0070. All drone data and experiments presented in this work were completed before December 22, 2025.

\bibliographystyle{IEEEtranS}
\bibliography{reference}

\begin{thebibliography}{100}
\providecommand{\url}[1]{#1}
\csname url@samestyle\endcsname
\providecommand{\newblock}{\relax}
\providecommand{\bibinfo}[2]{#2}
\providecommand{\BIBentrySTDinterwordspacing}{\spaceskip=0pt\relax}
\providecommand{\BIBentryALTinterwordstretchfactor}{4}
\providecommand{\BIBentryALTinterwordspacing}{\spaceskip=\fontdimen2\font plus
\BIBentryALTinterwordstretchfactor\fontdimen3\font minus \fontdimen4\font\relax}
\providecommand{\BIBforeignlanguage}[2]{{%
\expandafter\ifx\csname l@#1\endcsname\relax
\typeout{** WARNING: IEEEtranS.bst: No hyphenation pattern has been}%
\typeout{** loaded for the language `#1'. Using the pattern for}%
\typeout{** the default language instead.}%
\else
\language=\csname l@#1\endcsname
\fi
#2}}
\providecommand{\BIBdecl}{\relax}
\BIBdecl

\bibitem{illinois_drone}
{Andrew Adams}, ``{Illinois Expands Use of Police Surveillance Drones},'' \url{https://capitolnewsillinois.com/news/illinois-expands-use-of-police-surveillance-drones}, 2023.

\bibitem{athalye2018synthesizing}
A.~Athalye, L.~Engstrom, A.~Ilyas, and K.~Kwok, ``Synthesizing robust adversarial examples,'' in \emph{International Conference on Machine Learning (ICML)}, 2018.

\bibitem{dynamictrack}
{AUTEL}, ``{What Is Autel {{EVO 2}} Dynamic Track 2.0?}'' \url{https://www.autelpilot.com/blogs/news/what-is-autel-evo-2-dynamic-track-2-0}, 2021.

\bibitem{autel_evo_lite_640t_enterprise_2025}
{Autel Robotics (via AutelPilot)}, ``{Autel EVO Lite 640T Enterprise Product Page},'' \url{https://www.autelpilot.com/collections/evo-lite-enterprise/products/autel-evo-lite-640t-enterprise}, 2025.

\bibitem{autel_dynamic_track_2024}
{AutelPilots Forum}, ``{Dynamic Track Distance Limitations},'' \url{https://autelpilots.com/threads/dynamic-track-distance.7954/}, 2024.

\bibitem{bertinetto2016fully}
L.~Bertinetto, J.~Valmadre, J.~F. Henriques, A.~Vedaldi, and P.~H. Torr, ``{Fully-Convolutional Siamese Networks for Object Tracking},'' in \emph{European Conference on Computer Vision (ECCV)}, 2016.

\bibitem{activetrack}
{Billy Kyle}, ``{DJI Mini 4 Pro {{ActiveTrack 360º}} Full Tutorial - A Brand New Experience},'' \url{https://www.youtube.com/watch?v=xM_d0FBnweY}, 2024.

\bibitem{box2015time}
G.~E. Box, G.~M. Jenkins, G.~C. Reinsel, and G.~M. Ljung, \emph{{Time Series Analysis: Forecasting and Control}}.\hskip 1em plus 0.5em minus 0.4em\relax John Wiley \& Sons, 2015.

\bibitem{cao2021invisible}
Y.~Cao, N.~Wang, C.~Xiao, D.~Yang, J.~Fang, R.~Yang, Q.~A. Chen, M.~Liu, and B.~Li, ``{Invisible for Both Camera and {{LiDAR}}: Security of Multi-Sensor Fusion Based Perception in Autonomous Driving Under Physical-World Attacks},'' in \emph{IEEE Symposium on Security and Privacy ({{SP}})}.\hskip 1em plus 0.5em minus 0.4em\relax IEEE, 2021.

\bibitem{cao2019adversarial}
Y.~Cao, C.~Xiao, B.~Cyr, Y.~Zhou, W.~Park, S.~Rampazzi, Q.~A. Chen, K.~Fu, and Z.~M. Mao, ``{Adversarial Sensor Attack on LiDAR-Based Perception in Autonomous Driving},'' in \emph{ACM SIGSAC Conference on Computer and Communications Security (CCS)}, 2019.

\bibitem{7502612}
A.~Chakrabarty, R.~Morris, X.~Bouyssounouse, and R.~Hunt, ``{Autonomous Indoor Object Tracking with the {{Parrot AR.Drone}}},'' in \emph{International Conference on Unmanned Aircraft Systems ({{ICUAS}})}, 2016.

\bibitem{chen2021unified}
X.~Chen, C.~Fu, F.~Zheng, Y.~Zhao, H.~Li, P.~Luo, and G.-J. Qi, ``{A Unified Multi-Scenario Attacking Network for Visual Object Tracking},'' in \emph{Proceedings of the AAAI Conference on Artificial Intelligence}, 2021.

\bibitem{chen2020one}
X.~Chen, X.~Yan, F.~Zheng, Y.~Jiang, S.-T. Xia, Y.~Zhao, and R.~Ji, ``{One-Shot Adversarial Attacks on Visual Tracking with Dual Attention},'' in \emph{Proceedings of the {{IEEE/CVF}} Conference on Computer Vision and Pattern Recognition (CVPR)}, 2020.

\bibitem{chen2022anti}
Y.~Chen, Z.~Li, L.~Li, S.~Ma, F.~Zhang, and C.~Fan, ``{An Anti-Drone Device based on Capture Technology},'' \emph{Biomimetic Intelligence and Robotics}, 2022.

\bibitem{cheng2017autonomous}
H.~Cheng, L.~Lin, Z.~Zheng, Y.~Guan, and Z.~Liu, ``{An Autonomous Vision-Based Target Tracking System for Rotorcraft Unmanned Aerial Vehicles},'' in \emph{IEEE/RSJ International Conference on Intelligent Robots and Systems (IROS)}.\hskip 1em plus 0.5em minus 0.4em\relax IEEE, 2017.

\bibitem{chuang2019autonomous}
H.-M. Chuang, D.~He, and A.~Namiki, ``{Autonomous Target Tracking of UAV Using High-Speed Visual Feedback},'' \emph{Applied Sciences}, 2019.

\bibitem{mavros}
M.~Contributors, ``{MAVROS: MAVLink Extendable Communication Node for ROS},'' {https://github.com/mavlink/mavros}, 2014.

\bibitem{mmcv}
M.~Contributors, ``{MMCV: OpenMMLab Computer Vision Foundation},'' \url{https://github.com/open-mmlab/mmcv}, 2018.

\bibitem{cui2022mixformer}
Y.~Cui, C.~Jiang, L.~Wang, and G.~Wu, ``{MixFormer: End-to-End Tracking with Iterative Mixed Attention},'' in \emph{Proceedings of the {{IEEE/CVF}} Conference on Computer Vision and Pattern Recognition (CVPR)}, 2022.

\bibitem{forbes_ukraine}
{David Hambling}, ``{Forbes: Ukraine Rolls Out Target-Seeking Terminator Drones},'' \url{https://www.forbes.com/sites/davidhambling/2024/03/21/ukraine-rolls-out-target-seeking-terminator-drones}.

\bibitem{net_capture}
{Defense Central}, ``{Capturing Enemy Drones Using Nets?}'' \url{https://www.youtube.com/watch?v=fDftjuGM5mM}, 2024.

\bibitem{ding2021towards}
L.~Ding, Y.~Wang, K.~Yuan, M.~Jiang, P.~Wang, H.~Huang, and Z.~J. Wang, ``{Towards Universal Physical Attacks on Single Object Tracking},'' in \emph{Proceedings of the {{AAAI}} Conference on Artificial Intelligence}, 2021.

\bibitem{djimavic2023}
DJI, ``{DJI Mavic 3 Pro},'' \url{https://store.dji.com/product/dji-mavic-3-pro?vid=137691}, 2023.

\bibitem{dji2024}
DJI, ``{DJI Mini 4 Pro},'' \url{https://www.dji.com/mini-4-pro?site=brandsite&from=homepage}, 2024.

\bibitem{dji_activetrack_tracking_state}
{DJI}, ``{DJI Mobile SDK: DJIActiveTrackTrackingState Class Reference},'' \url{https://developer.dji.com/api-reference/android-api/Components/Missions/DJIActiveTrackTrackingState.html}, 2025.

\bibitem{dji2025}
DJI, ``{DJI Neo},'' \url{https://www.dji.com/neo}, 2025.

\bibitem{dji_mini4pro_specs_2025}
{DJI}, ``{Mini 4 Pro Specifications},'' \url{https://www.dji.com/mini-4-pro/specs}, 2025.

\bibitem{dji_active_tracking_2023}
{DJI Forum}, ``{DJI Active Tracking Distance Restrictions},'' \url{https://forum.dji.com/forum.php?mod=redirect&goto=findpost&ptid=284247&pid=2972409}, 2023.

\bibitem{skydio_update_2020}
{DroneXL}, ``{Skydio 2 Software Update Offers Many New Features},'' \url{https://dronexl.co/2020/06/25/skydio-2-software-update#:~:text=Maximum%20tracking%20distance%20with%20the%20phone%20has%20increased%20from%2010%20to%2020%20meters%20(66%20feet)}, 2020.

\bibitem{everingham2010pascal}
M.~Everingham, L.~Van~Gool, C.~K.~I. Williams, J.~Winn, and A.~Zisserman, ``{The {{PASCAL}} Visual Object Classes ({{VOC}}) Challenge},'' \emph{International Journal of Computer Vision}, 2010.

\bibitem{copterexpress_clover}
C.~Express, ``{Clover: ROS-Based Framework and Raspberry Pi Image to Control PX4-Powered Drones},'' \url{https://github.com/CopterExpress/clover}, 2025.

\bibitem{fang2022alphapose}
H.-S. Fang, J.~Li, H.~Tang, C.~Xu, H.~Zhu, Y.~Xiu, Y.-L. Li, and C.~Lu, ``{AlphaPose: Whole-Body Regional Multi-Person Pose Estimation and Tracking in Real-Time},'' \emph{IEEE Transactions on Pattern Analysis and Machine Intelligence}, 2022.

\bibitem{fradi2018autonomous}
H.~Fradi, L.~Bracco, F.~Canino, and J.-L. Dugelay, ``{Autonomous Person Detection and Tracking Framework Using Unmanned Aerial Vehicles ({{UAVs}})},'' in \emph{European Signal Processing Conference ({{EUSIPCO}})}.\hskip 1em plus 0.5em minus 0.4em\relax IEEE, 2018.

\bibitem{fu2022ad}
C.~Fu, S.~Li, X.~Yuan, J.~Ye, Z.~Cao, and F.~Ding, ``{AD\textsuperscript{2} Attack: Adaptive Adversarial Attack on Real-Time {{UAV}} Tracking},'' in \emph{International Conference on Robotics and Automation ({{ICRA}})}.\hskip 1em plus 0.5em minus 0.4em\relax IEEE, 2022.

\bibitem{garcia2020autonomous}
M.~Garc{\'\i}a, R.~Caballero, F.~Gonz{\'a}lez, A.~Viguria, and A.~Ollero, ``{Autonomous Drone with Ability to Track and Capture an Aerial Target},'' in \emph{International Conference on Unmanned Aircraft Systems (ICUAS)}.\hskip 1em plus 0.5em minus 0.4em\relax IEEE, 2020.

\bibitem{goodfellow2014explaining}
I.~J. Goodfellow, J.~Shlens, and C.~Szegedy, ``{Explaining and Harnessing Adversarial Examples},'' \emph{International Conference on Learning Representations (ICLR)}, 2015.

\bibitem{guo2017countering}
C.~Guo, M.~Rana, M.~Cisse, and L.~Van Der~Maaten, ``{Countering Adversarial Images Using Input Transformations},'' \emph{International Conference on Learning Representations (ICLR)}, 2018.

\bibitem{visionguard_code2024}
X.~Han, H.~Wang, K.~Zhao, G.~Deng, Y.~Xu, H.~Liu, H.~Qiu, and T.~Zhang, ``{VisionGuard: Official Code Repository},'' \url{https://zenodo.org/records/11140958}, 2024.

\bibitem{han2024visionguard}
X.~Han, H.~Wang, K.~Zhao, G.~Deng, Y.~Xu, H.~Liu, H.~Qiu, and T.~Zhang, ``{VisionGuard: Secure and Robust Visual Perception of Autonomous Vehicles in Practice},'' in \emph{{{ACM SIGSAC}} Conference on Computer and Communications Security}, 2024.

\bibitem{he2016deep}
K.~He, X.~Zhang, S.~Ren, and J.~Sun, ``{Deep Residual Learning for Image Recognition},'' in \emph{Proceedings of the {{IEEE}} Conference on Computer Vision and Pattern Recognition (CVPR)}, 2016.

\bibitem{hochreiter1997long}
S.~Hochreiter, ``{Long Short-Term Memory},'' \emph{Neural Computation {{MIT}}-Press}, 1997.

\bibitem{holybro_px4_development_kit}
{Holybro}, ``{PX4 Development Kit X500 V2},'' \url{https://holybro.com/collections/x500-kits/products/px4-development-kit-x500-v2}, 2025.

\bibitem{hoverair2025}
HOVERAir, ``{HOVERAir Camera Drone Ultra-Light Self-Piloting},'' \url{https://www.amazon.com/HOVERAir-Camera-Drone-Ultra-Light-Self-Piloting/dp/B0DLGRP5PQ}, 2025.

\bibitem{howard2017mobilenets}
A.~G. Howard, ``{MobileNets: Efficient Convolutional Neural Networks for Mobile Vision Applications},'' \emph{arXiv preprint arXiv:1704.04861}, 2017.

\bibitem{canonsburg_drone_stalking}
S.~Ingram, ``{Canonsburg Man Charged with Stalking After Allegedly Using Drone to Follow Woman},'' \url{https://www.wtae.com/article/canonsburg-drone-stalking-washington-county/61100793}, 2024.

\bibitem{intel_realsense_d435i}
{Intel Corporation}, ``{Intel RealSense Depth Camera D435i},'' \url{https://www.intelrealsense.com/depth-camera-d435i/}, 2025.

\bibitem{ji2021poltergeist}
X.~Ji, Y.~Cheng, Y.~Zhang, K.~Wang, C.~Yan, W.~Xu, and K.~Fu, ``{Poltergeist: Acoustic Adversarial Machine Learning Against Cameras and Computer Vision},'' in \emph{2021 IEEE Symposium on Security and Privacy (SP)}.\hskip 1em plus 0.5em minus 0.4em\relax IEEE, 2021.

\bibitem{jia2020robust}
S.~Jia, C.~Ma, Y.~Song, and X.~Yang, ``{Robust Tracking against Adversarial Attacks},'' in \emph{European Conference on Computer Vision (ECCV)}.\hskip 1em plus 0.5em minus 0.4em\relax Springer, 2020.

\bibitem{jia2021iou}
S.~Jia, Y.~Song, C.~Ma, and X.~Yang, ``{IoU Attack: Towards Temporally Coherent Black-Box Adversarial Attack for Visual Object Tracking},'' in \emph{Proceedings of the IEEE/CVF Conference on Computer Vision and Pattern Recognition (CVPR)}, 2021.

\bibitem{jia2022fooling}
W.~Jia, Z.~Lu, H.~Zhang, Z.~Liu, J.~Wang, and G.~Qu, ``{Fooling the Eyes of Autonomous Vehicles: Robust Physical Adversarial Examples Against Traffic Sign Recognition Systems},'' \emph{Network and Distributed System Security (NDSS) Symposium}, 2022.

\bibitem{jiao2021end}
R.~Jiao, H.~Liang, T.~Sato, J.~Shen, Q.~A. Chen, and Q.~Zhu, ``{End-to-End Uncertainty-Based Mitigation of Adversarial Attacks to Automated Lane Centering},'' in \emph{IEEE Intelligent Vehicles Symposium ({{IV}})}.\hskip 1em plus 0.5em minus 0.4em\relax IEEE, 2021.

\bibitem{us-cbp_drone}
{John, Davis}, ``{U.S. Customs and Border Protection: Small but Mighty: Border Patrol’s use of small drones is a game changer in border security},'' \url{https://www.cbp.gov/frontline/cbp-small-drones-program}.

\bibitem{new_jersey_drone_incursions}
{Joseph Trevithick}, ``{New Jersey Base Confirms Multiple Past Drone Incursions… By Contraband Smugglers},'' \url{https://news.yahoo.com/news/jersey-confirms-multiple-past-drone-211029636.html}, 2024.

\bibitem{kendall2014on}
A.~G. Kendall, N.~N. Salvapantula, and K.~A. Stol, ``{On-Board Object Tracking Control of a Quadcopter with Monocular Vision},'' in \emph{International Conference on Unmanned Aircraft Systems ({{ICUAS}})}, 2014.

\bibitem{new_auto}
{Khari Johnson}, ``{This New Autonomous Drone for Cops Can Track You in the Dark},'' \url{https://www.wired.com/story/new-autonomous-drone-for-cops-can-track-you-in-the-dark/}, 2023.

\bibitem{koenig2004design}
N.~Koenig and A.~Howard, ``{Design and Use Paradigms for Gazebo, an Open-Source Multi-Robot Simulator},'' in \emph{IEEE/RSJ International Conference on Intelligent Robots and Systems (IROS)}.\hskip 1em plus 0.5em minus 0.4em\relax IEEE, 2004.

\bibitem{krizhevsky2012imagenet}
A.~Krizhevsky, I.~Sutskever, and G.~E. Hinton, ``{ImageNet Classification with Deep Convolutional Neural Networks},'' \emph{Advances in Neural Information Processing Systems (NeurIPS)}, 2012.

\bibitem{levine2020randomized}
A.~Levine and S.~Feizi, ``{(De)Randomized Smoothing for Certifiable Defense Against Patch Attacks},'' \emph{Advances in Neural Information Processing Systems (NeurIPS)}, 2020.

\bibitem{li2019siamrpn++}
B.~Li, W.~Wu, Q.~Wang, F.~Zhang, J.~Xing, and J.~Yan, ``{SiamRPN++: Evolution of Siamese Visual Tracking with Very Deep Networks},'' in \emph{Proceedings of the {{IEEE/CVF}} Conference on Computer Vision and Pattern Recognition (CVPR)}, 2019.

\bibitem{li2018high}
B.~Li, J.~Yan, W.~Wu, Z.~Zhu, and X.~Hu, ``{High Performance Visual Tracking with Siamese Region Proposal Network},'' in \emph{Proceedings of the {{IEEE}} Conference on Computer Vision and Pattern Recognition (CVPR)}, 2018.

\bibitem{liang2020efficient}
S.~Liang, X.~Wei, S.~Yao, and X.~Cao, ``{Efficient Adversarial Attacks for Visual Object Tracking},'' in \emph{European Conference on Computer Vision (ECCV)}.\hskip 1em plus 0.5em minus 0.4em\relax Springer, 2020.

\bibitem{lin2014microsoft}
T.-Y. Lin, M.~Maire, S.~Belongie, J.~Hays, P.~Perona, D.~Ramanan, P.~Doll{\'a}r, and C.~L. Zitnick, ``{Microsoft COCO: Common Objects in Context},'' in \emph{European Conference on Computer Vision (ECCV)}.\hskip 1em plus 0.5em minus 0.4em\relax Springer, 2014.

\bibitem{liu2016ssd}
W.~Liu, D.~Anguelov, D.~Erhan, C.~Szegedy, S.~Reed, C.-Y. Fu, and A.~C. Berg, ``{SSD: Single Shot Multibox Detector},'' in \emph{European Conference on Computer Vision (ECCV)}.\hskip 1em plus 0.5em minus 0.4em\relax Springer, 2016.

\bibitem{liu2016delving}
Y.~Liu, X.~Chen, C.~Liu, and D.~Song, ``{Delving into Transferable Adversarial Examples and Black-Box Attacks},'' \emph{International Conference on Learning Representations (ICLR)}, 2016.

\bibitem{liu2023decompiling}
Z.~Liu, Y.~Yuan, S.~Wang, X.~Xie, and L.~Ma, ``{Decompiling x86 Deep Neural Network Executables},'' in \emph{32nd USENIX Security Symposium (USENIX Security)}, 2023.

\bibitem{ma2023transferable}
W.~Ma, Y.~Li, X.~Jia, and W.~Xu, ``{Transferable Adversarial Attack for Both Vision Transformers and Convolutional Networks via Momentum Integrated Gradients},'' in \emph{Proceedings of the IEEE/CVF International Conference on Computer Vision (ICCV)}, 2023.

\bibitem{madry2017towards}
A.~Madry, A.~Makelov, L.~Schmidt, D.~Tsipras, and A.~Vladu, ``{Towards Deep Learning Models Resistant to Adversarial Attacks},'' \emph{arXiv Preprint {{arXiv:1706.06083}}}, 2017.

\bibitem{percepguard_code2023}
Y.~Man, R.~Muller, M.~Li, Z.~B. Celik, and R.~Gerdes, ``{PercepGuard: Official Code Repository},'' \url{https://github.com/Harry1993/PercepGuard}, 2023.

\bibitem{man2023person}
Y.~Man, R.~Muller, M.~Li, Z.~B. Celik, and R.~Gerdes, ``{That Person Moves Like A Car: Misclassification Attack Detection for Autonomous Systems Using Spatiotemporal Consistency},'' in \emph{32nd USENIX Security Symposium (USENIX Security)}, 2023.

\bibitem{reuters_ukraine}
{Mariano, Zafra and Max, Hunder and Anurag, Rao and Sudev, Kiyada}, ``{Reuters: How Drone Combat in Ukraine is Changing Warfare},'' \url{https://www.reuters.com/graphics/UKRAINE-CRISIS/DRONES/dwpkeyjwkpm/}.

\bibitem{medium2023}
{Marvelous Catawba Otter}, ``{A Brief Discussion: The Computational Cost of Backward Propagation Is Approximately Twice That of Forward Propagation},'' \url{https://medium.com/@marvelous_catawba_otter_200/a-brief-discussion-the-computational-cost-of-backward-propagation-is-approximately-twice-that-of-5dd0eac9b389}, 2023.

\bibitem{mavlink}
{MAVLink Contributors}, ``{MAVLink Micro Air Vehicle Communication Protocol},'' \url{https://mavlink.io/en/}, 2009.

\bibitem{muller2022physical}
R.~Muller, Y.~Man, Z.~B. Celik, M.~Li, and R.~Gerdes, ``{Physical Hijacking Attacks Against Object Trackers},'' in \emph{ACM SIGSAC Conference on Computer and Communications Security (CCS)}, 2022.

\bibitem{vogues_code2023}
R.~Muller, Y.~Man, M.~Li, R.~Gerdes, J.~Petit, and Z.~B. Celik, ``{VOGUES: Official Code Repository},'' \url{https://github.com/purseclab/VOGUES}, 2024.

\bibitem{muller2024vogues}
R.~Muller, Y.~Man, M.~Li, R.~Gerdes, J.~Petit, and Z.~B. Celik, ``{VOGUES: Validation of Object Guise using Estimated Components},'' in \emph{33rd USENIX Security Symposium (USENIX Security)}, 2024.

\bibitem{nakka2022universal}
K.~K. Nakka and M.~Salzmann, ``{Universal, Transferable Adversarial Perturbations for Visual Object Trackers},'' in \emph{European Conference on Computer Vision (ECCV)}.\hskip 1em plus 0.5em minus 0.4em\relax Springer, 2022.

\bibitem{netgun2025}
NetGun, ``{NetGun Official Website},'' \url{https://www.net-gun.com/}.

\bibitem{ultranetHD2025}
NetGun, ``{UltraNet HD -- Large Animal Target Net Gun},'' \url{https://netgun.com/netgun-info/ultranet-hd-large-animal-target-net-gun}.

\bibitem{nvidia2019jetson}
NVIDIA, ``{NVIDIA Jetson Nano},'' \url{https://www.nvidia.com/en-us/autonomous-machines/embedded-systems/jetson-nano/product-development}, 2019.

\bibitem{drone_forensics_paraben}
{Paraben Corporation}, ``{Drone Forensics: Navigating the New Frontier of Digital Evidence},'' \url{https://paraben.com/drone-forensics-navigating-the-new-frontier-of-digital-evidence/?utm_source=chatgpt.com}, 2024.

\bibitem{logan_airport_drone_incident}
{People Staff}, ``{Two Arrested for Flying Drones Dangerously Close to Logan Airport, Boston Police Report},'' \url{https://people.com/drones-two-arrested-dangerously-close-logan-airport-boston-police-8761977}, 2024.

\bibitem{pestana2013vision}
J.~Pestana, J.~L. Sanchez-Lopez, P.~Campoy, and S.~Saripalli, ``Vision-based {{GPS}}-denied object tracking and following for unmanned aerial vehicles,'' in \emph{IEEE International Symposium on Safety, Security, and Rescue Robotics ({{SSRR}})}.\hskip 1em plus 0.5em minus 0.4em\relax IEEE, 2013.

\bibitem{pestana2014computer}
J.~Pestana, J.~L. Sanchez-Lopez, S.~Saripalli, and P.~Campoy, ``{Computer Vision-Based General Object Following for GPS-Denied Multirotor Unmanned Vehicles},'' in \emph{American Control Conference}.\hskip 1em plus 0.5em minus 0.4em\relax IEEE, 2014.

\bibitem{petit2015remote}
J.~Petit, B.~Stottelaar, M.~Feiri, and F.~Kargl, ``{Remote Attacks on Automated Vehicles Sensors: Experiments on Camera and LiDAR},'' \emph{Black Hat Europe}, 2015.

\bibitem{potensic_atom_tracking_2023}
{Potensic}, ``{How to Better Use Potensic’s Atom’s Visual Tracking},'' \url{https://store.potensic.com/blogs/news/how-to-better-use-atoms-visual-tracking?srsltid=AfmBOoqWRqzrmDs16OWVJwVc6qtEd35r3qeXj82le2gPjJCAuXZ8eBrl}.

\bibitem{potensic_atom2_specs_2025}
{Potensic}, ``{ATOM-2 Specifications},'' \url{https://www.potensic.com/atom-2.html}, 2025.

\bibitem{prolific2025}
{Prolific}, ``{Prolific: Participant Recruitment for Research},'' \url{https://www.prolific.com}, 2025.

\bibitem{px4}
{PX4 Contributors}, ``{PX4},'' \url{https://github.com/PX4/PX4-Autopilot}.

\bibitem{quigley2009ros}
M.~Quigley, K.~Conley, B.~Gerkey, J.~Faust, T.~Foote, J.~Leibs, R.~Wheeler, A.~Y. Ng \emph{et~al.}, ``{ROS: An Open-Source Robot Operating System},'' in \emph{ICRA Workshop on Open Source Software}.\hskip 1em plus 0.5em minus 0.4em\relax Kobe, Japan, 2009.

\bibitem{redmon2018yolov3}
J.~Redmon, ``{YOLOv3: An Incremental Improvement},'' \emph{arXiv preprint arXiv:1804.02767}, 2018.

\bibitem{federal_register_2017_p1315}
F.~Register, ``{Federal Register Document 2017-01058, Paragraph 1315},'' \url{https://www.federalregister.gov/d/2017-01058/p-1315}, 2017.

\bibitem{federal_register_2017_p1375}
F.~Register, ``{Federal Register Document 2017-01058, Paragraph 1375},'' \url{https://www.federalregister.gov/d/2017-01058/p-1375}, 2017.

\bibitem{federal_register_2017_p203}
F.~Register, ``{Federal Register Document 2017-01058, Paragraph 203},'' \url{https://www.federalregister.gov/d/2017-01058/p-203}, 2017.

\bibitem{reuters_drones_2024}
Reuters, ``{Explainer: How Drones Are Being Used in the Ukraine War},'' \url{https://www.reuters.com/graphics/UKRAINE-CRISIS/DRONES/dwpkeyjwkpm/}, 2024.

\bibitem{autel2020}
A.~Robotics, ``{Autel EVO 2},'' \url{https://shop.autelrobotics.com/collections/autel-evo-ii-series}, 2020.

\bibitem{sandler2018mobilenetv2}
M.~Sandler, A.~Howard, M.~Zhu, A.~Zhmoginov, and L.-C. Chen, ``{MobileNetV2: Inverted Residuals and Linear Bottlenecks},'' in \emph{Proceedings of the IEEE Conference on Computer Vision and Pattern Recognition (CVPR)}, 2018.

\bibitem{sato2023lidar}
T.~Sato, Y.~Hayakawa, R.~Suzuki, Y.~Shiiki, K.~Yoshioka, and Q.~A. Chen, ``{LiDAR Spoofing Meets the New-Gen: Capability Improvements, Broken Assumptions, and New Attack Strategies},'' \emph{Network and Distributed System Security (NDSS) Symposium}, 2024.

\bibitem{sato2021dirty}
T.~Sato, J.~Shen, N.~Wang, Y.~Jia, X.~Lin, and Q.~A. Chen, ``{Dirty Road Can Attack: Security of Deep Learning Based Automated Lane Centering Under {{Physical-World}} Attack},'' in \emph{30th {{USENIX}} Security Symposium ({{USENIX}} Security)}, 2021.

\bibitem{sato2023intriguing}
T.~Sato, J.~Yue, N.~Chen, N.~Wang, and Q.~A. Chen, ``{Intriguing Properties of Diffusion Models: An Empirical Study of the Natural Attack Capability in Text-to-Image Generative Models},'' in \emph{Proceedings of the IEEE/CVF conference on computer vision and pattern recognition (CVPR)}, 2024.

\bibitem{schiller2023drone}
N.~Schiller, M.~Chlosta, M.~Schloegel, N.~Bars, T.~Eisenhofer, T.~Scharnowski, F.~Domke, L.~Sch{\"o}nherr, and T.~Holz, ``{Drone Security and the Mysterious Case of DJI's DroneID},'' in \emph{Network and Distributed System Security (NDSS) Symposium}, 2023.

\bibitem{shafahi2019adversarial}
A.~Shafahi, M.~Najibi, M.~A. Ghiasi, Z.~Xu, J.~Dickerson, C.~Studer, L.~S. Davis, G.~Taylor, and T.~Goldstein, ``{Adversarial Training for Free!}'' \emph{Advances in Neural Information Processing Systems (NeurIPS)}, 2019.

\bibitem{skydio2023}
Skydio, ``{Skydio 2 Plus},'' \url{https://www.skydio.com/skydio-2-plus-enterprise}, 2022.

\bibitem{skydio_2_plus_enterprise_2025}
{Skydio Inc.}, ``{Skydio 2+ Enterprise Product Page},'' \url{https://www.skydio.com/skydio-2-plus-enterprise}, 2025.

\bibitem{son2015rocking}
Y.~Son, H.~Shin, D.~Kim, Y.~Park, J.~Noh, K.~Choi, J.~Choi, and Y.~Kim, ``{Rocking Drones with Intentional Sound Noise on Gyroscopic Sensors},'' in \emph{24th USENIX Security Symposium (USENIX Security)}, 2015.

\bibitem{soomro2012ucf101}
K.~Soomro, A.~R. Zamir, and M.~Shah, ``{UCF101: A Dataset of 101 Human Actions Classes from Videos in the Wild},'' \emph{arXiv preprint arXiv:1212.0402}, 2012.

\bibitem{netgunstore2025}
T.~N.~G. Store, ``{The Net Gun Mega Pack -- Most Popular},'' \url{https://thenetgunstore.com/products/the-net-gun-mega-pack-most-popular?variant=9806433189945}, 2025.

\bibitem{forbes_skydio}
{Thomas, Brewster}, ``{Forbes: Founded By Ex-Google Engineers, Meet The Drone Startup Scoring Millions In Government Surveillance Contracts},'' \url{https://www.forbes.com/sites/thomasbrewster/2020/06/03/funded-by-kevin-durant-and-founded-by-ex-googlers-this-drone-startup-is-scoring-millions-in-government-surveillance-contracts/}.

\bibitem{tu2018injected}
Y.~Tu, Z.~Lin, I.~Lee, and X.~Hei, ``{Injected and Delivered: Fabricating Implicit Control Over Actuation Systems by Spoofing Inertial Sensors},'' in \emph{27th USENIX Security Symposium (USENIX Security)}, 2018.

\bibitem{vaswani2017attention}
A.~Vaswani, N.~Shazeer, N.~Parmar, J.~Uszkoreit, L.~Jones, A.~N. Gomez, {\L}.~Kaiser, and I.~Polosukhin, ``{Attention Is All You Need},'' \emph{Advances in Neural Information Processing Systems (NeurIPS)}, 2017.

\bibitem{wan2022too}
Z.~Wan, J.~Shen, J.~Chuang, X.~Xia, J.~Garcia, J.~Ma, and Q.~A. Chen, ``{Too Afraid to Drive: Systematic Discovery of Semantic DoS Vulnerability in Autonomous Driving Planning Under Physical-World Attacks},'' in \emph{Network and Distributed System Security (NDSS) Symposium}, 2022.

\bibitem{wang2023does}
N.~Wang, Y.~Luo, T.~Sato, K.~Xu, and Q.~A. Chen, ``{Does Physical Adversarial Example Really Matter to Autonomous Driving? Towards System-Level Effect of Adversarial Object Evasion Attack},'' in \emph{Proceedings of the IEEE/CVF International Conference on Computer Vision (ICCV)}, 2023.

\bibitem{wang2024revisiting}
N.~Wang, S.~Xie, T.~Sato, Y.~Luo, K.~Xu, and Q.~A. Chen, ``{Revisiting Physical-World Adversarial Attack on Traffic Sign Recognition: A Commercial Systems Perspective},'' in \emph{Network and Distributed System Security (NDSS) Symposium}, 2025.

\bibitem{wiyatno2019physical}
R.~R. Wiyatno and A.~Xu, ``{Physical Adversarial Textures That Fool Visual Object Tracking},'' in \emph{Proceedings of the {{IEEE/CVF}} International Conference on Computer Vision (ICCV)}, 2019.

\bibitem{wu2022dnd}
R.~Wu, T.~Kim, D.~J. Tian, A.~Bianchi, and D.~Xu, ``{DnD}: A cross-architecture deep neural network decompiler,'' in \emph{31st USENIX Security Symposium (USENIX Security)}, 2022.

\bibitem{xiang2021patchguard}
C.~Xiang, A.~N. Bhagoji, V.~Sehwag, and P.~Mittal, ``{PatchGuard: A Provably Robust Defense Against Adversarial Patches via Small Receptive Fields and Masking},'' in \emph{30th USENIX Security Symposium (USENIX Security)}, 2021.

\bibitem{xiang2021detectorguard}
C.~Xiang and P.~Mittal, ``{DetectorGuard: Provably Securing Object Detectors Against Localized Patch Hiding Attacks},'' in \emph{ACM SIGSAC Conference on Computer and Communications Security (CCS)}, 2021.

\bibitem{xiang2021patchguard++}
C.~Xiang and P.~Mittal, ``{PatchGuard++: Efficient Provable Attack Detection Against Adversarial Patches},'' in \emph{ICLR Workshop on Security and Safety in Machine Learning Systems}, 2021.

\bibitem{xie2017mitigating}
C.~Xie, J.~Wang, Z.~Zhang, Z.~Ren, and A.~Yuille, ``{Mitigating Adversarial Effects Through Randomization},'' \emph{International Conference on Learning Representations (ICLR)}, 2018.

\bibitem{flytrap_adv_2025}
S.~Xie, M.~H. Fakih, J.~Lu, F.~Alshammari, N.~Wang, T.~Sato, H.~Bouzidi, M.~A. Al~Faruque, and Q.~A. Chen, ``{Flytrap Project Website},'' \url{https://sites.google.com/view/av-ioat-sec/flytrap}, 2025.

\bibitem{xu2024physcout}
Y.~Xu, G.~Deng, X.~Han, G.~Li, H.~Qiu, and T.~Zhang, ``{PhyScout: Detecting Sensor Spoofing Attacks via Spatio-Temporal Consistency},'' in \emph{{{ACM SIGSAC}} Conference on Computer and Communications Security (CCS)}, 2024.

\bibitem{yan2020cooling}
B.~Yan, D.~Wang, H.~Lu, and X.~Yang, ``{Cooling-Shrinking Attack: Blinding the Tracker with Imperceptible Noises},'' in \emph{Proceedings of the {{IEEE/CVF}} Conference on Computer Vision and Pattern Recognition (CVPR)}, 2020.

\bibitem{yan2016can}
C.~Yan, W.~Xu, and J.~Liu, ``{Can You Trust Autonomous Vehicles: Contactless Attacks Against Sensors of Self-Driving Vehicle},'' \emph{Def Con}, 2016.

\bibitem{yan2020hijacking}
X.~Yan, X.~Chen, Y.~Jiang, S.-T. Xia, Y.~Zhao, and F.~Zheng, ``{Hijacking Tracker: A Powerful Adversarial Attack on Visual Tracking},'' in \emph{IEEE International Conference on Acoustics, Speech and Signal Processing (ICASSP)}.\hskip 1em plus 0.5em minus 0.4em\relax IEEE, 2020.

\bibitem{yin2024dimba}
X.~Yin, W.~Ruan, and J.~Fieldsend, ``{DIMBA: Discretely Masked Black-Box Attack in Single Object Tracking},'' \emph{Machine Learning}, 2024.

\bibitem{yu2020bdd100k}
F.~Yu, H.~Chen, X.~Wang, W.~Xian, Y.~Chen, F.~Liu, V.~Madhavan, and T.~Darrell, ``{BDD100K: A Diverse Driving Dataset for Heterogeneous Multitask Learning},'' in \emph{Proceedings of the IEEE/CVF Conference on Computer Vision and Pattern Recognition (CVPR)}, 2020.

\bibitem{yu2024physense}
Z.~Yu, A.~Li, R.~Wen, Y.~Chen, and N.~Zhang, ``{Physense: Defending Physically Realizable Attacks for Autonomous Systems via Consistency Reasoning},'' in \emph{ACM SIGSAC Conference on Computer and Communications Security (CCS)}, 2024.

\bibitem{zhang2019eye}
H.~Zhang, G.~Wang, Z.~Lei, and J.-N. Hwang, ``{Eye in the Sky: Drone-Based Object Tracking and 3D Localization},'' in \emph{Proceedings of the 27th {{ACM}} International Conference on Multimedia (MM)}, 2019.

\bibitem{zhang2022adversarial}
Q.~Zhang, S.~Hu, J.~Sun, Q.~A. Chen, and Z.~M. Mao, ``{On Adversarial Robustness of Trajectory Prediction for Autonomous Vehicles},'' in \emph{Proceedings of the IEEE/CVF Conference on Computer Vision and Pattern Recognition (CVPR)}, 2022.

\bibitem{zhao2019seeing}
Y.~Zhao, H.~Zhu, R.~Liang, Q.~Shen, S.~Zhang, and K.~Chen, ``{Seeing Isn't Believing: Towards More Robust Adversarial Attack Against Real-World Object Detectors},'' in \emph{{{ACM SIGSAC}} Conference on Computer and Communications Security (CCS)}, 2019.

\bibitem{zhou2022doublestar}
C.~Zhou, Q.~Yan, Y.~Shi, and L.~Sun, ``{DoubleStar: Long-Range Attack Towards Depth Estimation-Based Obstacle Avoidance in Autonomous Systems},'' in \emph{31st USENIX Security Symposium (USENIX Security)}, 2022.

\end{thebibliography}

\appendix

\subsection{Full-Stack ATT System Implementation}
\label{app:att-implement}

\begin{figure*}[t]
    \centering
    \includegraphics[width=\linewidth]{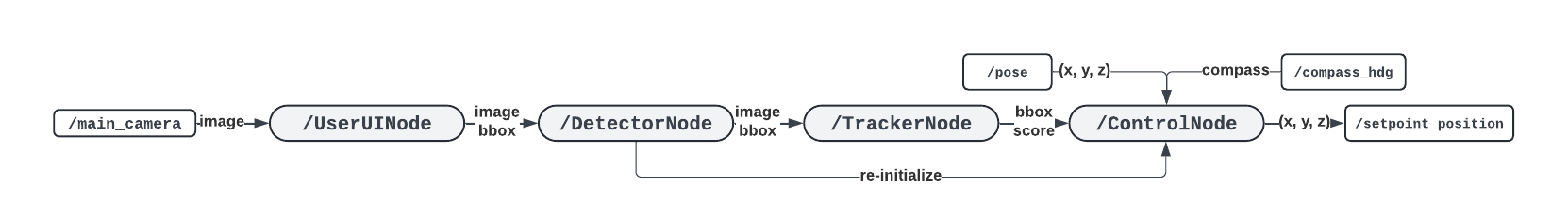}
    \caption{Overall pipeline of ROS node communication in our autonomous target tracking (ATT) drone.}
    \label{fig:rosnode}
\end{figure*}

\begin{figure*}[t!]
    \centering
    \includegraphics[width=\linewidth]{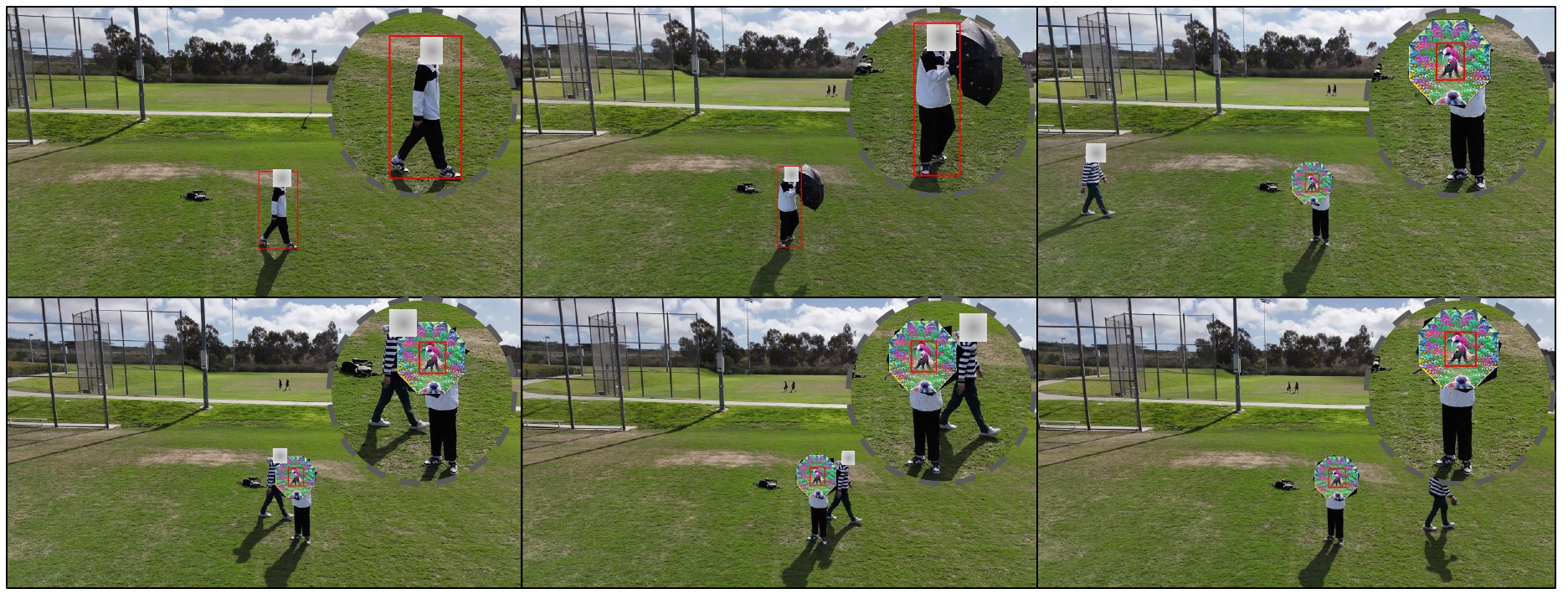}
    \caption{Environmental distraction case study. FlyTrap can achieve consistent attack effects even with the presence of an environmental distraction, e.g., another similar yet unobstructed person passing by the attacker. The umbrella pattern is applied offline and digitally, the same as the evaluation setups in Section~\ref{sec:digital-eval}. Zoomed-in view for clarity.}
    \label{fig:multi-obj}
\end{figure*}

\noindent
\textbf{Hardware setup.}
Our experimental platform is the Holybro X500 v2 drone~\cite{holybro_px4_development_kit}, a medium-lift quadcopter powered by a Pixhawk flight controller. The Pixhawk controller is equipped with onboard inertial sensors and a GPS module. We mount an NVIDIA Jetson Orin Nano~\cite{nvidia2019jetson} on the drone to handle tasks related to object detection, tracking, and control generation. The flight controller runs PX4~\cite{px4}, an open-source flight firmware, while an Intel RealSense D435i camera~\cite{intel_realsense_d435i} provides a live video stream. The Jetson creates a WiFi access point, enabling the operator to view the video via a browser-based interface; from this interface, the operator can select the intended target by drawing a bounding box shown in Fig.~\ref{fig:physical-setups}.

\noindent
\textbf{Software setup.}
We employ the Robot Operating System (ROS)~\cite{quigley2009ros} as the communication backbone for inter-module interactions within our system. The architecture of the overall pipeline is depicted in Fig.~\ref{fig:rosnode}. The system comprises four distinct ROS nodes: \texttt{UserUINode}, \texttt{DetectorNode}, \texttt{TrackerNode}, and \texttt{ControlNode}. The \texttt{UserUINode} serves as an interface, allowing users to select the target object through a simple web interface. This node subsequently publishes the selected bounding box (\texttt{bbox}) alongside the corresponding image. The \texttt{DetectorNode} processes the image and \texttt{bbox} by utilizing an SSD-MobileNet-V2 model~\cite{liu2016ssd, sandler2018mobilenetv2} to refine the user-selected \texttt{bbox}. To enhance precision, we filter the detection results to focus solely on the "person" class and calculate the Intersection over Union (IoU) between the detection \texttt{bbox} and the user-selected \texttt{bbox}, outputting the box with the highest IoU. This calibration step significantly improves the performance of the SOT model by providing an accurate target object location, which is particularly crucial when the drone is airborne, where it may be challenging for the user to select a precise \texttt{bbox}. The \texttt{TrackerNode} then takes the calibrated \texttt{bbox} and image as inputs and initiates object tracking, publishing the tracked \texttt{bbox} and the corresponding confidence score to the \texttt{ControlNode}. It should be noted that the same images are shared among \texttt{UserUINode}, \texttt{DetectorNode}, and \texttt{TrackerNode} during initialization to ensure information synchronization. Finally, the \texttt{ControlNode} processes these signals only if the confidence exceeds a predefined threshold. Specifically, the \texttt{ControlNode} determines the movement offsets along three coordinates (i.e., $x$, $y$, and $z$) based on the center position and area of the \texttt{bbox}. The control algorithm strives to maintain the target object at the center of the image while preserving its size as initially specified. In the experiments, we find the \texttt{/pose} topic published by MAVROS~\cite{mavlink, mavros} is not that accurate. Therefore, instead of publishing absolute locations, we add the offset $(\Delta x, \Delta y, \Delta z)$ to the current pose $(x, y, z)$ received from \texttt{/pose} in real time to ensure stability.

\begin{figure}
    \centering
    \includegraphics[width=0.7\linewidth]{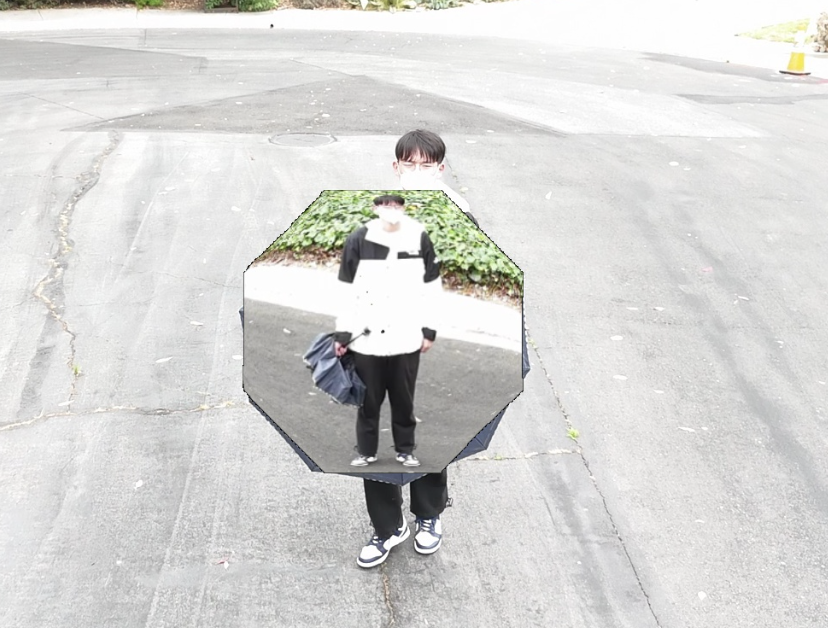}
    \caption{Target photo baseline attack (TGT) example. The attackers can print a photo of themselves to misguide the object tracker. This is considered a baseline given that the printed human naturally resembles the tracking target and is smaller than the original target. We use the images from the first frame of the training set as the default TGT. The observed distortion is induced in the rendering process to simulate real-world umbrella geometry.}
    \label{fig:tgt-visualization}
\end{figure}

\subsection{Human Figure Printing Baseline Attack}
\label{app:tgt}
In this section, we provide further information on our target photo baseline attack (TGT). As shown in Fig.~\ref{fig:tgt-visualization}, TGT means the attackers print the photos of themselves on the umbrella to misguide the ATT drone. We design the attack since it's a straightforward implementation: (1) the TGT images naturally resemble the genuine target being tracked; (2) the TGT images are smaller than the original target, which can potentially cause the shrinking effect, and (3) the attackers don't need any adversarial machine learning expertise to conduct the attack.

\begin{figure}
    \centering
    \includegraphics[width=\linewidth]{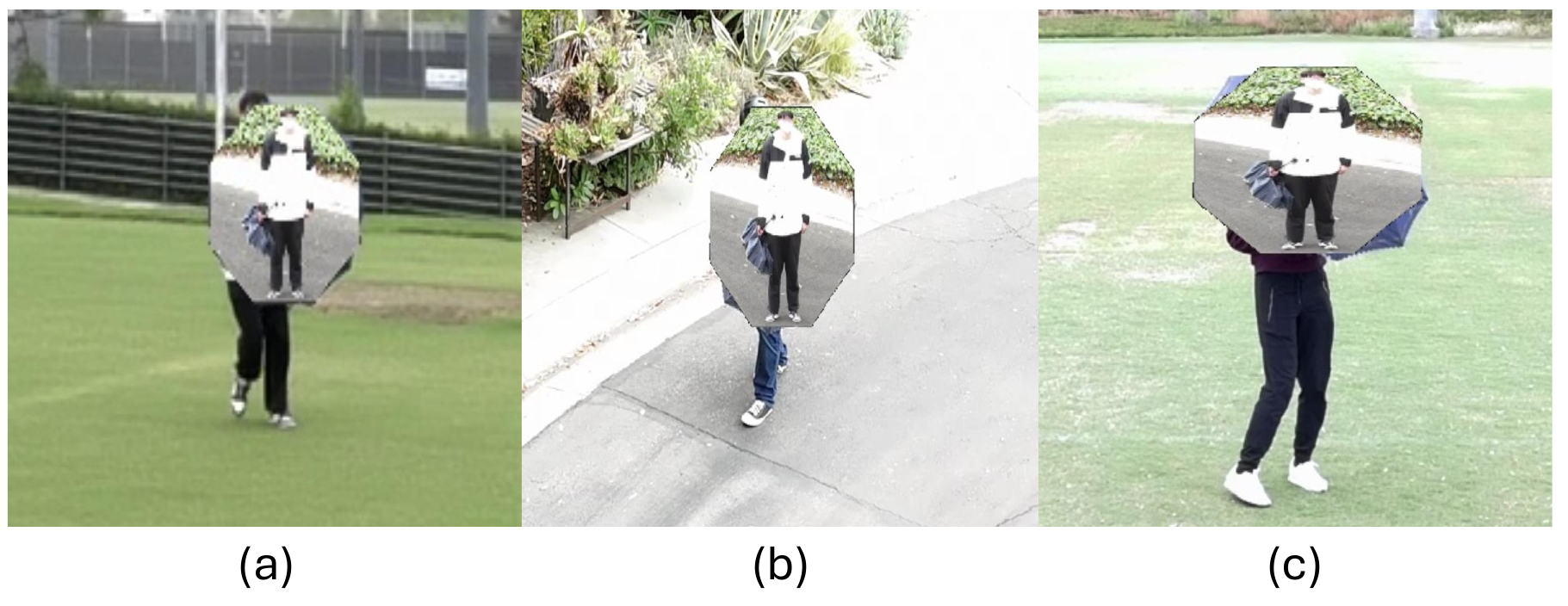}
    \caption{Target photo (TGT) universality evaluation: (a) universality to different backgrounds; (b) universality to different persons; (c) universality to different backgrounds and persons.}
    \label{fig:tgt-universal}
\end{figure}

To guarantee fair comparison with the proposed FlyTrap, we conduct a grid search to find the person/background ratio that can maximize the mASR$_\text{open}$. If the person is too large, the shrink is minor, and if the person is too small, it might not have enough features to misguide the tracker to lock on it. We randomly select four different people and four different backgrounds from the training set. Then, we generate TGT with human-figure to background ratios ranging from 0.01 to 0.9. For each TGT, we evaluate it on the test set with the same person and background (e.g., in Fig.~\ref{fig:tgt-visualization}) across four models. The results can be seen from Fig.~\ref{fig:tgt-search}. We find that Transformer-based MixFormer is more vulnerable to a smaller-sized printing target. On the other hand, SiamRPN-based models are more vulnerable to medium-sized printing targets due to their mechanism to penalize predictions that largely differ from the initialization size.

\begin{figure}
    \centering
    \includegraphics[width=\linewidth]{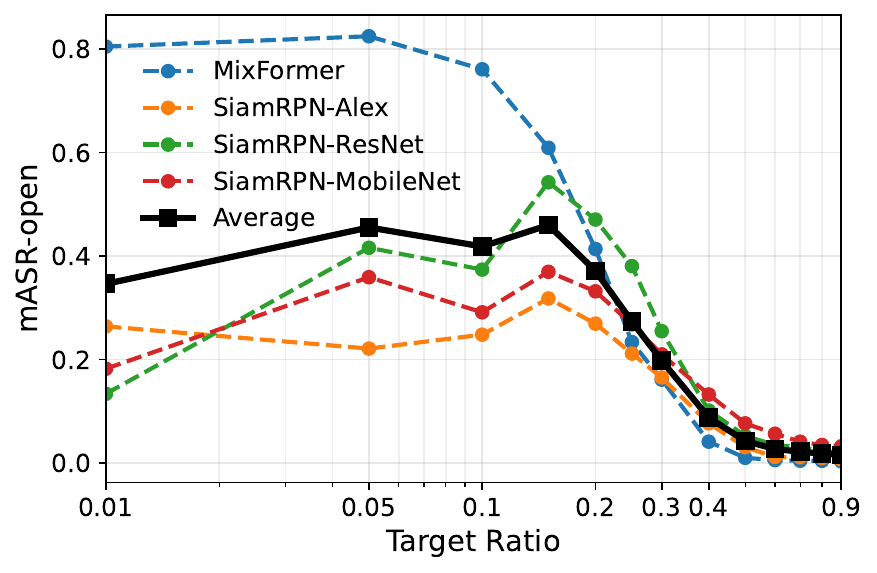}
    \caption{Target photo baseline attack. We randomly select four unique combinations of person and background from the training set, evaluating on the corresponding testing set. We adopt the ratio (i.e., 0.15) that achieves the highest mASR$_\text{open}$ as a fair baseline comparison with the proposed FlyTrap attack.}
    \label{fig:tgt-search}
\end{figure}

Regarding the universality evaluation in Section~\ref{sec:universality}, as shown in Fig.~\ref{fig:tgt-universal}, we test it with different deployed persons and/or backgrounds.

\subsection{Dataset Processing}
\label{app:data-process}

To facilitate the training process, we leverage the off-the-shelf single-object tracking model to label the videos. We also perform down-sampling to balance the amount of data in different scenarios.

\textbf{Automate labeling.} 
We need the tracked person's location to initialize the tracker. We also require the umbrella's location and size to estimate the distance, determine the shrink rate, and use it in the compositing process (Section~\ref{sec:progressive-distance-pulling}). To reduce manual efforts, we use a trained SOT that automates the labeling process. We find that these generated labels are sufficiently accurate for our simulation and optimization.

\noindent
\textbf{Down-sampling and balance scenario.}
Consecutive frames contain nearly identical visual information; thus, the raw videos are redundant in terms of scenario variety. Such redundancy increases optimization overhead without improving the generality of adversarial patterns. To address this, we down-sample the videos and balance the distribution of scenarios.

\begin{figure*}[htbp]
    \centering
    \begin{subfigure}[b]{0.23\textwidth}
        \centering
        \includegraphics[width=\textwidth]{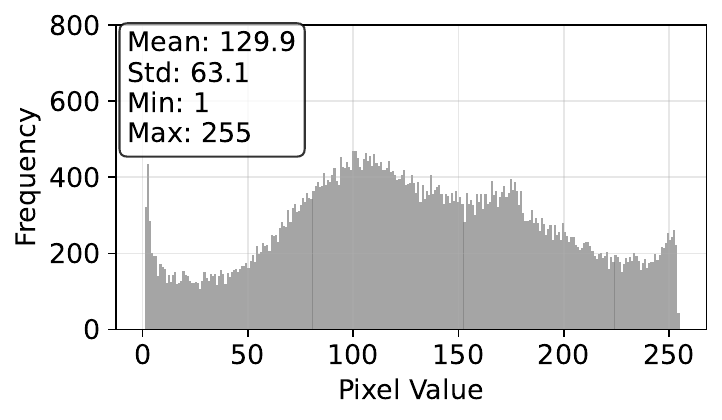}
        \caption{MixFormer}
        \label{fig:image1}
    \end{subfigure}
    \hfill
    \begin{subfigure}[b]{0.23\textwidth}
        \centering
        \includegraphics[width=\textwidth]{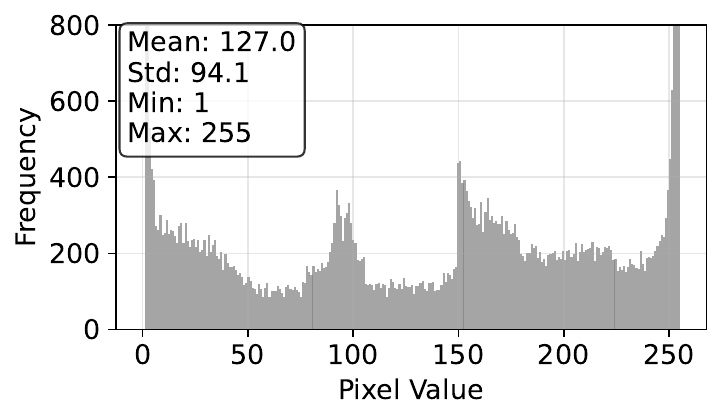}
        \caption{SiamRPN-AlexNet}
        \label{fig:image2}
    \end{subfigure}
    \hfill
    \begin{subfigure}[b]{0.23\textwidth}
        \centering
        \includegraphics[width=\textwidth]{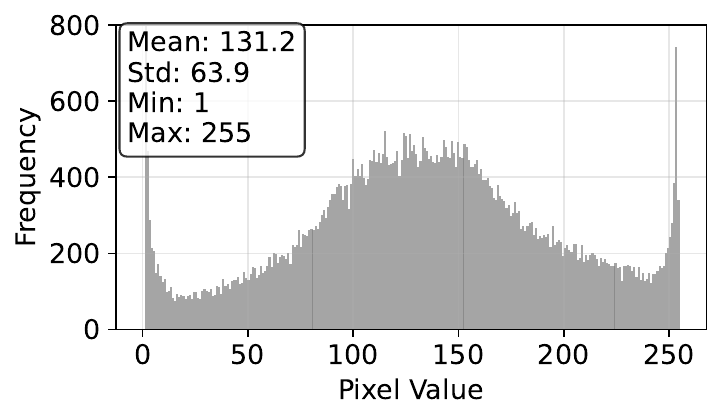}
        \caption{SiamRPN-MobileNet}
        \label{fig:image3}
    \end{subfigure}
    \hfill
    \begin{subfigure}[b]{0.23\textwidth}
        \centering
        \includegraphics[width=\textwidth]{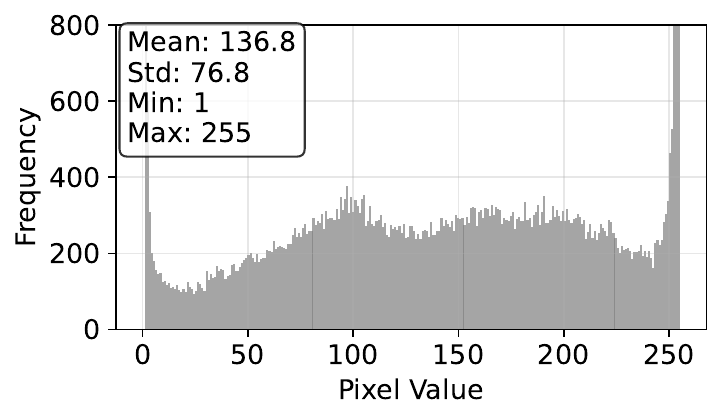}
        \caption{SiamRPN-ResNet}
        \label{fig:image4}
    \end{subfigure}
    \caption{Pixel histogram of attack patterns optimized against four different white-box single object tracking models. We find the attack patterns have similar color intensity (i.e., mean pixel value). Specifically, we find that the histogram of FlyTrap optimized against SiamRPN-MobileNet shares high similarity with MixFormer. However, patterns optimized for SiamRPN-MobileNet fail to attack commercial drones, suggesting the successful attack is due to the pattern, not color distinction.}
    \label{fig:pixel-distribution}
\end{figure*}

\subsection{Scenario with Multiple Object Case Study}
\label{sec:case-study}

Since our collected dataset mainly includes one person in the scenario, we also study a more complex scenario to illustrate the robustness of our attacks toward environmental distractions, e.g., with multiple objects in the scenario. In Fig.~\ref{fig:multi-obj}, we collect a video where the attacker is launching the attack while another person passes by the attacker closely. The scenario challenges whether FlyTrap can cause consistent DPA effects with a similar but unobstructed object (e.g., another person here) appearing in the scenario. As shown in Fig.~\ref{fig:multi-obj}, FlyTrap umbrella patterns consistently shrink to a subarea within the umbrella surface, even with the presence of another person passing by closely. Two explanations demonstrate the results. First is about how the SOT model works: SOT will search the tracked object within a subarea in the current image~\cite{li2018high, cui2022mixformer}, and the size of the search area is determined by the prediction bounding box size in the last frame. Thus, when the third person is away from the attacker, it is not presented in the SOT search area. Secondly, even when the third person is within the search area, our objective function is to maximize the prediction score (in Section~\ref{sec:adv-obj-func}), ensuring the prediction confidence on the attack area is still larger than the prediction confidence on the third person, thus being robust to similar environmental distractions.

\subsection{Real-World Attack Distance and Angle Discussion}
\label{app:distance-shrink-discussion}

For the attack distance, it is important to first note that in our problem setting, the attack distance limit should be measured and judged with respect to the tracking distance limit of the ATT feature, since the ATT feature itself is not designed for arbitrary tracking distance in the first place. Specifically, a primary factor that fundamentally affects both the maximum tracking distance of ATT systems and, consequently, the maximum attack distance of FlyTrap is the resolution and optical capabilities of the drone’s onboard camera. For example, popular consumer commercial drones such as the DJI Mini series, Potensic Atom, Autel, and Skydio 2 drones are equipped with 4K videos capabilities~\cite{dji_mini4pro_specs_2025, potensic_atom2_specs_2025, autel_evo_lite_640t_enterprise_2025, skydio_2_plus_enterprise_2025} and the maximum tracking distance they can support is typically up to 20 meters for person tracking~\cite{dji_active_tracking_2023, potensic_atom_tracking_2023, autel_dynamic_track_2024, skydio_update_2020}. Meanwhile, the Skydio drones designed specifically for police and equipped with a 65x zoom camera can read license plates from $\sim$250 meters and track a vehicle from as far as 4.8 kilometers away~\cite{new_auto}.

In our experiment setup, all the data used to generate the FlyTrap attack patterns are collected using the DJI Mini 4 Pro drone with ATT enabled, which are thus all bounded by its maximum tracking distance of 20 meters~\cite{dji_active_tracking_2023}. As shown, FlyTrap can maintain the attack effectiveness for distances up to 20 meters, which shows that it can fully cover the operational range of today's commercial drones with 4K videos capabilities~\cite{dji_active_tracking_2023, potensic_atom_tracking_2023, autel_dynamic_track_2024, skydio_update_2020}. Furthermore, for several adversarial umbrella instances, the attack can remain effective even at 30–40 meters, which is beyond the distances that these attack instances are optimized for (up to 20 meters with the current dataset and PDP); this shows that the attack effect of FlyTrap can potentially even generalize beyond the distance limit of the dataset used for attack optimization.

To conclude, in our current experimental setup, we do not find that FlyTrap has attack distance limitations with respect to the maximum tracking distance of the commercial drone setups it is optimized for. We leave such potential further confirmation of the attack capability in longer distance ranges to future work.

\begin{figure}
    \centering
    \includegraphics[width=\linewidth]{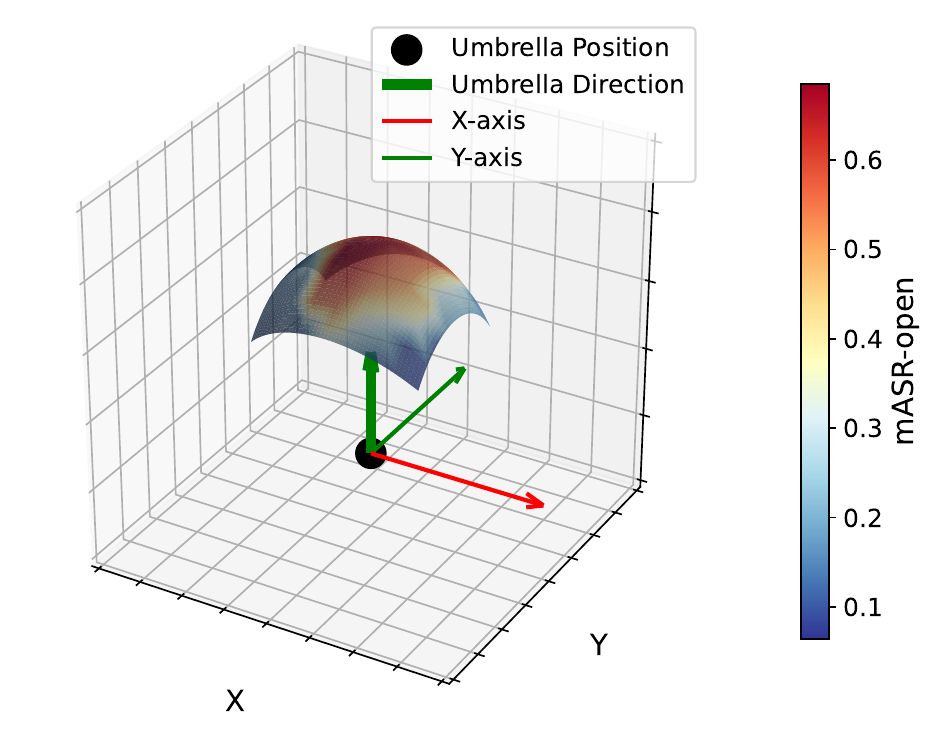}
    \caption{Spatial visualization of mASR$_\text{open}$ with different observation angle. The umbrella is pointed upwards, and the camera is always pointing directly at the umbrella from different spatial positions. This simulates the scenario where the attacker directs the umbrella to the ATT drone imperfectly.}
    \label{fig:angle}
\end{figure}

\begin{figure}
    \centering
    \includegraphics[width=0.8\linewidth]{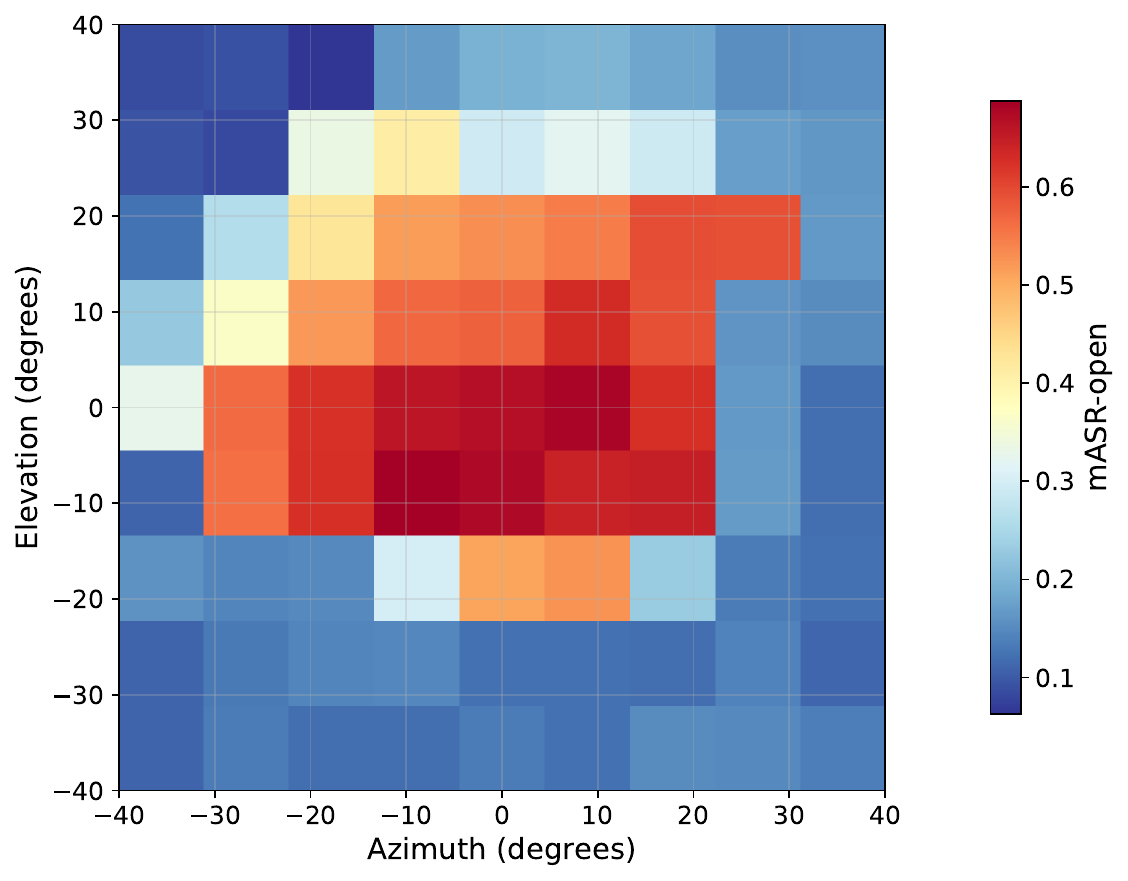}
    \caption{Heatmap of mASR$_\text{open}$ with different angle. The FlyTrap can largely tolerate an umbrella pointing error within 10 degrees, indicated by a similar mASR$_\text{open}$ from different angles within 10 degrees.}
    \label{fig:angle-heatmap}
\end{figure}

\begin{figure}
    \centering
    \includegraphics[width=0.5\linewidth]{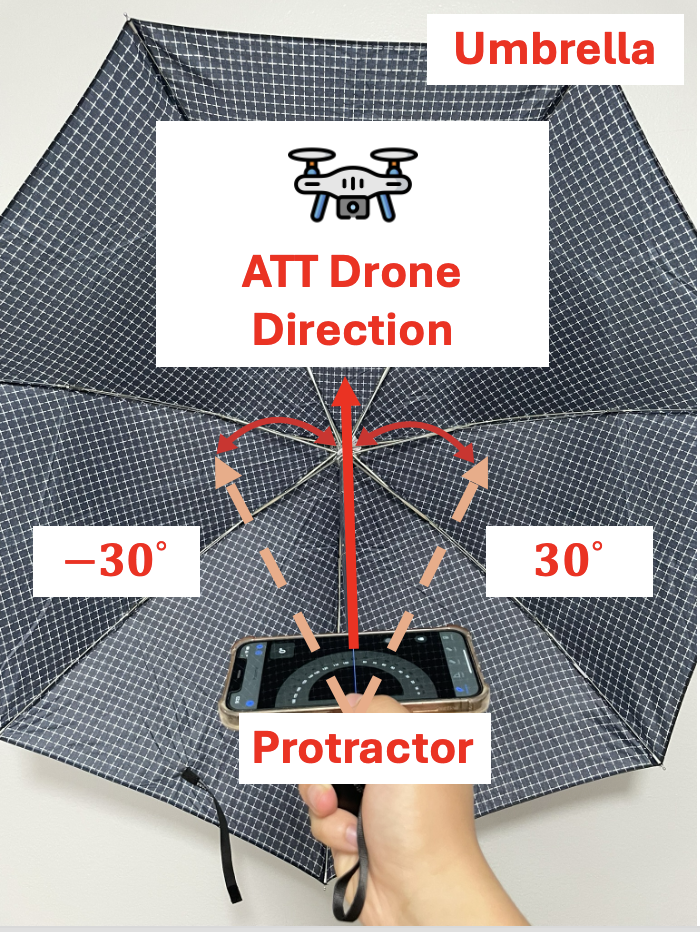}
    \caption{A real-world example of umbrella directing. The attack effectiveness only significantly degrades when getting to an angle of 30$^\circ$. We believe this angle is easily controllable for the attacker to aim the ATT drones.}
    \label{fig:umbrella-directing}
\end{figure}

In terms of the attack angle, since it's hard to accurately control the angle in physical experiments, we simulate it in the rendering process in the digital simulation. We change the azimuth angle and elevation angle of the simulated camera in rendering, and simulate the effect that the attacker doesn't point the umbrella at the ATT drone perfectly. The results are shown in Fig.~\ref{fig:angle} and \ref{fig:angle-heatmap}. We find that the attack success rate gradually degrades with a larger observation angle, suggesting the robustness design (Section IV-D-3) enables FlyTrap to tolerate umbrella directing error within a certain range (e.g., 10 degrees). Specifically, we randomly sample azimuth and elevation angles within [-5$^\circ$, 5$^\circ$] during the optimization. As shown in the heatmap, the umbrella works well within that range and even beyond. The attack effectiveness only significantly degrades when getting to an angle of 30$^\circ$, which we consider to be already well within the normal controllable range when a normal person intentionally tries to aim the umbrella at the drone (visualized in Fig.~\ref{fig:umbrella-directing} using a real umbrella). Our real-world experiments in Section V-F-2 and Section V-G-3 have also confirmed that such a level of attack robustness to directing angle errors is sufficient for real-world attack usage: In all these experiments, we closely simulated real-world attack usage in which the attacker starts by a best-effort aiming of the umbrella towards the drone (as shown by all the video demos in our website~\cite{flytrap_adv_2025}), and we never experienced any needs to have more than 1 round of such best-effort aiming in order to achieve the attack effect currently reported in these two sections.

\subsection{Commercial Experiment Comparison}
\label{app:commercial-justify}

We further conduct experiments to show that one of the umbrella successfully attack three commercial drones because of the pattern, not the occlusion or color brightness. First, we use a benign black umbrella, following the same experimental setup in Section~\ref{sec:commercial-setup}. We found that none of the commercial drones reacted to the black umbrella during 10 repeated tests. The black umbrella attack video demonstration can be found on our website. Second, we visualize the pixel distribution of our generated umbrella patterns in Fig.~\ref{fig:pixel-distribution}. The physical umbrella can also be seen from Fig.~\ref{fig:commercial-setups}. We observe that the umbrella pattern has a close pixel mean value, suggesting similar pixel intensity. Specifically, patterns against SiamRPN-MobileNet show a similar distribution to patterns against MixFormer. However, the former one fails to attack any commercial drone while the latter succeeds. All the evidence suggests that one of our umbrellas works because of the pattern, not the occlusion or the color distinction.

\subsection{User Study}
\label{sec:user-study}

In this section, we conduct a user study to more directly evaluate the stealthiness of the FlyTrap attack. We have gone through the IRB process, and our study is determined as in the IRB Exempt category since it does not involve the collection of any Personally Identifiable Information (PII) or target any sensitive population.

\textbf{Evaluation methodology.} To conduct the user study, we used the umbrella from our previous physical experiments alongside eight commercially available umbrellas. At the beginning of the study, we present participants with a video recorded on a non-sunny day, depicting a person holding an umbrella in a crowd. This ensures that participants fully understand the contextual setting of the subsequent questions. Participants then respond to a series of carefully designed questions to scientifically evaluate whether our Flytrap attack is likely to arouse suspicion in real-world scenarios.

\textbf{Evaluation setup.} We use the Prolific platform to perform this study, and in total collected 200 participants. The demographic distribution of the participants is illustrated in Fig.~\ref{fig:demographic information}. The participants demonstrate a clear understanding of our test scenario videos and questions, ensuring that our tests effectively simulate the plausibility of using the attack umbrella in real-world settings.

We design a series of questions to study how people perceive umbrella use on non-rainy days and compare our umbrella with eight commercially available options. The study focuses on four objectives: (1) to assess whether holding an umbrella on a non-rainy day attracts attention and why; (2) to understand if holding an umbrella on a non-rainy day is seen as unusual; (3) to test whether our umbrella is more noticeable; and (4) to examine if our umbrella design is considered suspicious. The questions are structured as follows.

\textit{\textbf{Q1}}: \textit{``In the video above, do you notice a person holding an umbrella? If so, how do you feel about this behavior?''}

Follow-up \textit{q}: \textit{``If your answer is yes, please briefly explain why. Otherwise, write N/A.''  }

\textit{\textbf{Q2}}: \textit{``Do you think seeing a person using an umbrella on a non-rainy day is abnormal?''}  

\textit{\textbf{Q3}}: \textit{``Which umbrella do you think is the most eye-catching, if any?''}  

\textit{\textbf{Q4}}: \textit{``Which umbrella do you think is the most suspicious, if any?''}

Follow-up \textit{q1}. \textit{``Do you think the umbrella(s) you chose as suspicious could be used in illegal behavior?''} 

Follow-up \textit{q2}: \textit{``Why do you think it’s suspicious? If not, type N/A.''}

\begin{figure}[t]
    \centering
    \begin{subfigure}[b]{0.38\linewidth}
        \centering
        \includegraphics[width=\linewidth]{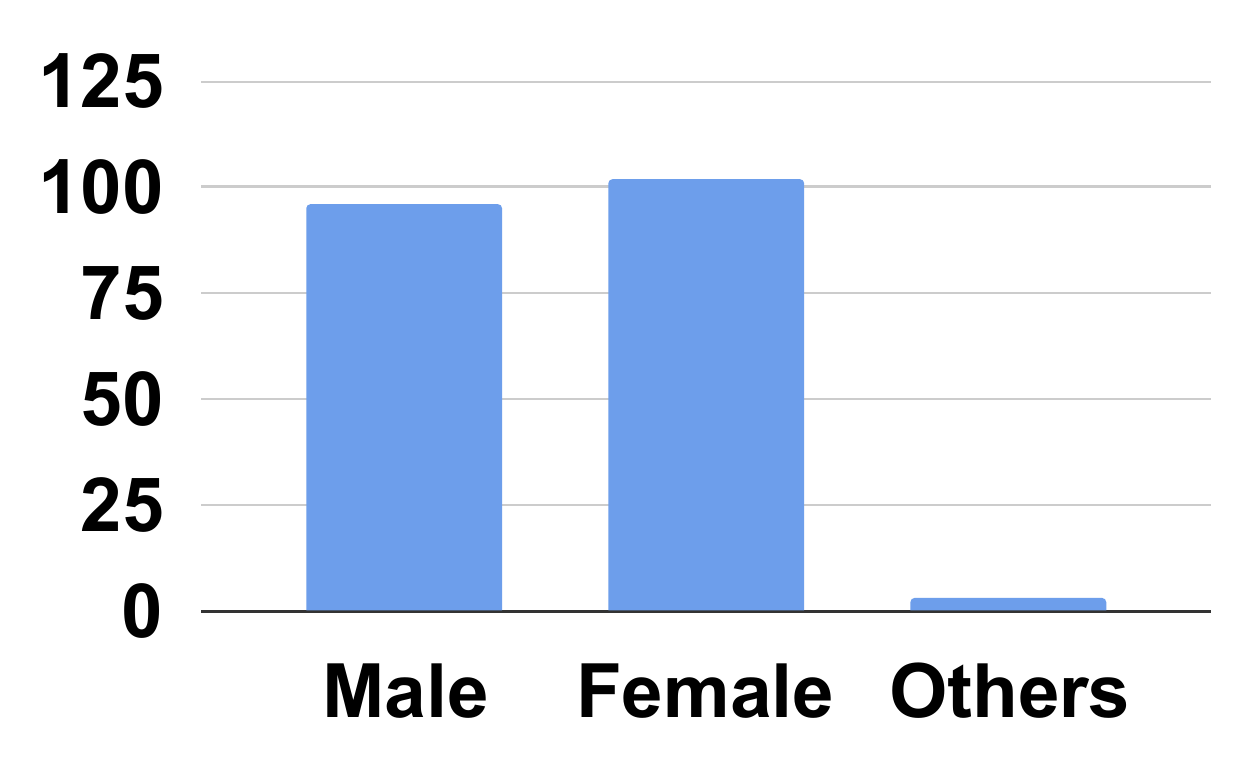}
        \caption{Gender.}
        \label{fig:gender}
    \end{subfigure}
    \hfill
    \begin{subfigure}[b]{0.5\linewidth}
        \centering
        \includegraphics[width=\linewidth]{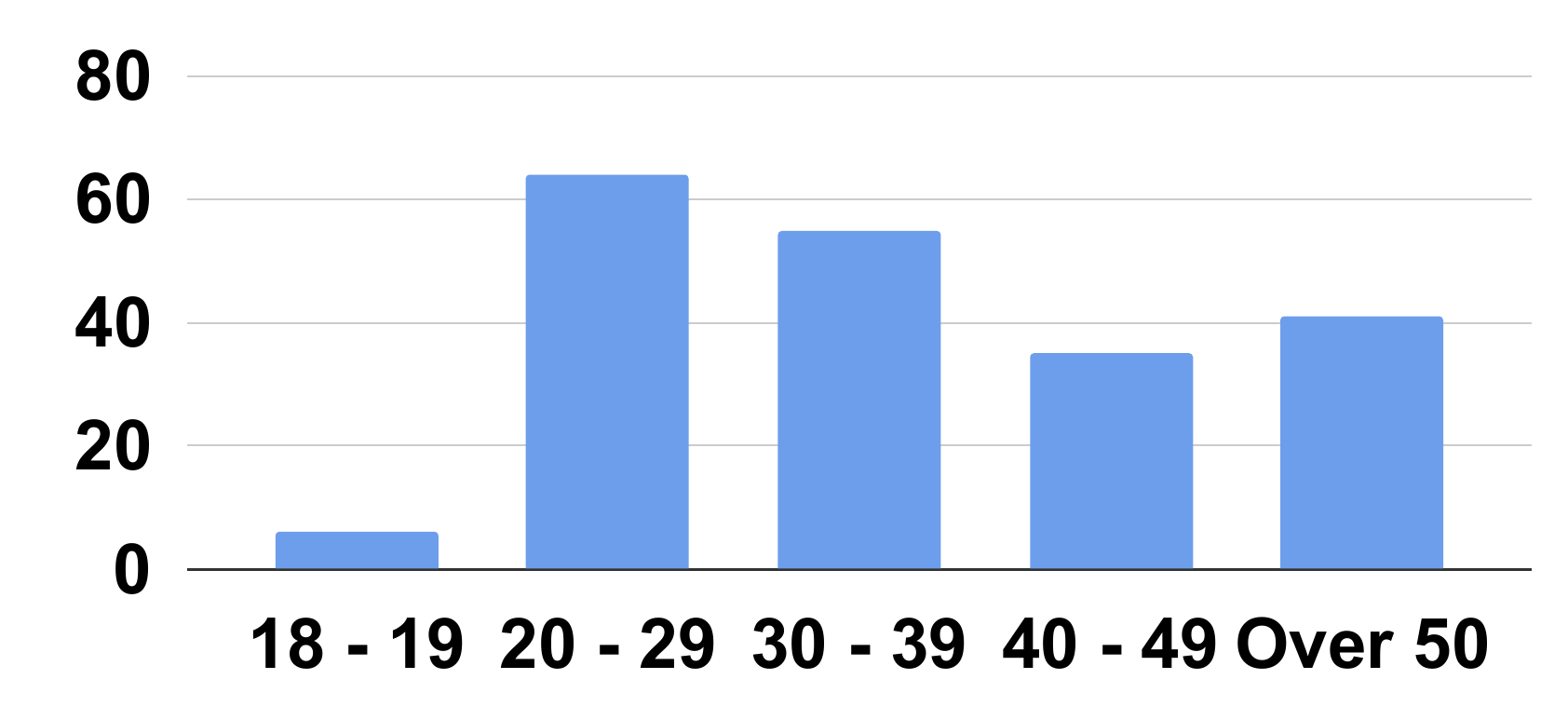}
        \caption{Age.}
        \label{fig:age}
    \end{subfigure}
    \caption{Statistics of the demographic in the user study.}
    
    \label{fig:demographic information}
\end{figure}

\textbf{Results.} Fig.~\ref{fig:userstudy1} and Fig.~\ref{fig:userstudy2} show the results of \textit{Q1} and \textit{Q2}. For \textit{Q1}, as shown in Fig.~\ref{fig:userstudy_1}, the majority of participants (116 out of 200, 58\%) do not notice anyone holding an umbrella in the crowd. Among those who do notice (65 participants, 32.5\%), most pay little attention to the umbrella holder, while only a small portion (19 participants, 9.5\%) pay special attention. As shown in Fig.~\ref{fig:interesting_reason}, further analysis reveals that the primary reason participants notice the umbrella is that the person is the only one holding an umbrella, while another common reason is the unusual use of an umbrella on a sunny day. For \textit{Q2}, most participants (158 out of 200, 79\%) believe that using an umbrella on a non-rainy day is acceptable. This suggests that umbrella use in such conditions is generally perceived as natural and does not raise suspicion.  

\begin{figure}[t]
    \centering
    \begin{subfigure}[b]{0.8\linewidth}
        \centering
        \includegraphics[width=\linewidth]{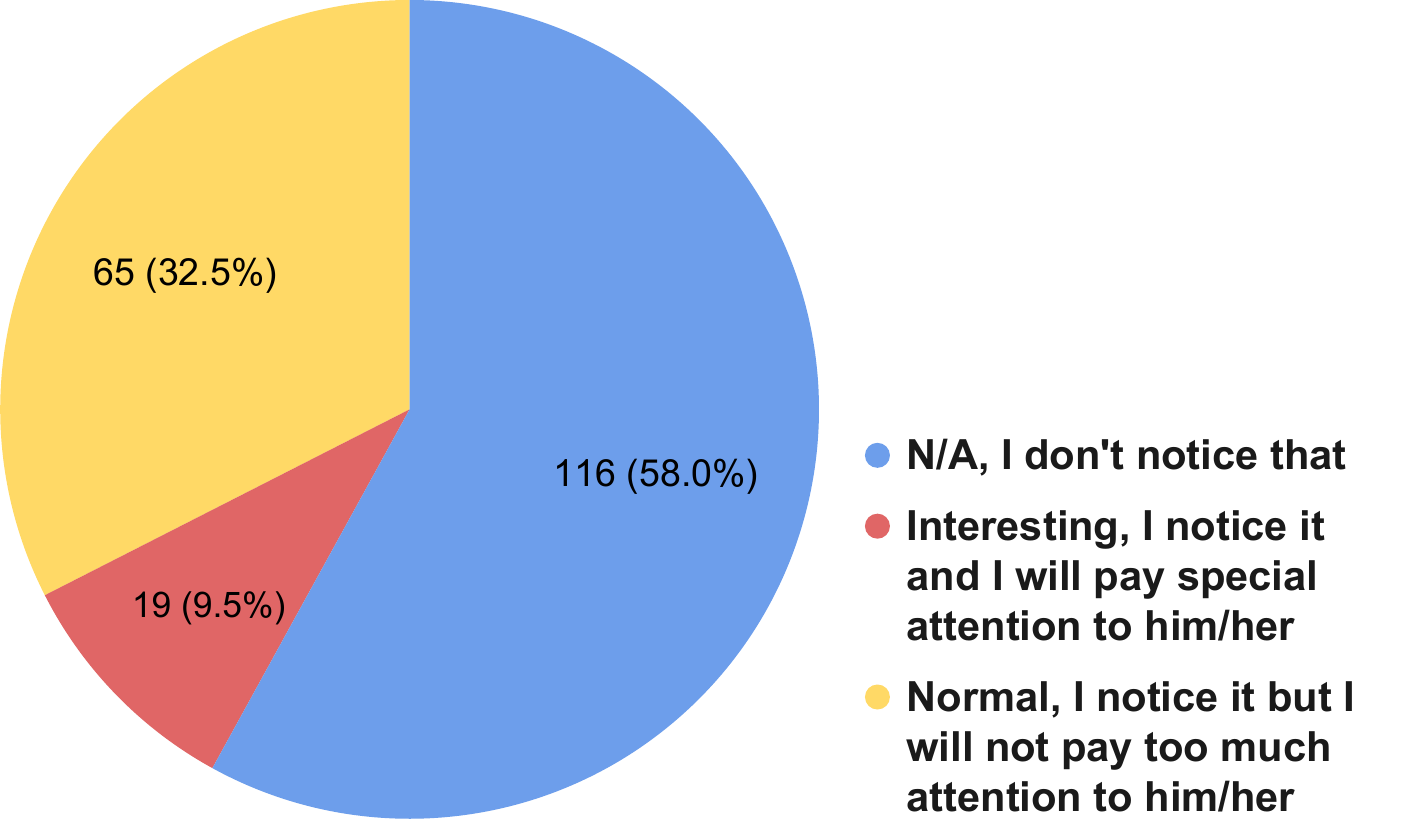}
        \caption{Participant responses to \textit{Q1}.}
        \label{fig:userstudy_1}
    \end{subfigure}
    \hfill
    \vspace{0.2cm}
    \begin{subfigure}[b]{0.8\linewidth}
        \centering
        \includegraphics[width=\linewidth]{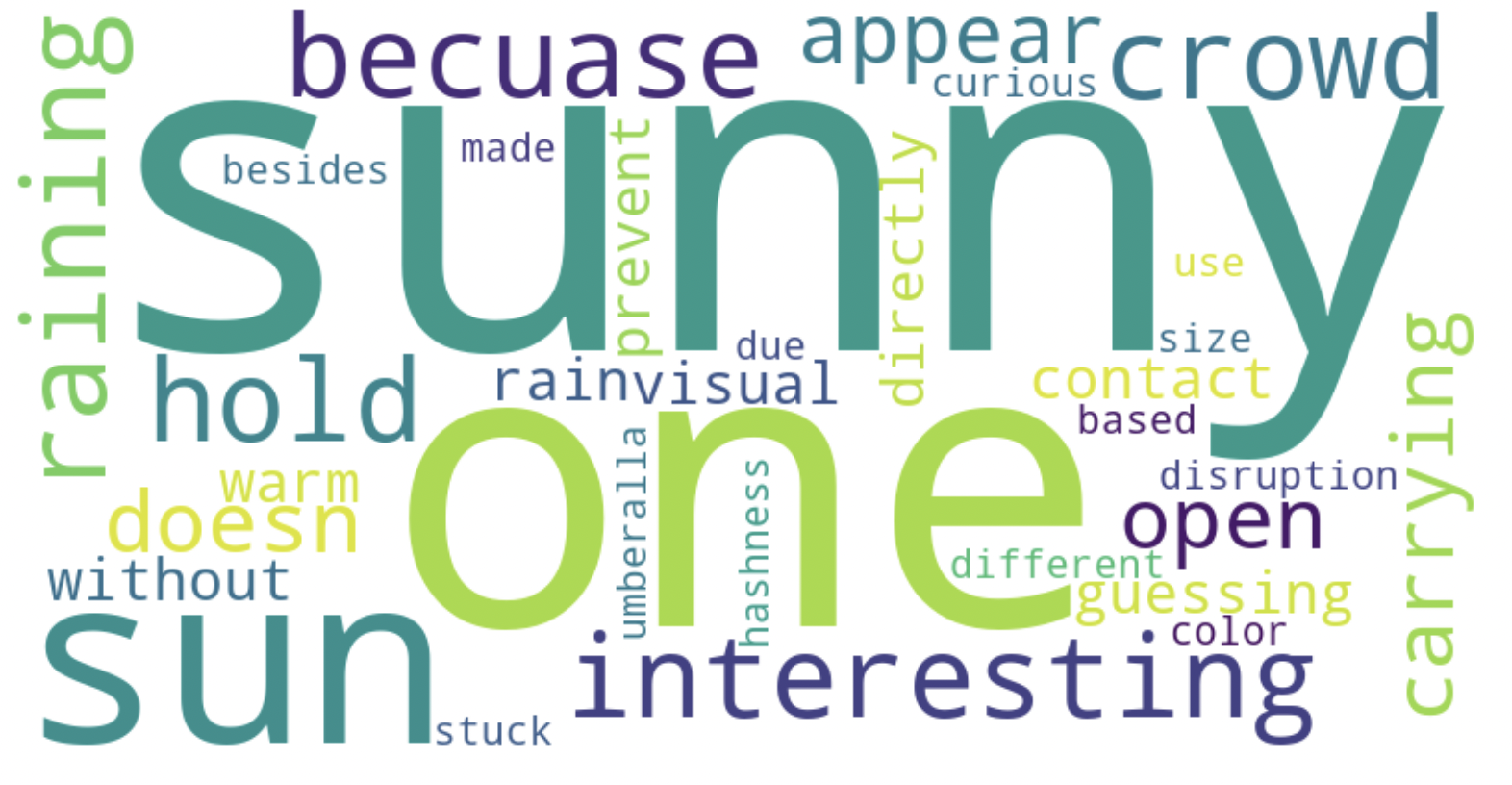}
        \caption{Reasons for selecting ``interesting'' above.}
        \label{fig:interesting_reason}
    \end{subfigure}
    \caption{Responses to \textit{Q1}: videos containing a person using an umbrella on a non-rainy day.}
    
    \label{fig:userstudy1}
\end{figure}

\begin{figure}[t]
    \centering
    \includegraphics[width=0.8\linewidth]{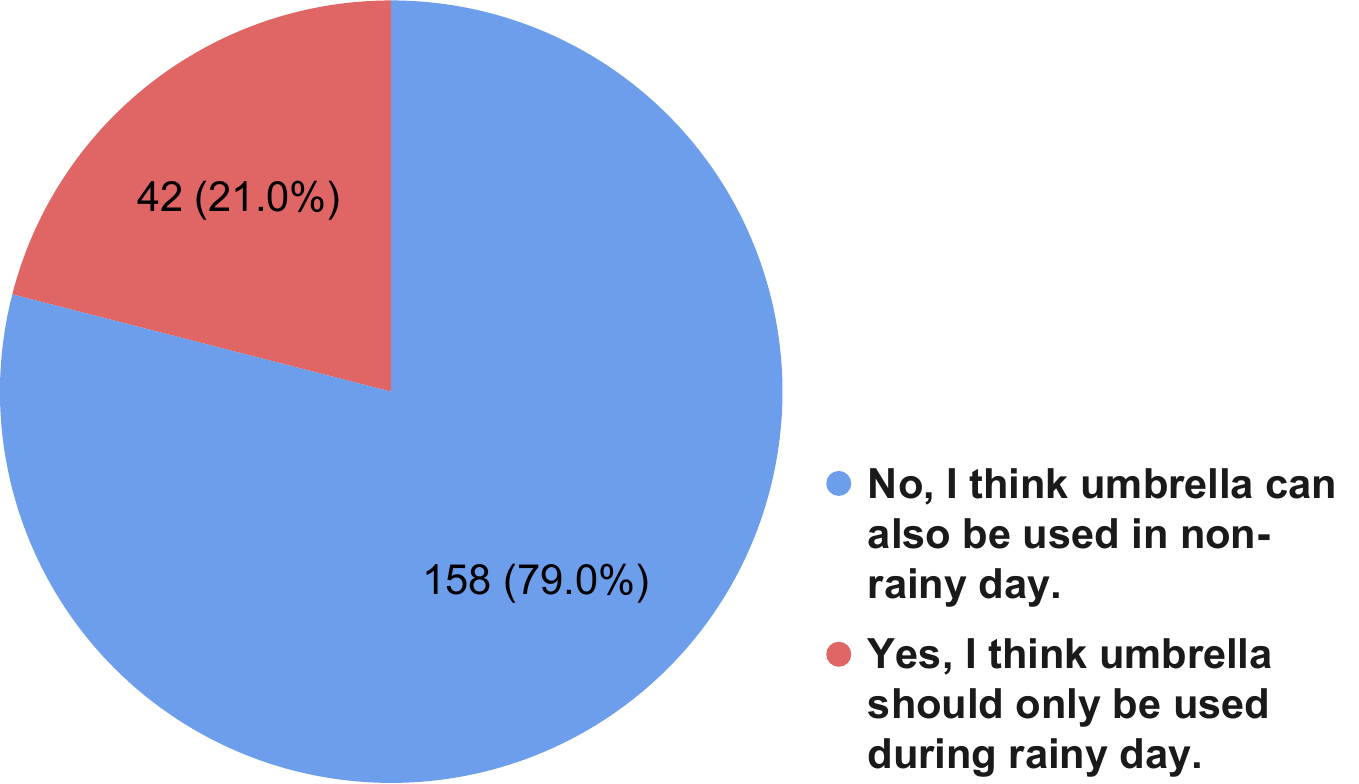}
    \caption{Response to \textit{Q2}: using umbrella on non-rainy days.}
    
    \label{fig:userstudy2}
\end{figure}

\begin{figure}[h]
    \centering
    \includegraphics[width=0.9\linewidth]{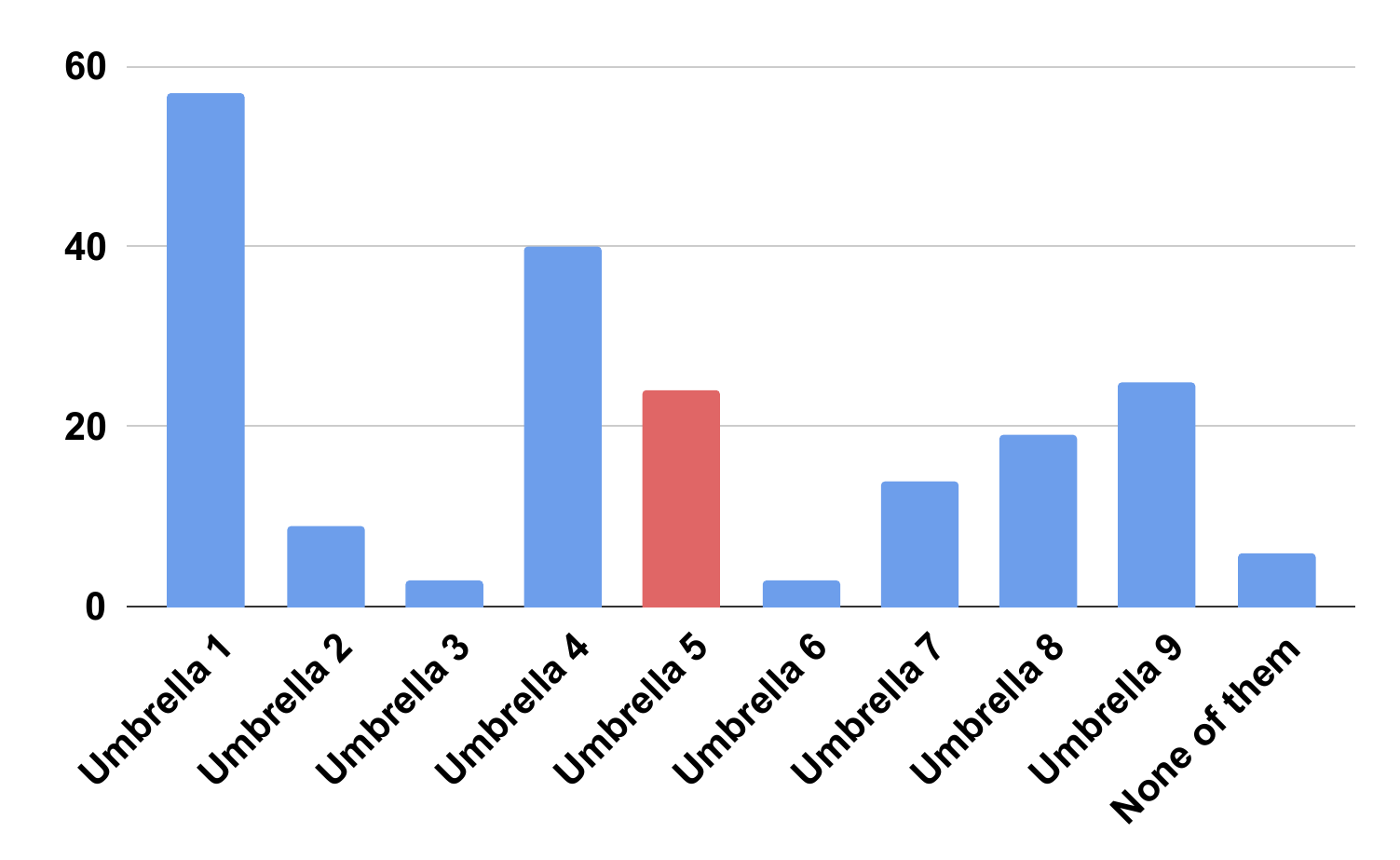}
    \caption{Response to \textit{Q3}: the most eye-catching umbrella.}
    
    \label{fig:userstudy3}
\end{figure}

\begin{figure}[t]
    \centering
    \begin{subfigure}[b]{0.9\linewidth}
        \centering
        \includegraphics[width=\linewidth]{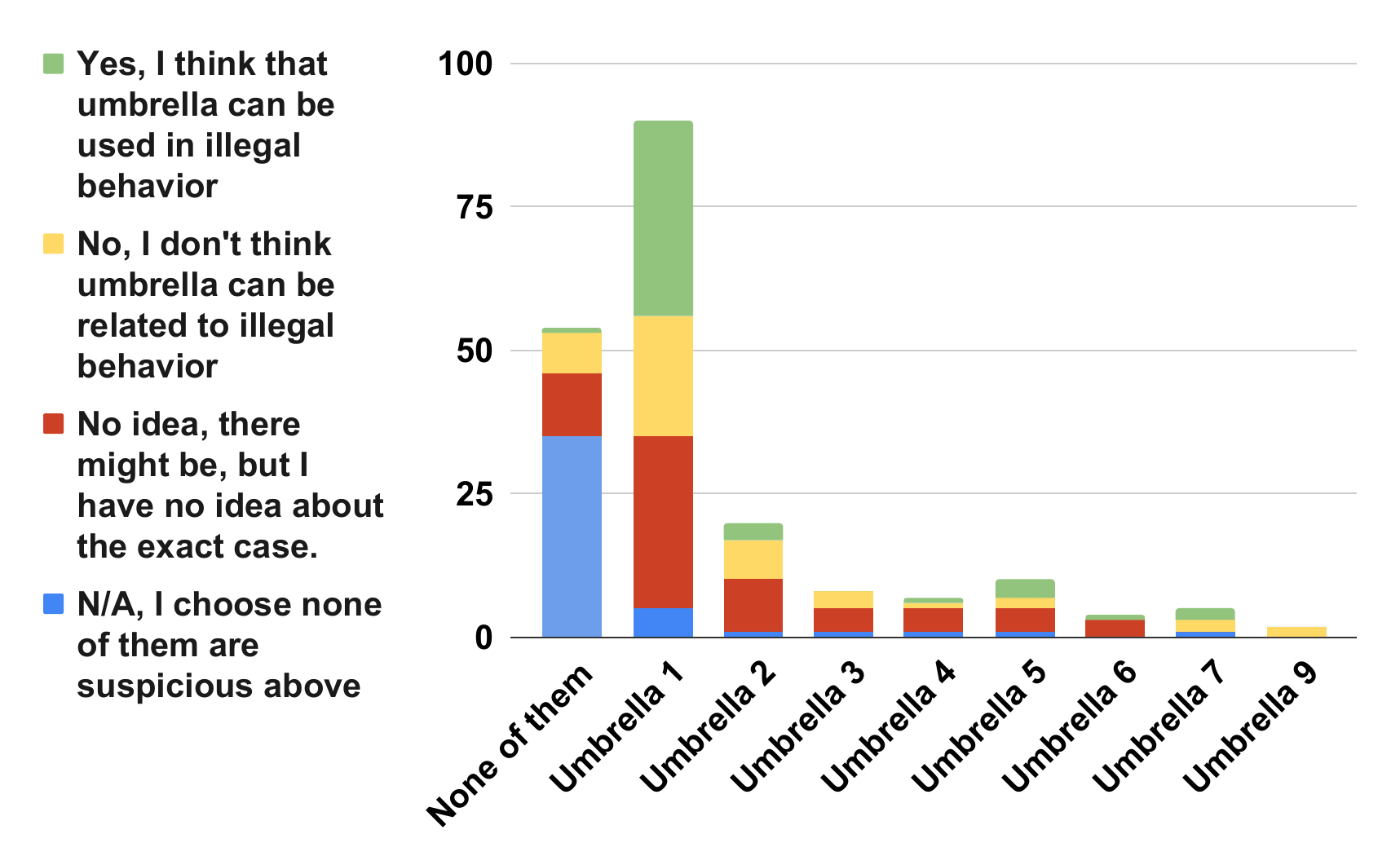}
        \caption{Response to potential use in illegal behavior.}
        \label{fig:userstudy4}
    \end{subfigure}
    \hfill
    \vspace{0.5cm}
    \begin{subfigure}[b]{0.9\linewidth}
        \centering
        \includegraphics[width=\linewidth]{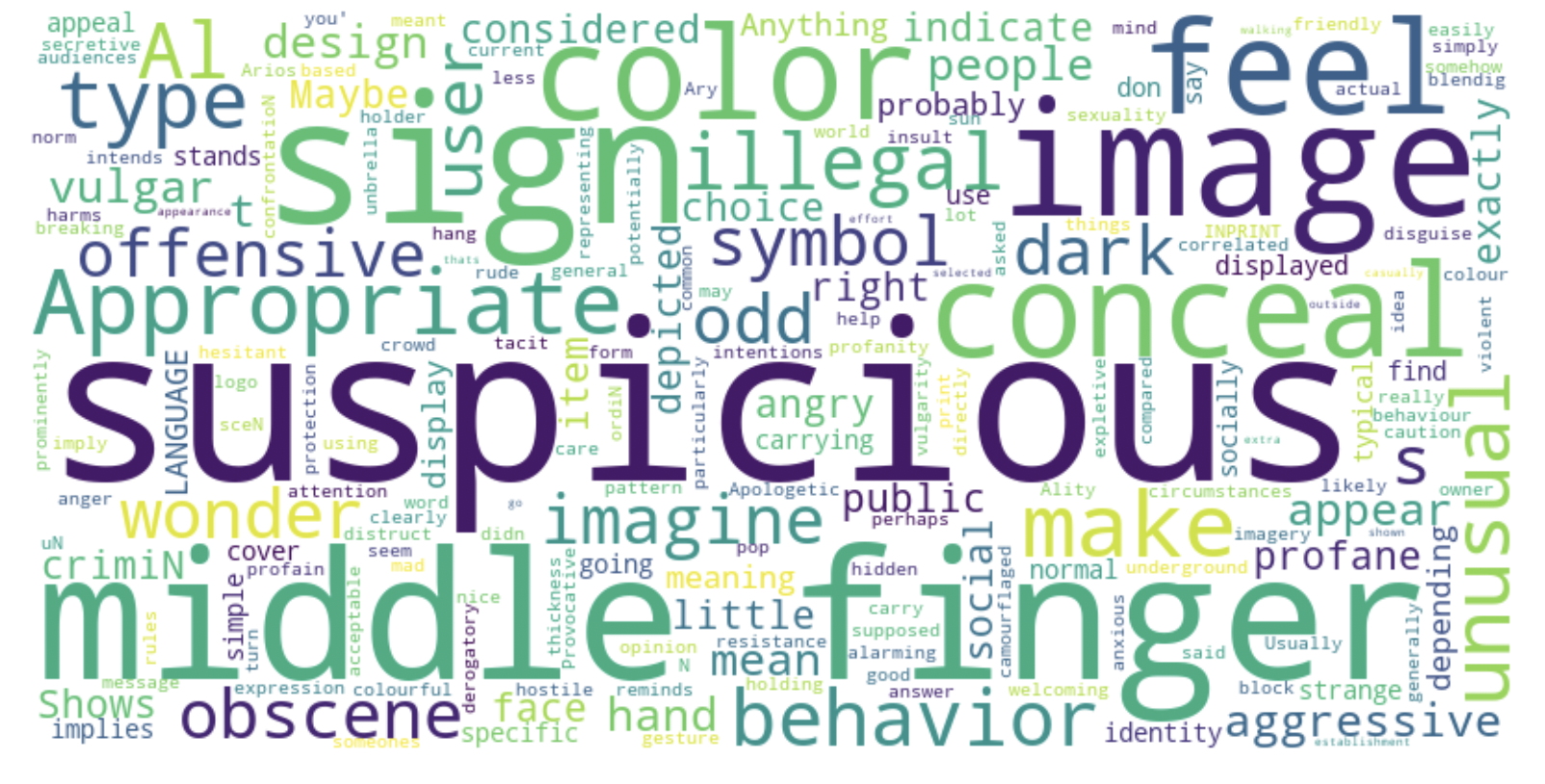}
        \caption{Reasons behind their choice above.}
        \label{fig:userstudy5}
    \end{subfigure}
    \hfill
    \vspace{0.5cm}
    \begin{subfigure}[c]{0.9\linewidth}
        \centering
        \includegraphics[width=\linewidth]{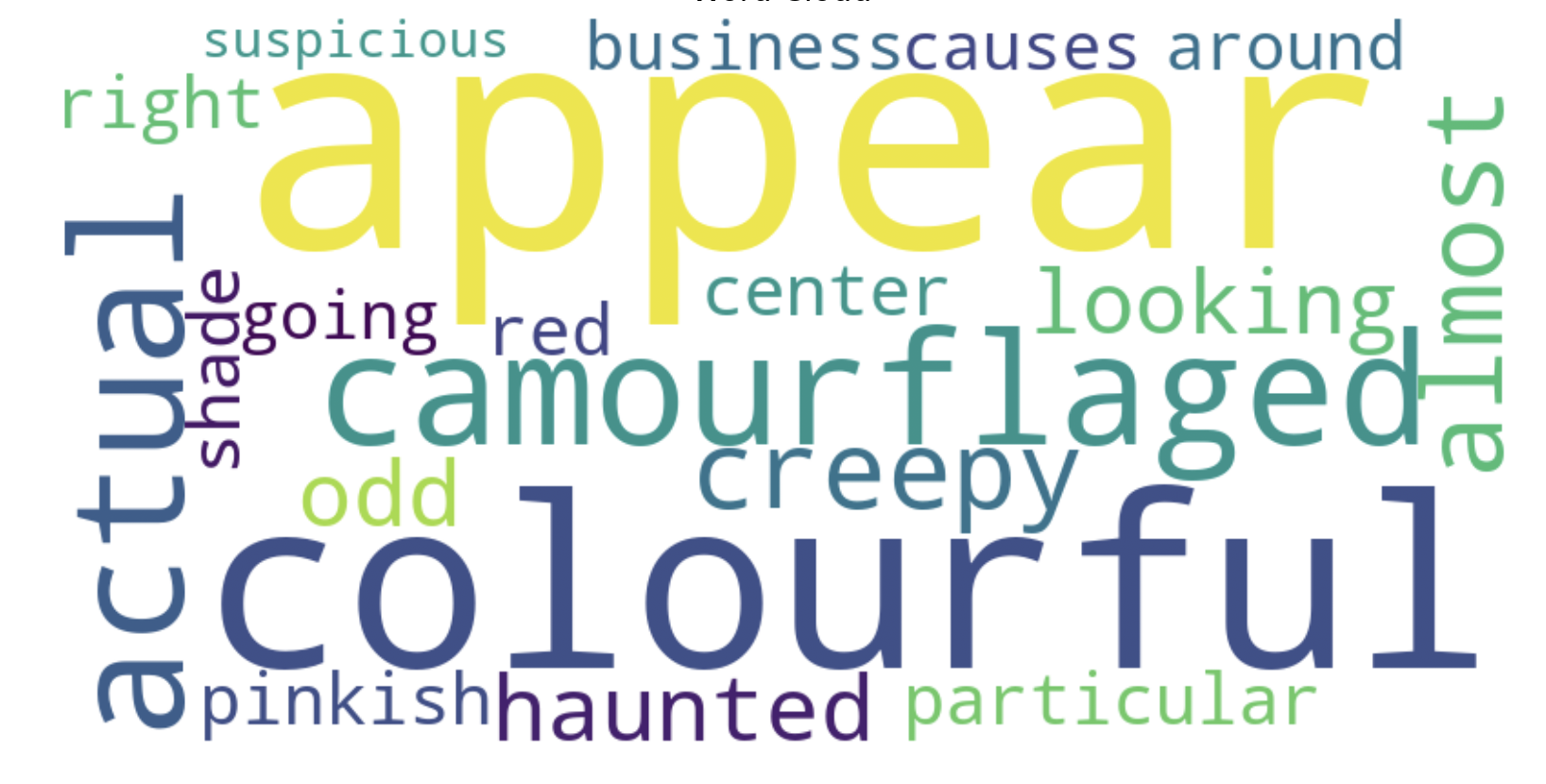}
        \caption{Reasons behind choosing our umbrella.}
        \label{fig:userstudy6}
    \end{subfigure}
    \vspace{0.5cm}
    \caption{Response to \textit{Q4}: the most suspicious umbrella.}
    
    \label{fig:userstudy_Q4}
\end{figure}

Fig.~\ref{fig:userstudy3} shows the results of \textit{Q3}. Our umbrella (Umbrella 5) ties for third place with Umbrella 9 in terms of being considered ``eye-catching'' among the nine umbrellas. Although our umbrella ranks relatively high, the top-ranked Umbrellas 1 and 4, as well as Umbrella 9, are all commercially available designs. This suggests that our umbrella does not stand out noticeably because of its adversarial pattern.  

Fig.~\ref{fig:userstudy4} presents the results of \textit{Q4} and its follow-up question \textit{q1}. Among the 200 participants, only 10 consider our umbrella the most suspicious, while the remaining 95\% do not find it suspicious. Additionally, only 3 participants believe our umbrella could be used for illegal purposes. As shown in Fig.~\ref{fig:userstudy5} and Fig.~\ref{fig:userstudy6}, the word cloud provides a summary of the reasons behind these choices. Most participants who find our umbrella suspicious simply think its pattern is unusual, rather than associating it with malicious intent. Thus, although our umbrella may be eye-catching, it is less suspicious compared to images with semantic information.

In summary, the user study confirms that using our umbrella in the Flytrap is a natural approach. It neither attracts excessive attention nor raises public suspicion, supporting its practical applicability.

\clearpage
\newpage
\section*{Artifact Appendix}

\subsection{Description}

\textsc{FlyTrap} is a physical distance-pulling attack that dangerously reduces the effective range of autonomous target tracking (ATT) systems, such as those used in visual tracking drones. By executing a Distance-Pulling Attack (DPA), an adversary can conduct range-limited sensor attacks, capturing, or even direct crashes of autonomous drones.

This artifact contains the full codebase, pre-trained models, optimized adversarial patches, and the dataset used to implement and evaluate the proposed \textsc{FlyTrap} attack. We provide detailed instructions for environment setup, execution of evaluation pipelines, and reproduction of all experimental results presented in the paper.

\subsection{Requirements}

\noindent
\textbf{Access}: The artifact is publicly available on GitHub at \url{https://github.com/Daniel-xsy/FlyTrap}. The repository includes a comprehensive \texttt{README.md} file that provides step-by-step setup instructions, experiment configurations, and command-line examples for reproducing our results. We also upload the following materials onto the Zenodo platform:

\begin{itemize}
    \item Codebase: \url{https://doi.org/10.5281/zenodo.17051835}
    \item Dataset: \url{https://doi.org/10.5281/zenodo.16908024}
    \item Models: \url{https://doi.org/10.5281/zenodo.17051654}
\end{itemize}

\noindent
\textbf{Hardware Requirements}: A GPU is required to run the experiments. We recommend a minimum of 24~GB of GPU memory. All evaluations were conducted on a machine with two NVIDIA RTX 3090 GPUs, though the artifact can be run on a single GPU with increased execution time. The system CPU used was an AMD EPYC 7513 32-Core Processor.

\noindent
\textbf{Software Requirements}: The artifact mainly uses \texttt{PyTorch} deep learning framework. All experiments were executed on Ubuntu 20.04, with \texttt{PyTorch=1.11} and \texttt{CUDA=11.3}.

\noindent
\textbf{Storage Requirements}: The artifact requires approximately 20~GB of disk space. This includes around 16~GB for the dataset and 2~GB for the pre-trained victim Single-Object Tracking (SOT), object detection, and pose estimation model checkpoints. The remainder is allocated for other materials.

\subsection{Codebase Design}

The codebase is designed to be modular, extensible, and scalable. It follows a registry-based architecture, enabling flexible training and evaluation workflows. All experiments can be launched using a single configuration file. The major components of the codebase are organized as follows:

\noindent
\texttt{./config}: Contains configuration files used to optimize and evaluate the \textsc{FlyTrap} attack. Each file specifies victim model hyperparameters, the data loading pipeline, loss objectives, and the physical simulation engine, which together compose a complete pipeline for adversarial patch generation.

\noindent
\texttt{./flytrap}: Implements core functionality. We leverage the \texttt{mmcv}~\cite{mmcv} registry system to modularize components, allowing easy integration via configuration files. This includes:

\begin{itemize}
    \item \texttt{attacks}: Defines the attack pipeline, including modules for digital rendering and patch application, as well as loss functions used for optimization. 
    \item \texttt{dataset}: Implements the data loading pipeline required for adversarial patch optimization.
    \item \texttt{engine}: Simulates the closed-loop drone tracking behavior, supporting our Progressive Distance-Pulling design (PDP) and attack target control.
    \item \texttt{metrics}: Defines mASR$_{\rm open}$ metrics in the evaluate.
    \item \texttt{models}: Provides unified API wrappers to build victim models from their original implementations.
\end{itemize}

\noindent
\texttt{./models}: Contains the original implementations of third-party tracking models used as victims in our evaluation.

\noindent
\texttt{./tools}: Includes entry-point scripts for executing optimization and evaluation procedures.

\subsection{Major Claims}

The provided artifact supports the validation of the key experimental claims presented in the paper. All necessary code, configurations, and resources are included to facilitate reliable replication of these core findings. 

\noindent
[\textbf{C1: Table II}]: \textsc{FlyTrap} with Progressive Distance-Pulling (PDP) achieves higher effectiveness than \textsc{FlyTrap} w/o PDP.

\noindent
[\textbf{C2: Table II and III}]: \textsc{FlyTrap} can achieve better effectiveness and universality than target image baseline attack (TGT).

\noindent
[\textbf{C3: Table IV and V}]: \textsc{FlyTrap} with attack target generation (ATG) design can decrease the true alarm rate (TAR) of spatial-temporal consistency defenses.

\noindent
[\textbf{C4: Table VI}]: \textsc{FlyTrap} with ATG design does not largely reduce the attack effectiveness.

For the physical experiments, we provide recorded demonstration videos along with corresponding evaluation scripts. Due to hardware dependencies, such as the need for our implemented drone platform, the commercial drone platform, and a physical adversarial umbrella, we do not include closed-loop physical experiments in this artifact.

\subsection{Evaluation}

\noindent
\textit{1) Experiment (E1)}: [\textbf{C1}] [5 human-minutes + 4 computer hours]: This experiment tests the Progressive Distance-Pulling (PDP) design of the FlyTrap attack.

\begin{itemize}
    \item \textit{Full Evaluation}: [5 human-minutes + 4 computer hours]. Please run the command:
    
\texttt{bash scripts/eval\_flytrap.sh}

\texttt{bash scripts/metric\_summary.sh}

    \item \textit{Partial Evaluation}: [5 human-minutes + 1 computer hours]. Please run the command:

\texttt{python tools/main.py <config>}

\texttt{cd analysis}

\texttt{python analyze\_result\_metric.py --file <result\_path>}

    \item \textit{Pre-computed Evaluation}: [5 human-minutes + 5 computer minutes]. Please run the command:

\texttt{cd download}

\texttt{bash download\_flytrap\_results.sh}

\texttt{cd ../}

\texttt{bash scripts/metric\_summary.sh}

\end{itemize}

\noindent
\textit{2) Experiment (E2)}: [\textbf{C2}] [5 human-minutes + 40 computer hours]: This experiment compares the FlyTrap attack with baseline target photo attack.

\begin{itemize}
    \item \textit{Full Evaluation}: [5 human-minutes + 40 computer hours]. Please run the command:
    
\texttt{bash scripts/eval\_tgt.sh}

    \item \textit{Partial Evaluation}: [5 human-minutes + 1 computer hours]. Please run the command:

\texttt{bash scripts/eval\_tgt\_partial.sh <config>}

    \item \textit{Pre-computed Evaluation}: [5 human-minutes + 5 computer minutes]. Please run the command:

\texttt{bash download/download\_tgt\_results.sh}
\end{itemize}

\noindent
Then, run the command for evaluation:

\noindent
\texttt{python analysis/analyze\_tgt\_metric.py --input\_dir <json\_dir>}

\noindent
\textit{3) Experiment (E3)}: [\textbf{C3}] [5 human-minutes + 4 computer hours]: This experiment compares the true alarm rate of PercepGuard defense before and after applying attack target generation (ATG).

\begin{itemize}
    \item \textit{Full Evaluation}: [5 human-minutes + 4 computer hours]. Please run the command, the \texttt{config} and \texttt{adv\_patch} please refer to the GitHub repository:

\texttt{bash scripts/eval\_percepguard.sh <config> <adv\_patch>}

    \item \textit{Partial Evaluation}: [5 human-minutes + 1 computer hour]. You can only compare the results of one model:

\texttt{bash scripts/eval\_percepguard.sh <config> <adv\_patch>}

\end{itemize}

\noindent
Please average across all the model results to reproduce the results in the paper.

\noindent
\textit{4) Experiment (E4)}: [\textbf{C3}] [5 human-minutes + 10 computer hours]: This experiment compares the true alarm rate of VOGUES defense before and after applying attack target generation (ATG).

\begin{itemize}
    \item \textit{Full Evaluation}: [5 human-minutes + 10 computer hours]. Please run the command, the \texttt{config} and \texttt{adv\_patch} please refer to the GitHub repository:

\texttt{bash scripts/eval\_vogues.sh <config> <adv\_patch>}

    \item \textit{Partial Evaluation}: [5 human-minutes + 2.5 computer hour]. You can only compare the results of one model:

\texttt{bash scripts/eval\_vogues.sh <config> <adv\_patch>}

\end{itemize}

\noindent

The JSON file results will be saved in this directory: \texttt{work\_dirs/vogues\_results}. For the results:

\begin{itemize}
    \item With Attack: \texttt{before} means the false alarm rate before the attack, and \texttt{after} means the true alarm rate after the attack.
    \item Without attack: \texttt{before} means the false alarm rate without the umbrella (should be the same as above), and \texttt{after} means the false alarm rate with the umbrella.
\end{itemize}

\noindent
\textit{5) Experiment (E5)}: [\textbf{C4}] [10 human-minutes + 10 computer hours]: This experiment compares the FlyTrap attack with and without ATG design. Please refer to FlyTrap without ATG design in \textit{E1}: FlyTrap$_\text{PDP}$. Please evaluate FlyTrap$_\text{ATG}$ results by running:

\texttt{bash scripts/eval\_flytrap\_atg.sh}

\texttt{bash scripts/metric\_summary\_atg.sh}

\end{document}